\documentclass[11pt]{article}
\usepackage{color}
\usepackage{hhline}
\usepackage[colorlinks=true, allcolors=blue,backref=page]{hyperref}
\usepackage[table]{xcolor}
\usepackage{amsmath, amssymb, amsthm, thmtools}
\usepackage{mathrsfs}
\usepackage{mathtools}
\usepackage[noabbrev,capitalize,nameinlink]{cleveref}
\usepackage[noadjust]{cite}
\usepackage{pifont}
\usepackage{tikz}
\usetikzlibrary{decorations.markings,shapes.geometric}
\usepackage{bbm}
\usepackage{float} 
\usepackage[T1]{fontenc}
\DeclareMathOperator*{\argmax}{arg\,max}

\usetikzlibrary{arrows.meta}

\usepackage{environ}
\usepackage{framed}
\usepackage[linesnumbered,ruled,vlined]{algorithm2e}
\usepackage[noend]{algpseudocode}
\usepackage[labelfont=bf]{caption}
\usepackage[framemethod=tikz]{mdframed}
\usepackage{appendix}
\usepackage{graphicx}
\usepackage[textsize=tiny,textwidth=2cm]{todonotes}
\usepackage{tcolorbox}
\allowdisplaybreaks[1]
\usepackage{stmaryrd}
\SetSymbolFont{stmry}{bold}{U}{stmry}{m}{n}

\usepackage{multirow}

\usepackage[margin=1in]{geometry}
\usepackage{braket}

\usepackage[shortlabels]{enumitem}
\crefformat{enumi}{#2#1#3}
\crefrangeformat{enumi}{#3#1#4 to~#5#2#6}
\crefmultiformat{enumi}{#2#1#3}
{ and~#2#1#3}{, #2#1#3}{ and~#2#1#3}

\DeclareSymbolFont{symbolsC}{U}{pxsyc}{m}{n}
\SetSymbolFont{symbolsC}{bold}{U}{pxsyc}{bx}{n}
\DeclareFontSubstitution{U}{pxsyc}{m}{n}
\DeclareMathSymbol{\medcircle}{\mathbin}{symbolsC}{7}

\crefformat{equation}{#2(#1)#3}
\crefrangeformat{equation}{#3(#1)#4 to~#5(#2)#6}
\crefmultiformat{equation}{#2(#1)#3}{ and~#2(#1)#3}{, #2(#1)#3}{ and~#2(#1)#3}
\AtBeginEnvironment{appendices}{\crefalias{section}{appendix}} 
\usepackage[color,final]{showkeys} 
\colorlet{refkey}{orange!20}
\colorlet{labelkey}{blue!30}

\numberwithin{equation}{section}
\newtheorem{theorem}{Theorem}[section]
\newtheorem{proposition}[theorem]{Proposition}
\newtheorem{lemma}[theorem]{Lemma}

\newtheorem{corollary}[theorem]{Corollary}

\newtheorem*{question*}{Question}

\theoremstyle{definition}
\newtheorem{definition}[theorem]{Definition}

\newtheorem*{definition*}{Definition}
\newtheorem*{theorem*}{Theorem}

\newtheorem{remark}[theorem]{Remark}

\newcommand{\mb}[1]{\mathbb{#1}}
\newcommand{\mc}[1]{\mathcal{#1}}
\newcommand{\mbf}[1]{\mathbf{#1}}

\newcommand{\msc}{\mathscr}

\newcommand{\on}{\operatorname}

\newcommand{\eps}{\varepsilon}
\renewcommand{\Im}{\mathsf{Im}}

\let\N\Natural

\let\R\Real
\let\Z\Integer
\let\F\Field
\let\FF\Field

\newcommand{\smcdef}{\sum_{j=0}^{\ell}c_j \alpha^j}
\newcommand{\smudef}{\sum_{j=0}^{\ell}u_j \alpha^j}

\newcommand{\K}{\mathbb{K}}
\newcommand{\U}{\mathscr{U}}

\newcommand{\W}{\mathscr{W}}
\newcommand{\ceil}[1]{\left\lceil #1\right\rceil}
\newcommand{\abs}[1]{\left|#1\right|}

\let\originalleft\left
\let\originalright\right
\renewcommand{\left}{\mathopen{}\mathclose\bgroup\originalleft}
\renewcommand{\right}{\aftergroup\egroup\originalright}

\newif\ifdraft
\draftfalse
\ifdraft
\newcommand{\mkw}[1]{\todo[color=green!20]{MW: #1}}
\newcommand{\mkwinline}[1]{\todo[inline,color=green!20]{MW: #1}}
\newcommand{\ks}[1]{\todo[color=yellow!20]{KS: #1}}

\newcommand{\ksinline}[1]{\todo[inline,color=yellow!20]{KS: #1}}

\newcommand{\rg}[1]{\todo[color=red!20]{RG: #1}}
\newcommand{\rginline}[1]{\todo[inline,color=red!20]{RG: #1}}

\newcommand{\vg}[1]{\todo[color=blue!20]{VG: #1}}
\newcommand{\vginline}[1]{\todo[inline,color=blue!20]{VG: #1}}
\usepackage{showlabels}
\else
\newcommand{\mkw}[1]{}
\newcommand{\mkwinline}[1]{}
\newcommand{\ks}[1]{}
\newcommand{\ksinline}[1]{}
\newcommand{\rg}[1]{}
\newcommand{\rginline}[1]{}
\newcommand{\vg}[1]{}
\newcommand{\vginline}[1]{}
\fi

\title{Improved LCL Equivalences and Curve Decoding Alphabet Sizes}
\begin{document}
\author{Rohan Goyal\thanks{MIT, \texttt{rohan\_g@mit.edu}. Research supported by Yael Tauman Kalai's NSF Award 6956054.}, Venkatesan Guruswami\thanks{UC Berkeley, \texttt{venkatg@berkeley.edu}. Research supported in part by a Simons Investigator Award, a UC Noyce initiative award, DOD Advanced
Research Projects Agency grant HR0011262E031, and NSF grant 2211972.}, Yihang Sun\thanks{Stanford University, \texttt{kimisun@stanford.edu}. Research supported by a NSF Graduate Research Fellowship and the Stanford Graduate Fellowship}, and Mary Wootters\thanks{Stanford University.  \texttt{marykw@stanford.edu}.  Work partially supported by NSF grant 2321489.}}
\date{}
\maketitle
\begin{abstract}
\emph{Local Coordinate-wise Linear} (LCL) properties were developed in the influential work of Levi, Mosheiff, and Shagrithaya~\cite{lms25} to capture coding-theoretic properties such as list decoding and list recovery. They showed that for every LCL property, the rate threshold for that property to be achieved coincides for for random linear codes (RLC) and random Reed-Solomon codes (RRS). Brakensiek, Chen, Dhar, and Zhang further extended this equivalence to random additive codes (linear codes over vector alphabets)~\cite{bcdz25b}. In particular, this showed that any property of all \emph{subspace design codes} is also a property of RLCs and RRS codes at the same rate.

In this work, we  
improve the probabilistic guarantees in the local equivalence theorem for RRS codes, and further extend the equivalence to all known subspace design code constructions individually. That is, each such construction, including folded Reed--Solomon codes, has the same LCL threshold as RLC and RRS codes.
These improvements can lead to applications in improving alphabet size guarantees for LCL properties---to demonstrate this, we focus on \emph{curve decoding}.

Curve decoding was introduced by \cite{gg25a,CarmonGHLLPS2026}; it is an LCL property that is sufficient to establish \emph{proximity gaps} (PG), \emph{correlated agreement} (CA), and \emph{mutual correlated agreement} (MCA).  The latter arise in cryptography and interactive proofs, for example in Interactive Oracle Proofs (IOPs) and Succinct Non-Interactive Arguments of Zero Knowledge (SNARKs).  Informally, a code $C \subset \mathbb{F}^n$ of rate $R$ is \emph{curve decodable} (with parameters $\ell, \delta, a$ and $b$) if, for any degree-$\ell$ curve $u:\F\to \mathbb{F}^n$, if $u$ is $\delta$-close to $C$ on at least $a$ points, then in fact at least $b \leq a$ of those close-by codewords must themselves lie on a degree-$\ell$ curve in $C$.

Prior work \cite{gg25a,ggsw} has established curve decodability for RLC and RRS codes, with an optimal value for $\delta$ that approaches $1 - R$ where $R$ is the rate.  However, these results also require a large alphabet size, in terms of $\ell$ and $\eta := 1-R-\delta$.  Using our improved LCL equivalence theorems, and inspired by the rank-sensitive counting framework of \cite{gg26} for list decoding, we exponentially improve the alphabet size requirement for these families while matching best known values of $(a, b)$. In particular, for $\ell=O(1)$, our curve decoding results match the best known alphabet size bounds for list decoding. This also implies improved results for PG, CA, and MCA.

\end{abstract}

\newpage
\tableofcontents

\newpage
\section{Introduction}
\label{sec:intro}

\emph{Local coordinate-wise linear} (LCL) properties 
were introduced by Levi, Mosheiff and Shagrithaya \cite{lms25} to capture properties of error correcting codes such as list-decoding and list-recovery. In their paper, the authors showed that random linear codes (RLC)\footnote{A random linear code $C \subset \F^n$ of design dimension $k$ is the image of a uniformly random linear map $\F^k\to\F^n$.} and random Reed-Solomon codes (RRS)\footnote{A random Reed--Solomon code $C\subset \F^n$ of dimension $k$ is 
$C:= \left\{ ( p(\alpha_1), p(\alpha_2), \ldots, p(\alpha_n) ) \,:\, p \in \F[X], \mathrm{deg}(p) < k \right\}$, where $(\alpha_1,\ldots,\alpha_n)$ is chosen uniformly from the ordered tuples of distinct elements of $\F$.} have the same \emph{threshold rate} with respect to any LCL property, and are therefore ``locally equivalent''.

In more detail, 
an LCL property can be represented as a family $\mc{F}$ of \emph{$r$-local profiles}.  Informally, such a family $\mathcal{F}$ describes a collection of ``bad sets'' of size $r$ which is not too large; we would like the code $C$ to not contain any bad set in this collection.\footnote{For example, a code $C \subseteq \F^n$ is \emph{$(\rho, L)$-list-decodable} if no Hamming ball of radius $\rho$ contains more than $L$ elements of $C$.  This can be phrased as $C$ avoiding all of the ``bad'' sets containing any $r = L+1$ points of $\F^n$ that lie in some radius-$\rho$ Hamming ball.  In this case, the family $\mathcal{F}$ would contain a collection of profiles that collectively describe this collection of bad sets.}    See \cref{def:local-profile} for a formal definition.

The work \cite{lms25} showed that, given an LCL property represented by $\mc{F}$, there is a threshold rate $R_{\mathcal F}$ (given with an explicit expression) such that a random linear code with rate $R>R_{\mathcal F}$ is extremely likely to contain some profile in family $\mc{F}$, while a random code with rate $R<R_{\mathcal F}$ is extremely unlikely to contain any profile in this family. 
Then they showed that RLC and RRS are locally equivalent, in the sense that this threshold rate $R_{\mathcal{F}}$ is the same for both families. 

Other work has established similar ``local equivalence'' statements with other families of codes.  
For Gallager's ensemble of random LDPC codes, an earlier local-property transfer theorem was proved in \cite{MosheiffRRSW2020}, and \cite{ggsw} establishes the corresponding threshold statement for row-span constrained LCL properties.
As another example, in a breakthrough connecting two lines of work, \cite{bcdz25b} showed that the family of \emph{subspace design codes} satisfy a (deterministic) local equivalence-type statement with the ensembles above, in the sense that every subspace design code of rate less than $R_{\mathcal{F}}$ avoids $\mathcal{F}$; and for every rate $R$ larger than $R_{\mathcal{F}}$, there {exists} a subspace design code of rate $R$ that does not avoid $\mathcal{F}$.

In addition to list-decoding and list-recovery, LCL properties were shown in \cite{gg25a,ggsw} to also capture \emph{curve decoding}. Informally, curve decoding is the following property.
Let $C \subset \F^n$ be a linear code (that is, a linear subspace of $\F^n$), for some finite field $\F$; throughout this paper, $\F$ denotes a finite field of cardinality $q$.  For integers $\ell, a, b$ and a parameter $\delta \in (0,1)$, we say that $C$ is \emph{$(\ell, \delta, a,b)$-curve-decodable} if, whenever a degree-$\ell$ curve $u(T)$ in $\F^n$ is $\delta$-close to the code $C$ on at least $a$ points, then at least $b \leq a$ of those close-by codewords actually lie on a degree-$\ell$ curve in $C$. See \cref{def:curve-decoding} for the formal definition.  

Curve decoding was introduced independently in \cite{gg25a,CarmonGHLLPS2026}; it is sufficient to establish \textit{proximity gaps} (PG), \emph{correlated agreement} (CA), and \emph{mutual correlated agreement} (MCA).  These latter notions arise in cryptography and interactive proofs, for example in Interactive
Oracle Proofs (IOPs) and Succinct Non-interactive Arguments of Zero Knowledge (SNARKs). See, e.g., \cite{proximityprize} for more on these applications.

As it can be viewed as an LCL property, it follows that curve decodability has the same threshold rate for random linear codes, random RS codes, and random LDPC codes.\footnote{Here and throughout the paper, ``random LDPC codes'' refers to \emph{Gallager's ensemble}.}
Using this framework, prior work~\cite{gg25a,ggsw} established curve decoding for these families up to the optimal value of $\delta$, namely $\delta$ approaching $1 - R$, where $R = \dim_{\F}(C)/n$ is the rate of the code $C \subset \F^n$.  However, other parameters---such as the alphabet size---remained sub-optimal.  This affected the parameters for proximity gaps, correlated agreement, and mutual correlated agreement.

\paragraph{Our Contributions.}
Conceptually, we have three main  contributions.
\begin{itemize}
    \item 
    First, we improve the local equivalence result between RLC and RRS from \cite{lms25}.  In more detail, for the below-threshold non-containment
    direction (that is, the direction ``if $R < R_{\mathcal{F}}$, then a random RS code likely avoids $\mathcal{F}$''), we improve the quantitative guarantees.  The quantitative improvement is discussed in \cref{rem:improve-lms}.
    \item Second, we show for nested code families how to deduce above-threshold containment statements directly from below-threshold non-containment statements concerning a certain class of critical quotient profiles. For random RS codes the hypothesis holds with high probability, while optimal nested subspace design codes satisfy it deterministically. This gives local equivalence for each such nested constructions---including folded Reed--Solomon code---in a black-box manner in contrast with prior approaches in e.g. \cite{lms25}\footnote{In fact, we repair the proof of the above-threshold containment statement in \cite{lms25} for Random Reed--Solomon codes---see \cref{rem:repair-lms}}; see \cref{rem:why-blackbox}.
    \item   
    Finally, we use our LCL equivalance theorems to study curve decoding. Along with more refined bounds on the size of the relevant LCL family, we improve the alphabet size bounds on curve decoding---and therefore proximity gaps and (mutual) correlated agreement---for RRS, RLC, and RLDPC. The quantitative improvement is given in \cref{tab:CD}.
\end{itemize}

We discuss each of these directions in more detail below.

\subsection{Main Results}
\label{subsec:results}
To state our main results, we informally define LCL properties, potential functions, and thresholds. See \cref{subsec:lcl} for formal definitions.
\begin{definition}[The LCL Framework; Briefly]
\label{def:lcl-informal}
A local coordinatewise linear family $\mc{F}$ is a collection of \emph{$r$-local profile} $(\mbf V,U)$, where
$\mbf V=(V_1,\ldots,V_n)$ is a tuple of subspaces of $\F^r$, and
$U\subset\F^r$.  A code $C\subset\F^n$ \emph{contains}
$(\mbf V,U)$ if there is a matrix $M\in\F^{n\times r}$ whose columns
belong to $C$, whose $i$th row belongs to $V_i$, and whose row span isexactly $U$. \footnote{For $\F$-additive codes, e.g. folded Reed-Solomon codes, the definition is generalized by \cite{bcdz25a}---see \cref{def:folded-profile}.}
For $R\in[0,1]$, define the \emph{potential}
\begin{equation}
  \phi_{\mbf V}(U,R)
  :=
  (R-1)\dim_{\F}U
  +\frac{1}{n}\sum_{i=1}^n\dim_{\F}(U\cap V_i).
\end{equation}
For subspace $W\subsetneq U$, let \emph{threshold} $R_{\mbf V}(U,W)$ be the value $R$ where
$\phi_{\mbf V}(U,R)=\phi_{\mbf V}(W,R)$. Let
\begin{equation}
  R_{\mbf V}(U)
  :=
  \max_{W\subsetneq U}R_{\mbf V}(U,W),
  \qquad
  R_{\mc F}
  :=
  \min_{(\mbf V,U)\in\mc F}R_{\mbf V}(U).
\end{equation}
Thus $R_{\mbf V}(U)$ is the rate above which the full row span $U$
has at least as large a potential as each of its proper subspaces, and
$R_{\mc F}$ is the threshold associated with the family $\mc F$ of profiles $(\mbf{V}, U)$
\end{definition}

\paragraph{Improved local equivalence.} We begin by stating our main LCL equivalence theorem for RRS.
Recall that $R_{\mathcal{F}}$ is the threshold for RLC.  \cite[Theorem 3.1]{lms25} shows that this is also the threshold for RRS.  We improve their statement to the following.

\begin{restatable}[Local Equivalence for RRS]{theorem}{rrsequiv}
\label{cor:rrs-equiv}
Let $\mc{F}$ be row-span constrained $r$-local family, let $\eta\in (0, 1]$, and let $C:=\mathsf{RS}(\boldsymbol{\alpha}, k)$ of rate $R=k/n$. Suppose $q:=|\F|>\max\{n, r(k-1)\}$, then
\begin{enumerate}
\item If $R\le R_{\mc F}-\eta$ and $\eta n\ge2r$, then
\begin{equation}
  \mb{P}\left(C\text{ contains }\mc F\right)
  \le
  |\mc F|\,r
  \left(
    \frac{2e r k}{\eta(q-n)}
  \right)^{\eta n/2}.
\end{equation}

\item If $R>R_{\mc F}+\eta$ and $\eta n\ge2r(r+1)$, then
\begin{equation}
  \mb{P}\left(C\text{ contains }\mc F\right)
  \ge 1-q^{2r^2}r
  \left(
    \frac{2e r k}
         {(\eta/r-2/n)(q-n)}
  \right)^{\frac{\eta n}{2r}-1}.
\end{equation}
\end{enumerate}
\end{restatable}
\begin{remark}[Comparison to \cite{lms25}; Statement 1]
\label{rem:improve-lms}
Our main quantitative improvement over \cite{lms25} is in the first (non-containment) statement
 (1). 
For this statement,
\cite[Theorem 3.1]{lms25} replaces the factor \[ r\left(\frac{2e r k}{\eta(|\F|-n)}\right)^{\eta n/2}  \qquad \text{with} \qquad
(2^r - 1)\left(\frac{(4r)^{4r} k }{\eta(|\F|-n)}\right)^{\eta n /2r}.\]
We observe two quantitative improvements (see \cref{prop:lms-fixed-pair} for comparison):
\begin{enumerate}
\item Our exponent is $\Omega(\eta n)$, improving from
$\Omega(\eta n/r)$. This leads to improved alphabet size bounds for curve decoding.
\item The numerator is $O(rk)$, improving from $r^{O(r)}k$. This is a conceptual improvement, but it does not improve alphabet size bounds for any properties we study.
\end{enumerate}
We further remark that, because list decoding and list recovery are 
LCL properties, this sharper non-containment statement (and Point 1 in particular) can also improve the alphabet size bounds for these properties for RRS. In particular, it recovers the state-of-the-art
RRS list-decoding alphabet size and improves the corresponding RRS
list-recovery bound. See \cref{rem:apply-to-ld-lr} for details.
\end{remark}

\begin{remark}[Comparison to \cite{lms25}; Statement 2]
\label{rem:repair-lms}
The containment direction (2) in \cref{cor:rrs-equiv} gives the same quantitative result as in \cite{lms25}, but unfortunately there was an issue with the proof in \cite{lms25} in this direction, which has been confirmed with the authors.\footnote{In more detail, at a key point in this direction of the proof, there is the implication ``If Eq. (25) applies to $V$' then it also applies to $V$ since, by linearity of $C$, containing $V$ implies containing $V'$.''  Unfortunately, the implication is backwards here: what would be needed is that containing $V'$ implies containing $V$, but this does not follow from the linearity of $C$.} 
There may be a simpler fix, but instead we give a more general proof.  This more general proof enables us to obtain a stronger containment statement for subspace design codes, discussed next. 
\end{remark}

\paragraph{Deterministic black-box implication from low-rate non-containment to
high-rate containment for nested code families.}
Our second LCL equivalence theorem concerns a \emph{nested family} of subspace-design codes.  Informally, a \emph{$(\gamma,r)$-optimal nested subspace-design code} is a chain $C_1\subset\cdots\subset C_k$ in which every $C_j$, whose rate is $j/(sn)$, is within additive error $\gamma$ of the optimal
$r$-dimensional subspace-design bound (see \cref{def:sdc,def:slacked-sdc-nesting}).
Thus the same construction is near-optimal simultaneously at every intermediate rate.

\begin{restatable}[Local Equivalence for SDC]{theorem}{sdcequiv}
\label{cor:sdc-equiv}
Let $\mc{F}$ be a row-span constrained $r$-local family and let
$(C_1,\ldots,C_k)$ be a $(\gamma, r)$-optimal nested subspace design code. Define rate and slack to be
\begin{equation}
  R:=\frac{k}{sn}\qquad\text{and}\quad \eta := r\gamma+\frac{2r}{sn}.
\end{equation}
Then the following hold.
\begin{enumerate}
\item If $R< R_{\mc F}-\gamma$, then $C_k$ avoids $\mc F$.
\item If $R> R_{\mc F}+\eta$, then $C_k$ contains $\mc F$.
\end{enumerate}
\end{restatable}
The factor of $r$ difference in the slack from containment to non-containment ($\gamma$ versus $\eta$) is analogous to the extra factor of $r$ both the base and the exponent of the probability bounds in \cref{cor:rrs-equiv}. Of course, subspace design codes are deterministic, so there is no probability here.

\begin{remark}[Comparison to \cite{bcdz25a}]
\label{rem:why-blackbox}
In contrast to \cite{bcdz25a,ggsw} (see \cref{thm:sdc-threshold-pair}), this theorem applies to \emph{any} optimal nested subspace design codes---which includes all known examples of subspace design codes---rather than the existence of one. 
That is, the result of \cite{bcdz25a} says that if $R > R_{\mathcal{F}} + \eta'$ (for some small $\eta'$) then there \emph{exists} a subspace design code that does not avoid $\mathcal{F}$.  In contrast, \cref{cor:sdc-equiv} applies to any (nested) subspace design code.

In \cref{thm:frs-equiv,cor:frs-equiv}, we instantiate this for folded Reed-Solomon codes (see \cref{def: FRS Codes}) with its natural nesting structure. 

In particular, this means that any future results specific to a single family of nested subspace design codes can be transferred to all locally equivalent codes.  

For example, in \cite{bcdz25b}, the authors showed that \emph{all} subspace design codes satisfied certain list-recovery properties, and this implied list-recovery for RLC and RRS.  In \cite{gg25a}, the authors showed that \emph{all} subspace design codes satisfied curve-decodability, and (along with \cite{ggsw}), this implied curve-decodability for RLC, RRS, and RLDPC.

However, one could imagine being able to prove similar results about, say, Folded Reed-Solomon codes only, using techniques that did not natively apply to any subspace design code.  With prior results, this would not transfer to locally equivalent ensembles.  Given \cref{cor:sdc-equiv}, we now know that such a result would indeed transfer.

\end{remark}

In fact, the result that we prove is more general; it applies to any family of \emph{$(\mathbf{V}, \gamma)$-good nested codes}---see \Cref{def:goodness} and \Cref{thm:key}.

\paragraph{Applications to Curve Decoding.}
Lastly, we discuss how our LCL equivalence theorems allows us to get stronger alphabet size bounds for curve decoding. Before discussing curve decoding, we make one remark about applications to list decoding and list recovery.

\begin{remark}[Potential applications to list-decoding and list-recovery]
\label{rem:apply-to-ld-lr}
While we do not focus on list decoding or list recovery in this work, we believe that our techniques will also work to bound the alphabet size for optimal list-decoding of RRS codes achieving the $\eta$-relaxed generalized Singleton bound with alphabet size $n\cdot \exp(1/\eta^2)$, matching the alphabet size achieved by \cite{AlrabiahGGLZ2025} in a more generic manner, and which the techniques of \cite{lms25} do not directly achieve. For list-recovery, we believe that our techniques can achieve an alphabet size of $n\cdot \exp(\ell^2/\eta^3 \cdot \log(\ell/\eta))$, exponentially improving the previous best doubly-exponential bound of $n\cdot \exp(\exp(\log \ell (R+\eta)/\eta))$ \cite{bcdz25a}.
\end{remark}

We first state the formal definition of curve-decoding.  Then we will state our results that improve alphabet size bounds for curve decoding, and therefore proximity gaps and (mutual) correlated agreement, for RLC and RRS. The analogous RLDPC consequence is stated in \cref{thm:rldpc-curve-final,cor:rldpc-proximity-package}.

\begin{definition}[Curve Decoding]
    \label{def:curve-decoding}
    Let $C \subset \F^n$ be a linear code.  Fix a positive integer $\ell$, and fix positive integers $a \geq b$.  Let $\delta \in (0,1)$.  We say that $C$ is $(\ell, \delta, a, b)$-curve-decodable if, for all $u_0, u_1, \ldots, u_\ell \in \F^n$ and every function $f: \F \to C$, whenever
    \[ A = \left\{\alpha \in \F \,:\, \Delta\left( \sum_{j=0}^\ell u_j \alpha^j, f(\alpha)\right) \leq \delta \right\}\]
    has $|A| \geq a$, then there exists $c_0, c_1, \ldots, c_\ell \in C$ so that 
    \[ B = \left\{ \alpha \in A \,:\, f(\alpha) = \sum_{j=0}^\ell c_j \alpha^j \right\}\]
    has $|B| \geq b$.
\end{definition}

\begin{theorem}[Curve Decoding with Smaller Alphabets, Informal; see
\cref{thm:rlc-curve-final,thm:rrs-curve-dec}]
\label{thm:curve-decoding-informal}
Fix $\ell\ge1$ and $\eta, \delta>0$. Let  rate $0<R\le 1-\delta-2\eta$, let $ a=O({\ell^2}/{\eta^3})$
and $ b=O({\ell}/{\eta})$. For all sufficiently large $n$, the following hold with
high probability.
\begin{enumerate}
\item A random linear code of rate $R$ is
$(\ell,\delta,a,b)$-curve-decodable provided
\begin{equation}
  q\ge
  \exp\left(
    O\left(
      \frac{\log(1/\eta)+\ell}{\eta}
    \right)
  \right).
\end{equation}
\item A random Reed--Solomon code of rate $R$ is
$(\ell,\delta,a,b)$-curve-decodable provided
\begin{equation}
  q\ge
  n\exp\left(
    O\left(
      \frac{\ell}{\eta^2}
      +\frac{\ell\log(\ell+1)}{\eta}
    \right)
  \right).
\end{equation}
\end{enumerate}
Moreover, under the same hypotheses, both codes have
$(\ell,\delta,\eps)$ mutual correlated agreement, and threfore perfect correlated
agreement and proximity gap (with slack $\gamma=0$), where
\begin{equation}
\label{eq:analog-a}
  \eps q
  =
  O\left(\frac{n\ell(1-R)}{\eta}+\frac{\ell^2}{\eta^3}
  \right).
\end{equation}
\end{theorem}
Before we compare our results to prior work, we explain in \Cref{rem:counting-conceptual} how our previous results on LCL equivalence theorems feed in to \Cref{thm:curve-decoding-informal}, and the additional technical and conceptual contribution that is required on top of it.

\begin{remark}[Conceptual contribution on top of our LCL results]
\label{rem:counting-conceptual}
With the LCL transfer theorems in place, it remains to count $|\mc{F}|$ to union bound more efficiently in the non-containment direction of the theorems. To do so, we implement a profile canonicalization method\footnote{\cite{gg26} canonicalized list-decoding profiles through the $k$-weakly partition hypergraphs of \cite{AlrabiahGGLZ2025}, and for list-recovery through the Brascamp-Lieb inequalities of \cite{bcdz25a}. The list-recovery canonicalization is not exact as no single combinatorial method is currently known for capturing list-recovery.} through the subspace-design code pruning techniques of \cite{Sri25, AHS26}. The application of pruning to curve decoding was first observed in \cite{gg25a}, and here we implement it directly at the level of local profiles.

In particular, we prove curve-decoding for the random ensembles of codes directly at the level of local profiles, without appealing to properties of subspace design codes, or to the results of \cite{gg25a} on curve decoding for subspace design codes. We believe this is of particular interest, and more firmly shows that previous arguments that were crafted for and seemed specific to subspace-design codes can directly be implemented at the level of local profiles. This profile-level viewpoint is also what lets combine the probability bounds with the refined counting in our final union bounds.
\end{remark}
\renewcommand{\arraystretch}{1.5}
\begin{table}[H]\centering
\begin{tabular}{|c|c|c|c|}
\hline
    & Source & $a$ & $|\F|$ \\
    \hline\hline
    \multirow{3}{*}{RLC}
        & \cite{gg25a}&  \cellcolor{red!10} $(\ell/\eta)^{O(\ell)}$ & \cellcolor{red!10}
        $\exp\left( (\ell/\eta)^{O(\ell)} \right)$\\
    \hhline{~|-|-|-}
        &\cite{ggsw} & \cellcolor{green!10}
        $O(\ell^2/\eta^3)$& \cellcolor{yellow!10} $\exp(\ell^2/\eta^4)$\\
    \hhline{~|-|-|-}
        &\cite{YZ26} & \cellcolor{red!10}
        $(1/\eta)^{O(\ell/\eta)}$& \cellcolor{yellow!10} $\exp\left(\frac{\ell \log (1/\eta)}{\eta}\right)$\\
    \hhline{~|-|-|-}
        &This work &\cellcolor{green!10}$O(\ell^2/\eta^3)$ & \cellcolor{green!10} $\exp\left(\frac{\ell+\log(1/\eta)}{\eta}\right)$ \\
    \hline\hline
    \multirow{3}{*}{RRS}
        &\cite{gg25a} & \cellcolor{red!10} 
        $(\ell/\eta)^{O(\ell)}$
    &
        \cellcolor{red!10} $n \cdot \exp\left( (\ell/\eta)^{O(\ell)} \right)$
        \\
    \hhline{~|-|-|-}
        &\cite{ggsw} & \cellcolor{green!10} $O(\ell^2/\eta^3)$
        & \cellcolor{yellow!10}
        $n \cdot \exp(\ell^4/\eta^7)$ 
      \\
          \hhline{~|-|-|-}
        &\cite{YZ26} & \cellcolor{red!10} 
        $\left( 1/\eta \right)^{O(\ell/\eta)}$
        & \cellcolor{yellow!10}
        $n \cdot \exp(\ell^2/\eta^3 )$
        \\
    \hhline{~|-|-|-}
        &This work & \cellcolor{green!10} $O(\ell^2/\eta^3)$& \cellcolor{green!10} $n \cdot\exp\left(\frac{\ell}{\eta^2}+\frac{\ell\log \ell}{\eta}\right)$\\
    \hline
\end{tabular}
\caption{Results for curve decodability of random linear codes (RLC) and random Reed--Solomon codes (RRS), with high probability.  We report results for codes of constant rate $R \in (0,1)$ and sufficiently large length $n$ that are $(\ell, \delta, a, b)$-curve-decodable for $\delta =1-R-\eta$ and $b = O(\ell/\eta)$.  This is the value of $b$ relevant for obtaining results for proximity gaps, CA, and MCA. 
The boxes colored green have the best known dependence on $\ell$ and $\eta$.  The boxes colored yellow and red have somewhat worse and significantly worse dependence, respectively.}\label{tab:CD}
\end{table}

We end this section with a comparison of curve decoding alphabet size bounds against prior work, shown in \Cref{tab:CD}. The brief history shown in \cref{tab:CD} can be summarized as follows.  The work \cite{gg25a} was the first to establish $(\ell, \delta, a,b)$-curve-decoding for codes of rate $R$ with $\delta = 1 - R - \eta$; this is optimal as $\eta \to 0$.  However, that work did not result in a good value for the parameter $a$ when the curve degree $\ell$ is large.  
It also required the field size to be quite large, doubly-exponential in $\ell$.  The value of $a$ was ameliorated in \cite{ggsw}; this work also substantially improved the requirement on the alphabet size, but still left room for improvement.  The work \cite{YZ26} further improved the alphabet size, but at the cost of a worse guarantee on $a$.

We note that when the degree of the curve being considered is constant---that is, if $\ell = O(1)$, such as in the case of line-decoding---we achieve alphabet sizes matching the best known upper bounds for the corresponding family even for list-decoding. A similar alphabet size was achieved by \cite{YZ26} for curve-decoding but with a worse dependence on the parameter $a$.\footnote{As written, \cite{YZ26} does not directly establish curve-decoding---they focus on PG and CA---but their methods for correlated agreement do directly achieve curve-decoding as well with the parameters listed in our table. Lifting these bounds to proximity gaps produces the same bounds as claimed in \cite{YZ26}.} This gap is due to using a pruning algorithm based on \cite{KoppartyRSW2023} instead of the improvements of \cite{Sri25, AHS26}.

Our work attains the best of both worlds: we obtain a good alphabet size---improving over \cite{YZ26} for large $\ell$---and we match the value of $a$ from \cite{ggsw}.

\subsection{Our Techniques}\label{sec:overview}
Our proof contains three main ingredients: an improved analysis of the
low-rate RRS process of \cite{lms25} in \cref{sec:rrs-improvement}, a black-box argument that converts
non-containment into containment for nested codes in~\cref{sec:black-box}, and a compressed
count of the local profiles witnessing failure of curve decoding in~\cref{sec:profile-counting}. We put these together in \cref{sec:alphabet} to get alphabet size bounds for curve decoding. We briefly discuss our techniques one-by-one, corresponding to the three ingredients and sections.

\paragraph{Improved Non-containment for RRS.}
We first prove non-containment statement in \cref{cor:rrs-equiv}. The proof follows a similar outline to that in \cite{lms25}. The key result (\Cref{thm:low-lms}) establishes that the probability that a particular profile $(\mathbf{V}, U)$ is contained in a random RS code is not too large, with improved parameters relative to  \cite{lms25} (see \cref{prop:lms-fixed-pair} for their analog).

We analyze a random process where we expose each evaluation point one at a time; this same process was also analyzed in \cite{lms25}. At each step of the process, we expose an evaluation point $\alpha_t$, and we consider the space $S_t$ of potential bad sets of codewords (corresponding to low-degree polynomials $f_1, \ldots, f_r$ for which  $(f_1(\alpha_t), \ldots, f_r(\alpha_t))$ satisfies constraint $t$.
As in \cite{lms25}, we define a dynamic potential that counts the dimension in our current candidate set $S_t$, minus the anticipated number of constraints the future evaluation point reveals would impose.
This dynamical potential starts at $-\eta n$ due to the low-rate hypothesis; at each step, it can increase only when an unlikely \emph{defect event} occurs, and by at most one. Therefore, after all evaluation points have been revealed, with high probability it remains negative, which implies $S_n$ has dimension zero.  In this case, there are no ``bad'' sets for the profile contained in the random RS code. This is the same outline in \cite{lms25} that we describe in the language of our paper.

The main insights for our two improvements discussed in \cref{rem:improve-lms} is somewhat technical.  
The first insight pertains to the handling of a reduction in dimension $d_t$ of the current solution space, which is always at most $r$. The accounting in \cite{lms25} separates the process into $r$ levels corresponding to this dimension, and distributes the $\eta n$ initial potential gap equally among them. The process is successful if, at every level, the potential slack is not consumed by the defect events, so they can tolerate $\eta n/r$ many events. In contrast, we handle reduction in $d_t$ dynamically. We track the sum of the (normalized) dynamic potential and $d_t$ as our monovariant, so initially it is at most $-\eta n+r$. This enables us to tolerate $\eta n-r=\Omega(\eta n)$ many defect events, which is the source of our improvement to the exponent (i.e. first bullet point of \cref{rem:improve-lms}).

Our second improvement in \cref{rem:improve-lms} is an improved bound on the defect event probability. Roughly, the defect is some the event where we anticipate an evaluation point to impose $c$ constraints by counting dimensions over $\K=\mb{F}(X)$, but it imposes fewer than $c'$ constraints once we specify evaluation point $X=\alpha_t$. This corresponds to the vanishing of a generically non-zero determinant minor upon specialization, and is therefore controlled by the degree of the polynomials. Our insight is a better \emph{presentation} for these polynomials without degree below-up to $r^{r}$ in \cite{lms25}  as dimension $d_t$ is reduced. The crucial insight is to remember the ``numerators and denominators'' of quotient steps we have taken in the past, to maintain a good presentation.
This gives the improved defect probability bounds.

\paragraph{From Non-containment to Containment.}
We begin by explaining how to reduce containment to proving upper bounds on the following quantity: for each row span $U$, let $X(U)$ be the number matrices satisfying $\mbf{V}$ and has rows in $U$ and columns in $C$. 
In contrast, the containment in \cref{def:lcl-informal} where we require row span \emph{exactly} $U$ not inside $U$. To bridge this gap, we use an idea of \cite{lms25} for RLCs: suppose $C$ does not contain any witness matrix satisfying $\mbf{V}$ and with row span exactly $U$, then
\begin{equation}
\label{eq:key-count-overview}
X(U)\le \sum_{W\subsetneq U}X(W).
\end{equation}
Lower bounding $X(U)$ in terms of $\phi_{\mbf{V}}(U, R)$ usually follows from constraint counting. An upper bound on $X(W)$ is similar when $C$ is an RLC, but this is the entire difficulty for RRS and arbitrary subspace design code. We motivate our approach first with RRS codes.

Observe that upper bounding $X(W)$ is has natural non-containment flavor instead, which motivates our
reduction from containment to non-containment. The first step is also suggested by
\cite{lms25}. Fix any $W\subset U$, choose proper
subspace $W_1\subsetneq W$ attaining the maximum in $R_{\mbf V}(W)$. If we quotient by $W_1$ and we consider an RRS code with rate just below $R_{\mbf{V}}(W)$, the resulting quotient profile has negative potential, so our non-containment direction of \cref{cor:rrs-equiv} makes it unlikely to occur. Climbing back to the original rate  incurs a controlled dimension
cost, so we may upper bound $X(W)$ in terms of $\phi_\mbf{V}(W, R)$. Since $R>R_{\mbf{V}}(U)$, $\phi_\mbf{V}(W, R)<\phi_\mbf{V}(U, R)$, which gives a contradiction to \cref{eq:key-count-overview} upon checking appropriate rate slacks.

The quotient handles profiles that remain in $W/W_1$, but it annihilates the candidates whose row
span is already contained in $W_1$.  Our first new idea is to handle
these candidates recursively: inside $W_1$, we choose its own critical proper subspace $W_2$ and repeat.  This produces a \emph{local chain} of
critical subspaces $W_j$ and a decreasing chain of local thresholds rates; see
\cref{def:local-chains,prop:decreasing-rate}.  The chain specifies both which quotient $W_j$ to take and what rate $R_j$ RRS code we descend to.  When we climb back to the original rate, the
resulting dimension costs telescope to the potential $\phi_{\mbf{V}}(W, R)$, contradicting \cref{eq:key-count-overview}.  This gives the non-containment statement in \cref{cor:rrs-equiv}.

Our second insight is that the lower-rate non-containment conditions
used by this argument is a deterministic property of a nested code
family $C_1\subset \dots \subset C_k\subset(\F^s)^n$ where $\dim C_j=j$. D
We call this property $(\mbf V,\gamma)$-\emph{goodness} (see
\cref{def:goodness}). Dscending to $C_j$ is exactly the analog of ``lower RRS to rate $j/sn$'', and the natural nesting of RRS codes is good with
high probability.  For subspace-design codes, the
subspace-design guarantee certifies goodness deterministically with
the same slack $\gamma$.  The argument therefore applies to every
$(\gamma,r)$-optimal nested subspace design code, proving the LCL equivalence theorems for these codes
(see \cref{thm:sdc-good,thm:sdc-equiv}), rather than merely establishing the
existence of one suitable code (see \cref{rem:why-blackbox}).  Finally,
the bounds of \cite{GK16} certify goodness for the natural nesting of
folded Reed--Solomon codes, giving our instantiation in \cref{thm:frs-equiv}.

\paragraph{Recap of the Curve Decoding Profile Family.}
Recall row-span constrained LCL properties $\mc{F}$ from \cref{def:lcl-informal}.
It was shown in \cite{ggsw} that curve decodability can be seen as such a property via some profile family $\mathcal{F}=\mc{F}_{\ell, \delta, a, b}$. We briefly recap this construction.

Consider a ``bad set'' $X \subset \F^n$ for $(\ell,\delta, a,b)$-curve decodability.  Such a set has size $a$, and is indexed by the set $A = \{ \alpha \in \F \,:\, \Delta(\sum_{j=0}^\ell u_j \alpha^j, f(\alpha) ) \leq \delta \}$ as in \Cref{def:curve-decoding}.  The $\alpha$th element of $X$ (for $\alpha \in A)$ represents the codeword $f(\alpha) \in C$ that is $\delta$-close to $u(\alpha)$, for some degree-$\ell$ curve $u:\F \to \F^n$.
For $X$ to be bad, it should be the case that (1) for all $\alpha \in A$, the $\alpha$th element of $X$, call it $X[\alpha]$, is $\delta$-close to $u(\alpha)$, for some degree-$\ell$ curve $u$; but (2) there is no set of $b \leq a$ elements of $X$ that themselves lie on a degree-$\ell$ curve.  Below, we explain how to capture this property as a collection of LCL profiles $(\mathbf{V}, U)$.  Item (1) is captured by $\mathbf{V}$, and Item (2) is captured by $U$.

To see how to capture Item (1), 
 let $A_i \subset A$ be the set of $\alpha$ so that $X[\alpha]$ and $u(\alpha)$ agree in the $i$th coordinate. Then consider 
\[ V_i = K_A(A_i) := \left\{ y \in \F^A : y|_{A_i}\text{ agrees with a polynomial of degree at most }\ell\right\}. \]
Recalling that $M_X$ is the matrix with the elements of $X$ as columns, it is not hard to see that if $X$ satisfies (1) with agreement sets $A_i$ as above, then the $i$th row of $M_X$ lies in $V_i$ for each $i$.

Item (2) is captured through the row-span constraint $U$.  We say that $U$ is \emph{$(\ell,A,b)$-curve-free} if, for all sets $B \subset A$ of size $b$, $U|_B$ is not contained in $K_B(B) \subset \F^B$.  Then it can be shown that for any $X$ that satisfies (2), $M_X$ has a row span $U$ that is $(\ell,A,b)$-curve-free.

Thus, the family $\mathcal{F}$ that captures curve decodability consists of all $(\mathbf{V}, U)$ so that $\mathbf{V}$ looks like the above (for some choice of $A_1, \ldots, A_n$), and so that $U$ is $(\ell,A,b)$-curve-free.  To bound the size of $\mathcal{F}$, the work \cite{ggsw} counted the number of such $\mathbf{V}$ and such $U$.  Their count was dominated by $2^{an}$, the number of ways to choose the sets $A_i$.
\paragraph{Compressing the Curve Decoding Profile Family.}
Our main improvement in this direction comes by improving the bound on the size of the curve decoding profile family $\mathcal{F}=\mc{F}_{\ell, \delta, a, b}$.  This leads to a smaller union bound over all the profiles in the family $\mathcal{F}$, which means we can get away with less concentration for each profile $(\mathbf{V}, U)$.  In prior work, the field size $|\F|$ had to be taken fairly large to get the desired concentration; a better union bound means that we can take $|\F|$ smaller.

Our first observation inspired by \cite{gg26} to do a more careful dimension-aware counting. We emphasize further conceptual contributions from our canonicalization. For RLC and RRS, we prove curve decodability directly, without go through equivalence theorems through subspace design codes or appealing to \cite{gg25a}. The key idea is do pruning inspired by \cite{AHS26,Srivastava2025} directly on the level of local profiles, instead of the subspace design codes. This perspective is what gives us our canonical profiles, and also improves the count and final alphabet size that would arise from counting profiles directly via \cite{gg26} techniques and \cite{ggsw} equivalences.

We now present a simplified version of our idea, the first step for dimension-aware counting is to observe that for any choice of $A_i$ where $U \cap K_A(A_i)$ has a fixed codimension $c$ in $U$, there is a ``canonical'' $T \subset A_i$ so that $|T| = \ell + 1 + c$, and so that $U \cap K_A(T)= U \cap K_A(A_i)$.  Thus, to count the number of spaces $V_i \cap U$, it suffices to count the number of such $T$.  This is at most $a^{\ell + 1 + c}$, which is much smaller than $2^a$. When we combine for the $A_i$'s, we obtain $2^{(\ell+1)n+\sum_i c_i}$ as the profile count for fixed codimensions $c_1, \dots, c_n$ of $U \cap K_A(A_i)$. 

The bound arising from canonical profile counting is particular nice for RLC, where the probability that it contains a particular $r$-local LCL profile $(\mathbf{V}, U)$ with particular codimensions $c_i$ is $q^{Rn\dim U-\sum_i c_i}$ (see \cref{lem:rlc-fixed-pair}); so the $c_i$'s cancel at each level. This idea is enough to improve the alphabet size for curve-decoding to $(\ell/\eta)^{O(\ell/\eta)}$ for RLC.

For RRS, we have the improved LCL equivalence theorem (\cref{cor:rrs-equiv}), but that does not allow us to break things up by codimension in the same nice way as RLC. Naively, we can bound $c_i\le \dim U\le r$, but that is rather lossy.
Instead, we prove a ``pre-processing'' theorem, which states that, with high probability, any bad example for curve decoding for random RS codes must be fairly low-dimensional, of dimension at most $D =O( (\ell + 1)/\eta)$, thereby allowing us to bound $c_i\le D$ instead. Using these ideas, we reduce the union bound needed when we apply our result from \Cref{cor:rrs-equiv}(1) and improve the alphabet size for curve-decoding to $(\ell/\eta)^{O(\ell/\eta^2)}$ for RRS.

The RLC bound of $(\ell/\eta)^{O(\ell/\eta)}$ and RRS bound of $(\ell/\eta)^{O(\ell/\eta^2)}$ we attain above are discussed in \cref{rem:nice-stop}. Our final alphabet size bounds are slightly better by some log-factors in the exponent; these additional improvements come from a refinement to this counting argument, which was suggested by ChatGPT 5.6-Sol.

\subsection{Related Work}\label{sec:related}

\paragraph{Local Equivalence Theorems.}  As explained in \Cref{sec:overview}, \cite{lms25} showed that any property captured by the avoidance of a small collection $\mathcal{F}$ of LCL profiles $(\mathbf{V}, U)$ can be transferred from random linear codes to random RS codes with approximately the same rate.  In more detail, that work showed that both families of random ensembles have a ``threshold rate'' $R^\star$, so that a random code in the ensemble of rate less than $R^{\star} - \eta$ will avoid $\mathcal{F}$ with high probability, while a code of rate more than $R^\star + \eta$ will with high probability not avoid $\mathcal{F}$; and it showed that these thresholds are approximately the same for random linear codes and random RS codes.  This line of work goes back to \cite{mrrsw20} (see also \cite{GLMRS21,gm22}), which showed a similar connection between random linear codes and Gallager's ensemble of LDPC codes.  However, the link between random linear codes and random RS codes is much more involved, and it is this link that we strengthen in our work. 

The area of subspace design codes at this time already had exciting developments: \cite{Sri25, CZ25, AHS26} showed optimal list-decoding of subspace-design codes by , and new connections to Brascamp-Lieb inequalities were brought in by \cite{bcdz25a} in further work.

This sparked an exciting line of work, the work \cite{bcdz25b} showed that a similar local equivalence result holds for \emph{subspace design codes}.  In more detail, they showed that, for an LCL family $\mathcal{F}$, all subspace design codes of rate less than $R_{\mathcal{F}}$ avoids $\mathcal{F}$.  They show a converse, which is that for any $R > R_{\mathcal{F}}$, there is some subspace design code of rate $R$ that does not avoid $\mathcal{F}$.  This connection was exploited by \cite{bcdz25a}, who used it to show that random linear codes and random RS codes achieve list-recovery capacity.
We note that the equivalence theorems and algorithms of \cite{bcdz25a} have had applications beyond coding theory; for example, they are the primary engine in the breakthrough result of \cite{ChatterjeeGGRT2026}, showing the bipartite matching is in NC.

Our secondary proof of the containment of all optimal nested subspace design codes generalizes \cite{bcdz25b}'s equivalence theorems to now capture \textit{all} optimal nested subspace-design codes which include all known explicit constructions. 

\paragraph{Correlated Agreement and Curve Decoding.} Correlated agreement (as well as the related notions of proximity gaps and mutual correlated agreement) is relevant to interactive protocols.  Briefly, in these settings one may wish to test whether or not every point on a degree-$\ell$ curve $u(T)$ is close to a code $C$.  If $C$ has $(\ell, \delta, \eps)$-CA for sufficiently small $\eps$, then one may choose a random point $u(\alpha)$ on the curve and check its distance to $C$.  If $\Delta(u(\alpha), C) \leq \delta$, then one may conclude that all points on $u$ are $\delta$-close to $C$, and one would be wrong with probability at most $\eps$.  This was first observed in \cite{RVW13}, and further work has developed both guarantees and impossibility results~\cite{AHIV17,BKS18,BGKS20,bciks23,whir,gg25a,dg25,cs25,bchks25,hkk2026,ggsw,YZ26}.  See~\cite{ABF26} for a survey of work on proximity gaps and related notions.

It was known that the best possible value of $\delta$ (for $(\ell, \delta, a, b)$-curve decodability, or $(\ell, \delta, \eps)$-CA) must satisfy $\delta \leq 1 - R$ for a rate-$R$ code; but until recently it was not known how to achieve this, with prior constructions requiring $\delta < 1 - \sqrt{R}$.  The work \cite{gg25a} introduced curve decoding, and showed that a family of codes called \emph{subspace design codes} can have $(\ell, \delta, a, b)$-curve decodability for $\delta$ approaching $1 - R$; this implies proximity gaps, CA, and MCA with $\delta$ approaching $1 -R$.  However, subspace design codes necessarily have larger alphabets (in particular, they are $\F$-linear codes over an alphabet $\F^s$, for sufficiently large $s$), so it is also interesting to ask about $\F$-linear codes with alphabet $\F$ itself, ideally for $\F$ not too large.
Via techniques from \cite{bcdz25b,lms25}, the work \cite{gg25a} extended their results for subspace design codes to random ensembles of codes, including random linear codes and random RS codes.  Random RS codes are of particular interest, as Reed--Solomon codes are standard in applications.  However, as discussed above, the other parameters for curve decoding were suboptimal.

The work \cite{ggsw} improved the results of \cite{gg25a} for random ensembles of codes by more directly casting curve decodability as an LCL property; the work \cite{gg25a} had gone through a proxy called $\bigvee$-decoding, which led to some parameter loss.  As shown in \Cref{tab:CD}, the work \cite{ggsw} improved the ``$a$'' parameter in curve decodability (and the corresponding parameter $a = |\F| \cdot \eps$ for proximity gaps, CA, and MCA), and also substantially improved the required alphabet size over \cite{gg25a}.  

The work \cite{YZ26} took a different approach to random ensembles, considering a syndrome-based view.  That work studied CA for random linear codes and random RS codes, and further improved the dependence on $\eta$ and $\ell$ in the alphabet size $|\F|$, for $(\ell, \delta = 1 - R - \eta, \eps)$-CA.  However, the parameter $a = |\F|\cdot \eps$, the number of points that must be ``close'' to $C$ on a degree-$\ell$ curve in order to conclude that all points on the curve are close, was worse, with an exponential (rather than polynomial) dependence on $\eta$ and $\ell$.

\paragraph{Alphabet Size for List Decoding and List Recovery.}
There has also been extensive work on alphabet-size for random and explicit constructions of generalized Singleton bound achieving list-decoding codes \cite{Elias1957,GR08,GuruswamiX2022, GuoR2022,BrakensiekGM2023, JeronimoMST2025, AlrabiahGGLZ2025, LiS2025, lms25, SrivastavaT2025, bcdz25a,gg26}. Related works such as \cite{LiS2025, JeronimoS2025,GoyalG2025-Structure,ggh26} also work on list-recovery alphabet sizes. Our methods most directly build on the framework of \cite{gg26}.

\paragraph{Concurrent Works}
In concurrent and independent work, Yuan and Zhu~\cite{rag} extend the LCL threshold equivalence of \cite{lms25} to random one-point algebraic geometry codes under explicit conditions.
They also identify the gap in the above-threshold argument of \cite{lms25} for random Reed--Solomon codes and repair it using a related argument that follows chains of subspaces and apply non-containment statements to smaller subcodes.
Our work formulates this containment mechanism deterministically for row-span constrained profiles and suitable nested code families, including optimal nested subspace design codes, and additionally strengthens the RRS non-containment bound and compresses curve-decoding profiles to improve alphabet sizes.

In a breakthrough, \cite{BrakensiekCPZZ2026} showed that all Reed-Solomon codes of fixed rate $R$ over sufficiently large prime-fields are list-decodable to capacity $\delta = 1-R-\eps$ with polynomial $q^{O_{R, \eps}(1)}$ list-sizes. Subsequent work \cite{rsldmca2026,jeronimo2026} established the same quantitative results for proximity gaps and mutual correlated agreement; see \cite{HarshaKS2026} for an exposition. These results are algorithmic and apply to arbitrary evaluation points. Our work uses LCL equivalence and profile counting to obtain guarantees for random evaluation points in arbitrary characteristic, and achieve the best known proximity gap and MCA error guarantees. 

\cite{AlrabiahGG2026} recently gave an explicit construction of Reed-Solomon codes over large extension fields that are near-optimal subspace design codes over a subfield. They used this to show that these codes consequentially have list-decoding and proximity-gaps at capacity. Our results give randomized constructions over smaller fields, provide guarantees for bounded-degree curves, and improve the alphabet-size bounds for random linear codes.

Recent works on lower-bounds for list-decoding and proximity gaps close to capacity \cite{Zhu26,DKTZ26, AlrabiahGG2026-Johnson} concern particular evaluation points and parameter regimes, and do not apply in our setting of random choices of evaluation points.

\subsection{Acknowledgment of LLM Use}
ChatGPT-5.6-Sol was used in the following ways:
\begin{itemize}
    \item In \cref{sec:rrs-improvement}, we used ChatGPT to understand prior work. The dynamic accounting that improves the exponent in \cref{thm:low-lms} was made by the authors. After reading our proof, ChatGPT suggested the idea of presentations \cref{eq:presentation} to further improve the numerator of the base from $r^{O(r)}k$ to $O(rk)$. See \cref{rem:improve-lms} for a detailed comparison with prior work.
    \item  In \cref{sec:profile-counting}, we observed that the canonical-profile ideas of \cite{gg26,Sri25,AHS26} yield alphabet size bounds of $(\ell/\eta)^{O(\ell/\eta)}$ for RLC and $n(\ell/\eta)^{O(\ell/\eta^2)}$ for RRS (see \cref{rem:nice-stop}). In particular, the method of canonicalization and defining \cref{eq:def-low-strata} was by us. ChatGPT then suggested a double-counting refinement (see \cref{thm:codim-N-count}), that improves the RLC bound from $(\ell/\eta)^{O(\ell/\eta)}$ to
    
\[\exp\left(O\left(\ell\log\ell+\frac{\log(1/\eta)}{\eta}+\frac{\ell}{\eta} \right)
\right).\]

The authors subsequently observed that the counting could be further improved to remove the first term of $\ell\log \ell$. ChatGPT also suggested an RRS refinement inspired by \cite{gg26} (see \cref{lem:low-family-count}), which improves the RRS alphabet size bound from $n(\ell/\eta)^{O(\ell/\eta^2)}$ to
\[
n\exp\left(
O\left(\frac{\ell\log\ell}{\eta}+\frac{\ell}{\eta^2}\right)\right).
\]
\item Outside of the contributions described above, ChatGPT was used only for copyediting and proof checking. All proofs appearing in the manuscript were written by the authors. All responsibility for errors  lies with the authors.
\end{itemize}

\section{Preliminaries}
\label{sec:prelim}

We begin by introducing the basic coding-theoretic definitions we will be using. For any two vectors $x, y$ in $\Sigma^n$, where $\Sigma$ is some alphabet, we define the fractional Hamming distance $\Delta(x, y) = \frac{1}{n} |\{i\in [n]: x_i \ne y_i\}|$ to be the fraction of coordinates where they differ.
For a set $S\subset \Sigma^n$, we define $\Delta(x, S) = \min_{y\in S} \Delta(x,y)$
to be the fractional distance of $x$ to its closest vector in $S$.
Throughout this paper, all distances are taken to be fractional unless stated otherwise.

\begin{definition}[Additive codes]\label{def: additive-code}
    Let $\F$ be the finite field with $q$ elements and $\Sigma = \F^s$ for some positive integer $s$.
    A code $C \subset \Sigma^n$ is said to be $\F$-additive (or just additive when the field $\F$ is clear from context) if $C$ is an $\F$-linear subspace of $\Sigma^n$.
\end{definition}

The two fundamental quantities associated with a code are its rate and distance.
For a code $C \subset \Sigma^n$, we define its (relative) distance as $\delta(C) = \min_{x, y \in C, x\ne y} \Delta(x,y)$.
Moreover, its rate $R(C)$ is defined as $R(C)= \frac{1}{n} \log_{|\Sigma|}|C|$.

Following \cite{ABF26}, we define the \emph{$r$-fold interleaving} of a code $C\subset(\F^s)^n$ by
\begin{equation}
  C^{\equiv r}
  :=
  \left\{
    (c^{(1)},\ldots,c^{(r)}):c^{(j)}\in C\text{ for every }j\in[r]
  \right\}.
\end{equation}
We view an element of $C^{\equiv r}$ as an $sn\times r$ matrix over
$\F$, whose $(i,h)$-th row is
$(c^{(1)}_{i,h},\ldots,c^{(r)}_{i,h})$.  When $s=1$, this is the usual
$n\times r$ matrix representation.

We focus on three families of random codes, defined in \cref{subsec:codes}.

\subsection{Row-span Constrained Local Properties}
\label{subsec:lcl}

As described in \cref{sec:overview}, a \emph{local profile} gives a set
of coordinate-wise linear constraints that a ``bad'' witness matrix
$M$ might satisfy.  We first recall local profiles from \cite{lms25},
as well as row-span constrained LCL properties from \cite{ggsw}.

Given an $\F$-vector space $V$, let $\mc{L}(V)$ denote the set of
linear subspaces of $V$.

\begin{definition}[$r$-Local Profile and Containment,
  {\cite[Definitions~3.1--3.2]{ggsw}}]
\label{def:local-profile}
Let $\F$ be a finite field and let $n,r\in\N$.  An $r$-local profile is
a tuple
\begin{equation}
  \mbf{V}=(V_1,\ldots,V_n)\in\mc{L}(\F^r)^n.
\end{equation}
An $r$-local pair is a pair $(\mbf{V},U)$ with
$U\in\mc{L}(\F^r)$.
A code $C\subset\F^n$ is said to \emph{contain} $(\mbf{V},U)$ if there
exists a matrix $M\in C^{\equiv r}$ such that:
\begin{itemize}
\item the $i$th row of $M$ lies in $V_i$ for every $i\in[n]$; and
\item $\on{rowspan}_{\F}M=U$.
\end{itemize}
We call such an $M$ a \emph{witness matrix}, and its columns
\emph{witness columns}.

It is convenient to let $\mc{M}_{\mbf{V}}(U)$ denote the set of
matrices satisfying the local constraints and having row span exactly
$U$, and let $\mc{M}_{\mbf{V}}^{\star}(U)$ denote the linear space of
matrices satisfying the local constraints and having row span contained
in $U$.
\end{definition}

\begin{definition}[Row-span Constrained LCL Property,
  {\cite[Definition~3.3]{ggsw}}]
\label{def:rsclp}
A \emph{row-span constrained $r$-local LCL family} $\mc{F}$ is a
collection of $r$-local pairs $(\mbf{V},U)$, where
$\mbf{V}=(V_1,\ldots,V_n)\in\mc{L}(\F^r)^n$ is an $r$-local profile
and $U\in\mc{L}(\F^r)$ is the row-span constraint of the witness
matrix.

A code $C$ \emph{contains} $\mc{F}$ if it contains some pair
$(\mbf{V},U)\in\mc{F}$, and it \emph{avoids} $\mc{F}$ otherwise.
We say that the \emph{row-span constrained $r$-local property} given by
$\mc{F}$ is the property of avoiding $\mc{F}$.
\end{definition}

Next, we recall the potential and threshold functions, first defined for
LCL properties by \cite{lms25} and generalized to row-span constrained
LCL properties by \cite{ggsw}.

\begin{definition}[Potential and Threshold,
  {\cite[Definition~3.7]{ggsw}}]
\label{def:potential-threshold}
For an $r$-local profile $\mbf{V}$, a subspace $X\subset\F^r$, and a
real number $R$, define the \emph{potential}
\begin{equation}
\begin{aligned}
  \phi_{\mbf{V}}(X,R)
  &:=(R-1)\dim_{\F}X
    +\mb{E}_{i\in[n]}\dim_{\F}(X\cap V_i)\\
  &=R\dim_{\F}X
    -\mb{E}_{i\in[n]}
       \on{codim}_{X}(X\cap V_i).
\end{aligned}
  \label{eq:def-potential}
\end{equation}
For an $r$-local pair $(\mbf{V},U)$ and $W\subsetneq U$, define
\begin{equation}
  R_{\mbf{V}}(U,W)
  :=
  1-\mb{E}_{i\in[n]}
  \left[
    \frac{
      \dim_{\F}(U\cap V_i)-\dim_{\F}(W\cap V_i)
    }{
      \dim_{\F}U-\dim_{\F}W
    }
  \right].
  \label{eq:def-relative-threshold}
\end{equation}
Define the \emph{threshold} of $(\mbf{V},U)$ by
\begin{equation}
  R_{\mbf{V}}(U)
  :=
  \max_{W\subsetneq U}R_{\mbf{V}}(U,W).
  \label{eq:def-pair-threshold}
\end{equation}
For a row-span constrained $r$-local LCL family $\mc{F}$, define
\begin{equation}
  R_{\mc{F}}
  :=
  \min_{(\mbf{V},U)\in\mc{F}}R_{\mbf{V}}(U).
  \label{eq:Rf}
\end{equation}
\end{definition}

The potential and threshold functions are closely related. One such useful consequence is the following identity (see e.g. \cite[Lemma 3.8]{ggsw})
\begin{equation}
\label{eq:potential-gap}
\phi_{\mbf{V}}(U,R)-\phi_{\mbf{V}}(W,R)=  \left(\dim_{\F}U-\dim_{\F}W\right)
  \left(R-R_{\mbf{V}}(U,W)\right).
\end{equation}

Finally, we recall quotient profiles and the maximal-quotient lemma from
\cite{ggsw}.

\begin{definition}[Quotient Profile,
  {\cite[Definition~3.11]{ggsw}}]
\label{def:quotient-profile}
Let $\mbf{V}=(V_1,\ldots,V_n)$ be an $r$-local profile and let
$\pi:\F^r\to\F^{r'}$ be a surjective linear map, where $r'\le r$.
Define the $r'$-local profile
\begin{equation}
  \pi\mbf{V}:=(\pi V_1,\ldots,\pi V_n).
  \label{eq:def-quotient-profile}
\end{equation}
For an $r$-local pair $(\mbf{V},U)$, if
$W:=\ker\pi\subsetneq U$, then $(\pi\mbf{V},\pi U)$ is a
\emph{quotient pair} of $(\mbf{V},U)$.  The condition
$\ker\pi\subsetneq U$ is equivalent to $\pi U\ne\{0\}$.
\end{definition}

We record the behavior of the potential function and profile containment under quotients. 

\begin{lemma}[{\cite[3.12 and 3.13]{ggsw}}]\label{lem:simple-quotient}
Suppose code $C$ contains $r$-local pair $(\mbf{V}, U)$ and $\pi: \FF_q^r\to \FF_q^{r'}$ is a surjective linear map with kernel $W\subsetneq U\in\mc{L}(\mb{F}_q^r)$ and image $\pi U\ne 0$. Then, $C$ contains $(\pi\mbf{V}, \pi U)$ and, for any $R\in\mb{R}$,
\begin{equation}
	\phi_{\pi\mbf{V}}(\pi U,R)
	=
	\phi_{\mbf{V}}(U,R)-\phi_{\mbf{V}}(W,R)
\end{equation}
\end{lemma}

We now recall the key lemma about maximal quotients from \cite{ggsw}.

\begin{lemma}[Maximal Quotient, {\cite[Lemma~3.14]{ggsw}}]
\label{lem:maximal-quotient}
Let $(\mbf{V},U)$ be an $r$-local pair with $U\ne\{0\}$, and suppose
\begin{equation}
  R\le R_{\mbf{V}}(U)-\eta
\end{equation}
for some $R\in\R$ and $\eta\ge0$.  There is a surjective linear map
$\pi:\F^r\to\F^{r'}$, depending only on
$(\mbf{V},U),R$, and $\eta$, such that $\ker\pi\subsetneq U$ and:
\begin{itemize}
\item every code containing $(\mbf{V},U)$ also contains
$(\pi\mbf{V},\pi U)$,
\item $\phi_{\pi\mbf{V}}(U',R)\le -\eta\dim_\F U'$
for every nonzero subspace $U'\subset\pi U$.
\end{itemize}
Any quotient map $\pi$ whose kernel is an inclusion-maximal element of $\argmax_W R_{\mbf{V}}(U, W)$ suffices.
\end{lemma}
The last statement is not in the statement of \cite[Lemma 3.14]{ggsw}, but it is clear by its proof in \cite{ggsw}. We formally record it here additionally for convenience.

We will usually apply \cref{lem:maximal-quotient} from a potential bound
like $\phi_{\mbf{V}}(U,R)\le-\eta\dim_{\F}U$.
We apply \cref{eq:potential-gap} with $W=\{0\}$ and trivial equality $\phi_{\mbf{V}}(\{0\},R)=0$ to obtain
\begin{equation}
  R\le R_{\mbf{V}}(U,0)-\eta
  \le R_{\mbf{V}}(U)-\eta,
  \label{eq:root-gap-implies-threshold}
\end{equation}
so the hypothesis of \cref{lem:maximal-quotient} holds.
\subsection{Curve Decoding}
\label{subsec:curve-dec}

We begin by recalling the definition of curve decoding, first introduced in \cite{gg25a}. This is the main LCL property we study in this paper.
\begin{definition}
\label{def: curve-decodability}
    An additive code $C\subset \Sigma^n$ is $(\ell, \delta, a, b)$-curve-decodable if, for every $u_0, u_1,\cdots, u_\ell\in \Sigma^n$ and every function $f: \F \to C$, whenever the set 
    \[ A=\left\{\alpha \in \F : \Delta\left(\smudef, f(\alpha)\right) \le  \delta\right\}\] has at least $a$ elements, there exist $c_0, c_1,\cdots, c_\ell \in C$ such that 
    \[\left|\left\{\alpha \in A  :  f(\alpha)=\smcdef \right\}\right|\ge b \ .  \]
\end{definition}

The key observation to cast curve decoding as a local property is the following definition. 

\begin{definition}[Curve-freeness]
\label{def:curve-free}
For $S\subset A\subset\F$, define
\begin{equation}
K_A(S)
:=
\left\{
y\in\F^A:
y|_S\text{ agrees with a polynomial of degree at most }\ell
\right\}.
\end{equation}
A subspace $U\subset\F^A$ is
\emph{$(\ell,A,b)$-curve-free} if
$
U\not\subset K_A(B)
$ for every $B\in\binom{A}{b}$.
\end{definition}

Since every row of a witness matrix with row-span $U$ already lies in
$U$, we use the normalized local spaces
$V_i=U\cap K_A(A_i)$ below. This is equivalent to
\cite[Theorem~4.4]{ggsw}, which uses $K_A(A_i)$.

\begin{theorem}[{\cite[Theorem 4.4]{ggsw}}]
\label{thm:ggsw-family}
Let $\ell+1\le b\le a\le q$ and $\delta\in(0,1)$.  A linear code $C$
fails to be $(\ell,\delta,a,b)$-curve-decodable if and only if there are
$A\in\binom{\F}{a}$, an $(\ell,A,b)$-curve-free subspace
$U\subset\F^A$, and sets $A_1,\ldots,A_n\subset A$ such that  $C$ contains $(\mbf V,U)$ where
\begin{equation}
  V_i:=U\cap K_A(A_i)
\end{equation}
for every $i\in [n]$ and, for every $\alpha\in A$,
\begin{equation}
  \left|\left\{i\in[n]:\alpha\in A_i\right\}\right|
  \ge(1-\delta)n.
\end{equation}
The collection $\mc{F}_{\ell, \delta, a, b}$ of all resulting profile pairs has size at most
$q^{a+a^2}2^{an}$.
\end{theorem}
We now explain the curve decoding parameters we need to obtain downstream proximity gap and (mutual) correlated agreement implications, and the relevant threshold rate.
\begin{corollary}[{\cite[Corollary~4.6]{ggsw}}]
\label{cor:curve-profile-threshold}
Throughout the paper we take
\begin{equation}
  D:=\ceil{\frac{\ell+1}{\eta}},
  \qquad
  b:=\ceil{\frac{4\ell}{\eta}+\ell},
  \qquad
  a:=\ceil{\frac{100\ell^2}{\eta^3}}
  \label{eq:ggsw-param}
\end{equation}
With this choice of parameters,
\begin{equation}
\label{eq:cd-threshold}
R_{\mc{F}_{\ell, \delta, a, b}} \ge 1-\delta-\eta
\end{equation}
\end{corollary}

We now informally recall the notions of proximity gaps and
(mutual) correlated agreement. Informally,
with respect to degree-$\ell$ curves, an $\F$-additive code
$C\subset\Sigma^n$ is said to have
\begin{itemize}
\item $(\ell,\delta,\eps,\gamma)$ \emph{proximity gap (PG)} if, for
every degree-$\ell$ curve
$u(x)=\sum_{j=0}^{\ell}u_jx^j$ and every $\delta'\le\delta$, either at
most an $\eps$ fraction of the points $u(\alpha)$ are
$\delta'$-close to $C$, or every point $u(\alpha)$ is
$\delta'/(1-\gamma)$-close to $C$.

\item $(\ell,\delta,\eps,\gamma)$ \emph{correlated agreement (CA)} if,
for every degree-$\ell$ curve
$u(x)=\sum_{j=0}^{\ell}u_jx^j$ and every $\delta'\le\delta$, either at
most an $\eps$ fraction of the points $u(\alpha)$ are
$\delta'$-close to $C$, or there is a codeword curve
$c(x)=\sum_{j=0}^{\ell}c_jx^j$, with $c_j\in C$, such that $u(x)$ and
$c(x)$ agree identically on a common set $S\subset[n]$ containing at least a $1-{\delta'}/(1-\gamma)$ fraction of the coordinates.

\item $(\ell,\delta,\eps)$ \emph{mutual correlated agreement (MCA)}
if, for every degree-$\ell$ curve
$u(x)=\sum_{j=0}^{\ell}u_jx^j$ and every $\delta'\le\delta$, the
following holds for at least a $1-\eps$ fraction of
$\alpha\in\F$: for every $S\subset[n]$ of size at least
$(1-\delta')n$, if some $v\in C$ agrees with $u(\alpha)$ on $S$, then
there is a codeword curve
$c(x)=\sum_{j=0}^{\ell}c_jx^j$, with $c_j\in C$, that agrees
identically with $u(x)$ on $S$.
\end{itemize}

We state the formal definitions. More details can be found in the survey
\cite{ABF26}.

\begin{definition}[Proximity Gaps (PG)]\label{def: curve PG}
   An $\FF$-additive code $C\subset  \Sigma^n$ is said to have $(\ell, \delta, \eps, \gamma)$ proximity gap if for all $u_0, u_1, \cdots, u_\ell \in \Sigma^n$, and for all $\delta'\le \delta$, \[\mb{P}_{\alpha\in \FF}\left(\Delta\left(\sum_{j=0}^\ell u_j \alpha^j, C\right)\le \delta'\right)>\eps \implies \forall \alpha \in \FF, \ \Delta\left(\sum_{j=0}^\ell u_j \alpha^j, C\right)\le \frac{\delta'}{1-\gamma}\ .   \]
   When $\gamma = 0$, we say that $C$ has an $(\ell, \delta, \eps)$ proximity gap.  
\end{definition}

\begin{definition}[Correlated Agreement (CA)]\label{def: curve CA}
    An $\FF$-additive code $C\subset  \Sigma^n$ is said to have $(\ell, \delta, \eps, \gamma)$ correlated agreement if for all $u_0, u_1, \cdots, u_\ell \in \Sigma^n$ and $\delta'\le \delta$ 
    \[\mb{P}_{\alpha\in \FF}\left( \Delta\left(\sum_{j=0}^\ell u_j \alpha^j, C\right)\le \delta'\right)>\eps\implies
    \exists c_0, c_1,\cdots c_{\ell} \in C:
    \left|\left\{i : \bigvee_{j=0}^\ell  c_{j, i}\ne u_{j, i}\right\}\right| \le\frac{\delta' n}{1-\gamma} . \] 
\end{definition}

\begin{definition}[Mutual Correlated Agreement (MCA)]\label{def: curve MCA}
    An $\FF$-additive code $C\subset  \Sigma^n$ is said to have $(\ell, \delta, \eps)$ mutual correlated agreement if for all $u_0, u_1, \cdots, u_\ell \in \Sigma^n$ and $\delta'\le \delta$ 
    \[\mb{P}_{\alpha\in\FF}\left(\exists S \subset  [n], c\in C, j\in\{0, 1, \dots, \ell\}: |S|\ge (1-\delta')n, c_{\mid S}=\sum_{j=0}^\ell \alpha^j u_{j\mid S},  c'_{\mid S}\ne u_{j\mid S} \forall c'\in C \right)\le \eps  . \]
\end{definition}

We note that $(\ell,\delta,\eps,\gamma)$ correlated agreement implies
an $(\ell,\delta,\eps,\gamma)$ proximity gap, while
$(\ell,\delta,\eps)$ mutual correlated agreement implies
$(\ell,\delta,\eps,\gamma=0)$ correlated agreement and hence
$(\ell,\delta,\eps,\gamma=0)$ proximity gap.
We record how these properties follow from curve decodability.  The
following theorem is written essentially verbatim from \cite{ggsw},
although the underlying implications were first derived in
\cite{gg25a}.

\begin{theorem}[{\cite[Corollary 2.11]{ggsw}}]\label{thm: main-ca-mca-gg}
If an $\F$-additive code $C\subset (\F^s)^n$ is \begin{equation}
\left(\delta\left(1+\frac{\ell}{b-\ell}\right), L\right)\text{ list decodable \quad and \quad } \left(\ell, \delta, a, b\right)\text{ curve-decodable}
\end{equation} 
then $C$ has: \begin{enumerate}
    \item for any $t\ge 1$, $(\ell, \delta, (t-1)L+a, t)$ curve decodability,
    \item for any $m >1$, $\left(\ell, \delta, \frac{\ell m L + a}{q}, \frac{1}{m}\right)$ correlated agreement and proximity gap,
    \item $(\ell, \delta, \frac{\ell n L+a}{q})$ mutual correlated agreement.
\end{enumerate}
\end{theorem}

\subsection{Random Code Ensembles}
\label{subsec:codes}
We now define our three main random code ensembles. These definitions are standard and, in particular, are taken verbatim from \cite{ggsw}. For each code, we record a fixed-pair estimate that upper-bounds the probability that a local profile pair $(\mbf{V}, U)$ is contained in a random code $C$ of rate $R$ in terms of the potential $\phi_\mbf{V}(U, R)$, as well as results showing that it is list decodable to capacity.

The remaining ingredient of the paper is to union bound over all pairs in specific families $\mc{F}$ corresponding to curve decoding, at which point we obtain curve decodability and capacity-achieving list decodability for the random codes with high probability, and apply \cref{thm: main-ca-mca-gg} to conclude proximity gaps and (mutual) correlated agreement.

\begin{definition}[Random linear codes]
\label{def:rlc}
    A random linear code $C\subset \F^n$ of rate $R$ is given by taking 
    \[ C = \{(\mathsf{Enc}_1(x),\cdots, \mathsf{Enc}_n(x)) : x\in \F^{Rn}\}\] where $\mathsf{Enc}_i : \F^{Rn}\to \F$ are independent uniformly random linear maps.
    More generally, a random $\F$-additive code $C \subset (\F^s)^n$ of rate $R$ is given by taking 
    \[ C = \{ (\mathsf{Enc}_1(x), \ldots, \mathsf{Enc}_n(x)) : x \in \F^{Rsn} \}\] where $\mathsf{Enc}_i: \F^{Rsn} \to \F^s$ are independent uniformly random linear maps. 
\end{definition}

We record the fixed-pair probability bound and list decodability of RLCs. The former is observed in \cite{lms25,ggsw} and follows from dimension-counting observations on $|\mc{M}_\mbf{V}^\star(U)|$ as well as \cite[Lemma 4.1]{lms25}.
\begin{lemma}[{\cite[Corollary~3.18]{ggsw}}]
\label{lem:rlc-fixed-pair}
If $C$ is a random linear code of rate $R\in [0, 1]$ and $(\mathbf{V}, U)$  is an $r$-local pair, then
\begin{equation}
\mb{P}\left(C\text{ contains }(\mbf V,U)\right) \le \mb E\left|
C^{\equiv r}\cap\mc M_{\mbf V}(U)
\right|
\le q^{n\phi_{\mbf V}(U,R)}.
\end{equation}
\end{lemma}

\begin{theorem}[{\cite[Corollary~3.11]{gg26}}]
\label{thm:rlc-ld}
For every $\eta,R\in(0,1)$, there is an absolute constant $c_0>0$
such that, if $q\ge \exp\left(c_0\log(2/\eta)/\eta\right)$ and $n$ is
sufficiently large, a random linear code $C\subset\F^n$ of rate $R$ is
$\left(1-R-3\eta/2,\ceil{(1-R)/\eta}\right)$ list-decodable with
probability at least $1-8q^{-\eta n/8}$.
\end{theorem}

The second random code family consists of Reed--Solomon codes
with random evaluation points.

\begin{definition}[Random Reed--Solomon codes]
\label{def:rrs}
    A Reed--Solomon code over a field $\F$ with evaluation points $\alpha_1,\cdots, \alpha_n$ is defined as follows:
    \[\mathsf{RS}(\boldsymbol{\alpha}, k)=\mathsf{RS}_{\F}(\alpha_1,\dots, \alpha_n, k):= \{(f(\alpha_1),\dots, f(\alpha_n)): f\in \F[X], \deg f<k\}\]
    A random Reed--Solomon code of rate $R$ is one in which $k=Rn$ and $(\alpha_1,\ldots,\alpha_n)$ is sampled uniformly from the ordered tuples of distinct elements of $\F$; that is, the evaluation points are sampled without replacement.
\end{definition}
We record the fixed-pair bound and list decodability of RRS codes. In \cref{sec:rrs-improvement}, we improve the former to remove the $r$ in the denominator of the exponent of \cref{eq:lms-fixed-pair}, thereby improving the alphabet size bound.

\begin{proposition}[{\cite[Proposition 6.1]{lms25}}]
\label{prop:lms-fixed-pair}
Let $n<q$, let $\mbf V=(V_1,\ldots,V_n)$ be an $r$-local profile,
let $U\in\mc L(\F^r)$, and let $k=Rn\in\N$. Suppose that, for some
$\eta>0$,
\begin{equation}
\phi_{\mbf V}(U',R)
\le
-\eta\dim_{\F}U'
\label{eq:lms-fixed-pair-hypothesis}
\end{equation}
for every nonzero $U'\subset U$.
Then, for $\boldsymbol{\alpha}$ sampled uniformly without replacement
from $\F$,
\begin{equation}
\mb{P}_{\boldsymbol{\alpha}}\left(
  \mathsf{RS}(\boldsymbol{\alpha},k)^{\equiv r}
  \cap\mc M_{\mbf V}^{\star}(U)
  \ne\{0\}
\right)
\le
(2^r-1)
\left(
  \frac{(4r)^{4r}k}{\eta(q-n)}
\right)^{\eta n/(2r)}.
\label{eq:lms-fixed-pair}
\end{equation}
\end{proposition}

\begin{theorem}[{\cite[Theorem 1.1]{AlrabiahGGLZ2025}}]\label{thm:rrs-ld}
For every $\eta\in(0,1)$, there is an absolute constant $c_0>0$
such that, if $q\ge n\exp\left(c_0/\eta^2\right)$, then with high
probability a random Reed--Solomon code $C\subset\F^n$ of rate $R$ is
\begin{equation}
	\left(1-R-\frac{3\eta}{2},\ceil{(1-R)/\eta}\right)
	\text{ list decodable}.
\end{equation}
\end{theorem}

We introduce our last family of codes, random Low-Density Parity-Check (LDPC) codes.
In this paper, by ``random LDPC codes'' we mean \emph{Gallager's ensemble} of random LDPC codes~\cite{gallager}.

\begin{definition}[Gallager random LDPC codes]
\label{def:rldpc}
    Assume $s\mid n$ and $(1-R)s\in\N$. Let $t=(1-R)s$, and let $G_i = (V, W_i, E_i)$ for $i=1,\cdots, t$ be uniformly random $(1,s)$ regular bipartite graphs with a shared left vertex set and disjoint right vertex sets $W_i$ of size $n/s$. The graph $G$ is then defined as $\bigsqcup_i G_i$; that is, $G = (V,W,E)$, where $W = \bigsqcup_i W_i$ and $E = \bigsqcup_i E_i$.  The edge weights $w_e$ are sampled independently and uniformly at random from $\F^\ast$. A \emph{Random LDPC Code} $C$ of rate $R$ is then the LDPC code derived from $G$ and the weights $\{w_e\}$, as described above.
\end{definition}

We now record the statement analogous to \cref{lem:rlc-fixed-pair} for RLDPC codes. The fact that the ``good event'' happens with high probability was first observed in \cite[Theorem 2.14]{mrrsw20}. 

\begin{proposition}[Fixed-pair RLDPC estimate, {\cite[Lemmas~3.22--3.23]{ggsw}}]
\label{prop:rldpc-fixed-pair}
For any $\eta, R \in (0,1)$ and any positive integer $r$, there exists an $s_0=s_0(\eta, r, R, q)$ such that the following holds for all odd $s>s_0$ such that $(1-R)s\in\N$, and all sufficiently large multiples $n$ of $s$. Let $\mbf{V}=(V_1,\ldots,V_n)$ be an $r$-local profile, let
$U\in\mc{L}(\F^r)$, and let 
\begin{equation}
\mathsf{Good}
    \coloneqq
    \left\{
	        \delta(C)\ge \frac{1}{2}h_q^{-1}(1-R)
    \right\}.
\end{equation}
Then, with probability over $C=\mathsf{RLDPC}(n,q,s,R)$,
\begin{equation}
\mb{P}
    \left(\mathsf{Good}\right)
 \ge 1-o_{n\to\infty}(1)
\quad\text{and}\quad 
\mb{P}\left(
    C\text{ contains }(\mathbf{V},U)
\middle|
    \mathsf{Good}
\right)
\le q^{n\phi_{\mathbf{V}}(U,R)+\eta n}.
\end{equation}
\end{proposition}

The following theorem is immediate by rerunning the refined list-decoding
profile count underlying \cref{thm:rlc-ld} (quoted from
\cite[Corollary~3.11]{gg26}) with the fixed-pair RLDPC estimate and
good-event reduction already recorded in \cref{prop:rldpc-fixed-pair}.

\begin{theorem}[{\cite[Corollary~3.11]{gg26} combined with \cite[Lemmas~3.22--3.23]{ggsw}}]
\label{thm:rldpc-ld}
There is an absolute constant $c_0>0$ such that, for every $\eta, R\in(0,1)$, there exists
$s_0=s_0(\eta,q,R)$ such that the following holds for every odd
$s>s_0$ such that $(1-R)s\in\N$.  Let $C=\mathsf{RLDPC}(n,q,s,R)$ be a random LDPC code from
Gallager's ensemble of rate $R$. If
$q\ge \exp\left(c_0\log(2/\eta)/\eta\right)$\footnote{Note that this alphabet size suffices for list decoding of RLDPC codes due to the counting of list-decoding profiles of \cite{gg26} and the fixed-pair RLDPC estimate of \cite{ggsw}.} and $n$ is a sufficiently large multiple of $s$, then with high
probability $C$ is
\begin{equation}
    \left(
        1-R-\frac{3\eta}{2},
        \ceil{2(1-R)/\eta}
    \right)
    \text{ list decodable}.
\end{equation}
\end{theorem}

The fixed-pair statements above are the probability estimates that we will use. In \cite[Corollaries~3.19 and~3.26]{ggsw}, maximal quotienting and a union bound over a general family $\mc F$ turn the RLC and RLDPC fixed-pair estimates into general local threshold theorems. Similarly, \cite[Corollary~3.21]{ggsw} applies maximal quotienting and a union bound to \cref{prop:lms-fixed-pair} to obtain the low-rate RRS threshold theorem.  These general results replace the potential of each individual profile by a uniform worst-case gap and then take a single union bound over $\mc F$. 

For the particular curve decoding local property in \cref{subsec:curve-dec}, that is too coarse. In \cref{sec:profile-counting}, we instead compress $\mc F$ and apply the fixed-pair estimates separately to the resulting families, retaining some extra dependence for a refined analysis. For random Reed--Solomon codes, we additionally improve \cref{prop:lms-fixed-pair} in \cref{sec:rrs-improvement} before carrying out this refined union bound.

\subsection{Subspace-Design Codes}
\label{subsec:sdc}
This section concerns $\F$-additive codes $C\subset(\F^s)^n$.
We introduce notation and results from \cite{gg25a,ggsw} that we will use.

\begin{definition}[Subspace-Design Code]\label{def: Subspace-designs}
    For any function $\tau: \mathbb N\rightarrow \R_{\le 1}$, an $\FF_q$-additive code $C\subset  (\FF_q^s)^n$ is said to be a $\tau$-\textit{subspace design} code if for every $r\in \mathbb N$, and every $\F_q$-linear subspace $A$ of $C$ of dimension at most $r$, the following holds: 
    \[\frac{1}{n}\sum\limits_{i=1}^n \dim A_i\le \dim(A)\cdot \tau(r) \]
    where $A_i = \{a\in A\mid a_i=0\}$.
\end{definition}

\begin{definition}\label{def:sdc}
For an additive code $C\subset \left(\mb{F}_q^s\right)^n$ and dimension $r$, let
\begin{equation}
\sigma(r)\coloneqq \max\left\{\mb{E}_{i\in [n]}\left[\frac{\dim V_i}{\dim V}\right]:V\subset C, \dim V=r\right\}
\end{equation}
where $V_i = \{x\in V:x_i=0\}$. 
In the message space with encoding maps $\on{Enc}_i:\mb{F}_q^k\to\mb{F}_q^s$, we have
\begin{equation}\label{eq:sdc-message}
	\sigma(r)\coloneqq \max\left\{\mb{E}_{i\in [n]}\left[\frac{\dim (A\cap \ker(\on{Enc}_i))}{\dim A}\right]:A\subset \mb{F}_q^k, \dim A=r\right\}
\end{equation}
Following the notation of \cite{gg25a}, we can define the monotone envelope $\tau(r)\coloneqq \max\{\sigma(1), \dots, \sigma(r)\}$.  Observe that with this choice, $C$ is by definition a $\tau$-subspace design code.
\end{definition}

The paper \cite{bcdz25b} uses slightly different language; they say that code $C$ of rate $R$ is \emph{$r$-subspace designable if}, for every $d\in [r]$
\[ 
	\sigma(d)\le R+\frac{1}{nd}.
\]
We recall that $\sigma (r)\ge R-r/sn$ always. Thus, having the subspace design property with ``good'' parameters means that we should have $\tau(d)\le R+\eta$ for some ``small'' $\eta>0$.

A canonical subspace design code family are folded Reed-Solomon codes, defined below.

\begin{definition}[$s$-Folded Reed-Solomon Codes \cite{GR08}]\label{def: FRS Codes}
        Let $n, k, s$ be positive integers and $\FF_q$ be a field with $q>sn$ and $\gamma$ be an element of the field.
        Additionally, let $\alpha_1, \alpha_2, \ldots, \alpha_n \in \FF_q$ be distinct elements such that $\alpha_i \gamma^t \ne \alpha_j$ for all $i\ne j$ and $t<s$.
        
        The $s$-Folded Reed-Solomon Code $\mathcal C\subseteq (\FF_q^s)^n$ with message length $k$ on $(\alpha_1, \ldots, \alpha_n)$ is given by: \[\on{FRS}_{q, n, k, s}( \bar\alpha)=\left\{\left(\begin{bmatrix}
            f(\alpha_1) \\
            f(\alpha_1\gamma)\\
            \vdots\\
            f(\alpha_1 \gamma^{s-1})
        \end{bmatrix}, \begin{bmatrix}
            f(\alpha_2) \\
            f(\alpha_2\gamma)\\
            \vdots\\
            f(\alpha_2 \gamma^{s-1})
        \end{bmatrix}, \cdots, \begin{bmatrix}
            f(\alpha_n) \\
            f(\alpha_n\gamma)\\
            \vdots\\
            f(\alpha_n \gamma^{s-1})
        \end{bmatrix} \right)\Bigg| f\in \FF_{q}[x], \deg f<k \right\} \ . \qedhere \]
\end{definition}

\begin{theorem}[\cite{GK16}]\label{thm: GKS16 subspace design}
    Folded Reed-Solomon codes of order $s$ and rate $R$ are $\tau$-subspace-design codes for $\tau(r)= \frac{sR }{s-r+1}$ for all $1\le r\le s$ and $\tau(r)=1$ otherwise.
\end{theorem}

We next state what it means for profile containment for $\F$-additive codes \cite{bcdz25b}. Essentially, we unfold each symbol into $s$ scalar coordinates
and impose $V_i$ on each scalar row in block $i$. It turns out this folding and unfolding preserves the threshold, as observed in \cite{bcdz25b}. We state verbatim the version in \cite{ggsw}.

\begin{definition}[Folded Containment]
\label{def:folded-profile}
Let $C\subset(\F^s)^n$ be an $\F$-additive code.  For an
$r$-local profile
$\mbf{V}=(V_1,\ldots,V_n)\in\mc{L}(\F^r)^n$, define the
\emph{duplicated profile}
$\mbf{V}^{\otimes s}\in\mc{L}(\F^r)^{sn}$ by
\begin{equation}
  \left(\mbf{V}^{\otimes s}\right)_{(i,h)}:=V_i
  \label{eq:tensor-profile}
\end{equation}
for every $i\in[n]$ and $h\in[s]$.
We say that $C$ contains $(\mbf{V},U)$ if the ``unfolded'' version
$C\subset\F^{sn}$ contains $(\mbf{V}^{\otimes s},U)$.
Here, by the ``unfolded code,'' we mean the image of
$C\subset(\F^s)^n$ in $\F^{sn}$ under the natural identification.
For a family $\mc{F}$ of $r$-local pairs, define
\begin{equation}
  \mc{F}^{\otimes s}
  :=
  \left\{
    (\mbf{V}^{\otimes s},U):
    (\mbf{V},U)\in\mc{F}
  \right\}.
\end{equation}
Thus, an $\F$-additive code $C\subset(\F^s)^n$ contains $\mc{F}$ if
and only if its unfolding $C\subset\F^{sn}$ contains
$\mc{F}^{\otimes s}$.
\end{definition}
\begin{proposition}[{\cite[Proposition 3.9]{ggsw}}]\label{prop:unfolding}
Let $(\mbf{V}, U)$ be any $r$-local pair on plain codes in $\mb{F}_q^n$ and let $W\subsetneq U$.
Then, for every $s\in\mb{N}$
\begin{equation}
R_{V^{\otimes s}}(U,W)=R_V(U,W).
\end{equation}
Consequently, for any row-span constrained local LCL family, $R_{\mc{F}}=R_{\mc{F}^{\otimes s}}$ is independent of $s$.
\end{proposition}
Previous works on the local equivalence of subspace design codes take the following existential form. In particular, the second containment statement is proven by taking a random additive code and observing it is subspace design with high probability, and has the same threshold as the unfolded plain RLC by \cref{prop:unfolding}. We quote the version of \cite{ggsw} but it is observed proven first in \cite[Theorems~4.1--4.2]{bcdz25b}.
\begin{theorem}[{\cite[Theorem 3.27]{ggsw}, }]\label{thm:sdc-threshold-pair}
Let $(\mbf{V}, U)$ be an $r$-local pair. Then
\begin{enumerate}
	\item Let $C\subset(\mb{F}_q^s)^n$ be an additive $\tau$-subspace design code. If $\tau(r) < R_\mbf{V}(U)$, then $C$ does not contain $(\mbf{V}, U)$.
	\item For any $\eta>0$, there exists 
    a $\tau$-subspace design additive code $C\subset(\mb{F}_q^s)^n$ with $\tau(r)\le R_{\mbf{V}}(U)+\eta$ that contains $(\mbf{V}, U)$.
\end{enumerate}
\end{theorem}

For completeness, we remark that \cite[Theorem~4.7]{gg25a} gives a direct
curve decoding theorem for subspace design codes. Since we are only interested in curve decoding fo random ensembles, we omit this result; we will only require \cref{cor:curve-profile-threshold}.

\subsection{Notations and Conventions}
\label{sec:notations}
We let $[n]\coloneqq \{1, \ldots, n\}$ with the convention $[0]=\emptyset$. Let $\binom{n}{\le k}$ and $\binom{n}{\ge k}$ denote the sum of binomial coefficients $\binom{n}{i}$ for $i\le k$ and $i\ge k$, respectively. Let $\binom{S}{k}$ denote the set of $k$-element subsets of $S$. $\binom{S}{\le k}$ and $\binom{S}{\ge k}$ denote the set of subsets of $S$ with size at most $k$ and at least $k$, respectively. For a finite set $S$, let $\mb{E}_{s\in S}$ denote the expectation over a uniform random element $s$ of $S$.

For a matrix $A\in\F^{m\times n}$, let $A_{i\star}\in\F^n$ be the $i$th row of $A$ and let $A_{\star j}\in\F^m$ be the $j$th column of $A$. Given an $\F$-vector space $V$, let $\mc{L}(V)$ denote the set of
linear subspaces of $V$.

We also adopt standard notation from asymptotic analysis: as $x\to\infty$, we write  $f(x)=o(g(x))$ if $f(x)/g(x)\to 0$;  $f(x)=O(g(x))$ or $g(x)=\Omega(f(x))$ if there exists a finite, positive constant $C$ such that $f(x)\le Cg(x)$ for all sufficiently large $x$.

We use $h_q:[0,1] \to [0,1]$ and $h_q^{-1}$ to denote the $q$-ary entropy function and its inverse, i.e.
\[ h_q(x) := x \log_q(q-1) + x \log_q\left(\frac{1}{x}\right) + (1 - x) \log_q\left( \frac{1}{1 - x}\right).\]

\section{Local Profile Non-containment for Random Reed--Solomon Codes}
\label{sec:rrs-improvement}

Let $\mbf{V}=(V_1,\ldots,V_n)$ be an $r$-local profile
over $\F$, let $0\ne U\subset\F^r$, let $d:=\dim_{\F}U$, and let
  $q:=\abs{\F}$.  In every probability statement,
  $\alpha_1,\ldots,\alpha_n$ are sampled uniformly without replacement
  from $\F$.  Recall the normalized potential
\begin{equation}
  \phi_{\mbf{V}}(W,R)
  :=R\dim_{\F}W
    -\frac1n\sum_{i=1}^n
      \on{codim}_{W,\F}(W\cap V_i).
  \label{eq:intro-potential}
\end{equation}
Recall the matrix space $\mc M_{\mbf V}^{\star}(U)$ from
\cref{def:local-profile}, where the star denotes that the row span is
contained in $U$, and let $\mathsf{RS}(\boldsymbol\alpha,k)$ denote the
Reed--Solomon code from \cref{def:rrs}. As before, let $C^{\equiv r}$ denote the $r$-wise interleaving of code $C$.

Our main result in this section is an improved fixed-profile theorem of
\cite{lms25}, i.e. the low-rate non-containment direction of \cref{cor:rrs-equiv}. We state it formally below.
\begin{theorem}[Improved RRS Non-containment]
\label{thm:low-lms}
Let $q>\max\{n,d(k-1)\}$, $k=Rn\in [n]$, and $\eta \in (0, 1]$ such that
$\phi_{\mbf{V}}(W,R)\le-\eta\dim_{\F}W$
for every $0\ne W\subset U$, and that $d:=\dim_\F U \le \eta n/2$. Then
\begin{equation}
  \mb{P}_{\boldsymbol\alpha}\left(
    \mathsf{RS}(\boldsymbol\alpha,k)^{\equiv r}
    \cap\mc{M}_{\mbf{V}}^\star(U)\ne\{0\}
  \right)
  \le
  d\left(
    \frac{2ed(k-1)}{\eta(q-n)}
  \right)^{\eta n-d}
  \le r\left(
    \frac{2erk}{\eta(q-n)}
  \right)^{\eta n/2}.
  \label{eq:relative-bound}
\end{equation}
\end{theorem}

In this section, we give a self-contained proof of \cref{thm:low-lms}, including the stochastic process in the proof of \cite[Proposition 6.1]{lms25}. Here, we begin with a brief overview to outline the steps:

\begin{enumerate}
\item We first remove irrelevant ambient coordinates to assume without loss of generality that $U=\F^d$. We choose isomorphism $U\simeq \F^d$ and replace $V_i$ by the image of $U\cap V_i$ under the isomorphism.
\item We run a stochastic process where we reveal the evaluation points $\boldsymbol{\alpha}$, and therefore
the constraints they introduce, one by one in a random process indexed by $t\in[n]$. We let $S_t\subset \F[X]_{<k}^d$ be the candidates remaining subject to the constraints we see up to time $t$. \footnote{Technically, we could have taken quotients and the $S_t$ is not a subspace of $\F[X]_{<k}^d$, but morally it is useful to view it as a symbolic version the space of possible witness matrices.}
\item We define a monovariant $m_t+d_t$ that tracks the predicted dimension of the final solution space $S_n$, pretending that the unexposed evaluation points are generic. To facilitate this, we lift the potential to symbolic space over $\K=\F(X)$ and define a dynamic potential $\varphi_t$.
\item Each time step $t$, as we reveal $\alpha_{t}$, we impose the constraints introduced by $\alpha_{t}$ on the previous candidates $S_{t-1}$ to obtain $S_{t-1}'$. Usually, we produce one child state by taking $S_{t}=S_{t-1}'$, but sometimes we fix a critical subspace, branch into two children $S_t^{\on{in}}$ and $S_t^{\on{out}}$, and look for candidates in either the quotient by or the restriction to the subspace instead.
\item At every step, including branching, we show that the monovariant of a child state cannot increase unless
there is a ``defect event'', which occurs only if a symbolically nonzero determinant vanishes when we choose evaluation point $\alpha_t$. On this defect event, the monovariant increases by at most $1$. The defect happens with probability $p_d$ by a Schwartz-Zippel type lemma.
\item We show that initially at time $t=0$, the dynamic potential (and therefore the monovariant), is at most $-\Omega(\eta n)$ by assumptions on the negativity of the ordinary potential $\phi_{\mbf{V}}(U, R)$. At the end, if there is a witness matrix, one of the branches of the tree must terminate with positive potential (and therefore the monovariant). 
\item We need $L=\Omega(\eta n)$ defects along some branch of $n$ steps for the monovariant to become positive at the end and produce a witness matrix. This gives a bound of the shape $\binom{n}{L}p_d^L$.
\end{enumerate}
We note that the first improvement in \cref{rem:improve-lms} comes from tracking a dynamic monovariant in Step 3 rather than a fixed schedule; and the second improvement in \cref{rem:improve-lms} is an improved bound on $p_d$ in Step 5 by a more careful analysis of the degree of the determinant polynomial.

\subsection{Dynamic Potential}
\label{subsec:inputs}

For Step 1, we choose an isomorphism
$B:\F^d\to U$ and replace each $V_i$ by
$B^{-1}(U\cap V_i)$.  Multiplication by $B$ identifies the corresponding matrix spaces and preserves the potential of every subspace of
$U$.  We carry out the proof with
$U=\F^d$ and undo the identification at the end, i.e. the ambient dimension is $d$ and locality $r$ appears only through $d\le r$. We reuse $\mbf V$ and
$U$ for the resulting pair.

We now define the dynamic potential and our monovariant, per Steps 2 and 3.  Let $\K:=\F(X)$. We treat $n, k, d$ as fixed throughout the process. Initially without revealing any $\boldsymbol{\alpha}$, we
\begin{itemize}
    \item Choose for
every $i$ a full-row-rank parity-check matrix $A_i\in \F^{\on{codim}_{\F^d}V_i\times d}$ over $\F$ with
kernel $V_i$, and extend $A_i$ coefficientwise to $\F[X]^d$ and $\K$-linearly to $\K^d$. 
\item Define initial candidates of the witness matrix to be $S_0:=\F[X]_{<k}^{d}$. By Step 1, locality which is the number of columns of the matrix, is now $d$ instead of $r$.
\end{itemize}
Then, define $\U_t:=\on{span}_{\K}S_t$ and
$d_t:=\dim_{\K}\U_t$, so initially $\U_0 = \K^d$ and $d_0 = d$. 

As we reveal $\boldsymbol{\alpha}$ one-by-one, the stochastic process of how $S_t$ is derived from $S_{t-1}$ upon revealing $\alpha_t$ is deferred to \cref{subsec:gap-process}. In this section, we study the situation after we observe $\alpha_t$ but before we define $S_t$. On the previous candidates $S_{t-1}$, we define the subspace of those satisfying the constraint imposed by $A_t$ and $\alpha_t$, i.e.
\begin{equation}
  S_{t-1}':=\{P\in S_{t-1}:A_t(P)(\alpha_t)=0\}.
  \label{eq:exposure}
\end{equation}
Analogously, we let $\U_{t-1}':=\on{span}_{\K}S_{t-1}'$ and
$d_{t-1}':=\dim_{\K}\U_{t-1}'$.

\begin{remark}
We look ahead slightly to explain the difference between $S_{t-1}'$ versus $S_t$.
\begin{itemize}
    \item Usually, there is no difference: each state $S_{t-1}$ gives rise to one child state $S_{t}=S_{t-1}'$, so we will produce the chain
\[ \F[X]_{<k}^d=S_0\supset S_1\supset\cdots\supset S_n.\]
\item In rare cases (see \cref{subsec:gap-process}), we \emph{branch} from $S_{t-1}$ into two children corresponding to a restriction to and a quotient by some subspace $\W_t\subset \U_{t-1}'$. We first compute $S_{t-1}'$ as above, but then the next step would split into $S_t^{\on{in}}$ and $S_t^{\on{out}}$ (see \cref{prop:branching}). 
\item In general, this chain $(S_0, \dots, S_n)$ refers to the candidates and spaces inherited by each branch; sometimes we abuse notation and say ``consider $S_0, \dots, S_n$ along a branch'' in this way.
\end{itemize}
\end{remark}

The above remark on branching is only to build intuition and explain our notation. In this section, we only consider the transition from $S_{t-1}$ to $S_{t-1}'$ in \cref{eq:exposure}. We define the dynamic potential function indexed by time $t$, supposing the current candidate space at time $t$ along a branch is $S_t$.

\begin{definition}[Dynamic potential]
\label{def:dynamic-potential}
For a nonzero $\F$-linear subspace $S$ with $\U = \on{span}_\K S$, define
\[
  \msc{G}(S)
  :=\left\{
      0\ne\W\subset\U:
      \on{span}_{\K}(S\cap\W)=\W
    \right\}
\]
to be $S$-generated spaces.
For $\W\in\msc{G}(S_t)$, define
\begin{equation}
  \varphi_t(\W)
  :=\dim_{\F}(S_t\cap\W)-\sum_{i=t+1}^n c_i(\W)
  \quad\text{where}\quad
  c_i(\W)
  :=\dim_{\K}A_i(\W).
  \label{eq:dynamic-potential}
\end{equation}
Analogously, we define $\varphi_{t-1}'(\W) :=\dim_\F (S_{t-1}'\cap \W)-\sum_{i=t+1}^nc_i(\W)$.
\end{definition}
Note that $\msc{G}(S_t)$ are live spaces at time $t$ in \cite{lms25}.
We make a few remarks to intuitively explain what $\varphi_t$ does, and why it is a dynamic version of $\phi_{\mbf{V}}$. We use script letters for spaces over $\K$ and ordinary letters for spaces over
$\F$. 
\begin{itemize}
\item At time $t=0$, 
$\varphi_0$ is morally $\phi_{\mbf V}$ up to a change in normalization by $n$ (see \cref{lem:initial}).
\item After we reveal some evaluation points $\boldsymbol{\alpha}$ to obtain candidates $S_t$, $\dim_{\F}(S_t\cap\W)$ counts the remaining candidate degrees of
freedom. By rank-nullity, $c_i(\W)= \on{codim}_{\W,\K}
      \left(\W\cap\ker_{\K}A_i\right)$ is the generic number of constraints imposed
on $\W$ by the future constraints $V_i$ given by $A_i$.

\item The family $\msc G(S)$ is finite: every member is the $\K$-span
of an $\F$-subspace of the finite-dimensional space $S$.  It contains
$\on{span}_{\K}S$ whenever $S\ne\{0\}$.
\end{itemize}

We now record three deterministic ingredients proven in \cite{lms25}, in preparation of the stochastic process analysis. Our first lemma concerns the initialization of the dynamic potential. This combines the rank-preserving specialization, local codimension comparison, and coordinate dimension estimates in
\cite[Lemma 7.1, 7.6, 7.7]{lms25}.

\begin{lemma}
\label{lem:initial}
Let $q > d(k-1)$. For every $\W\in\msc{G}(S_0)$, there exists a nonzero subspace
$W\subset U$ such that $\dim_{\F}W=\dim_{\K}\W$ and
\begin{equation}
  \frac{\varphi_0(\W)}{\dim_{\K}\W}
  \le
  n\frac{\phi_{\mbf V}(W,R)}{\dim_{\F}W}.
  \label{eq:root-comparison}
\end{equation}
\end{lemma}
\begin{proof}
Let $T:=S_0\cap \msc{W}$, so $w:=\dim_\K \W\le \dim_\K \U_0 = d$. Since $\W$ is $S_0$-generated, we can pick a $\K$-basis of $\W$ inside $T$, say $P_1, \dots, P_w\in T\cap\W$. Its columns are independent, so projection onto some $w$ coordinates indexed by $J\in \binom{d}{w}$ is an isomorphism. In particular, it embeds $T$ into $\F[X]_{<k}^w$ indexed by $J$, so
\[ \dim_\F T \le \dim_\F (\F[X]_{<k}^w) = kw.\]

For the matrix with columns $P_1, \dots, P_w$, the $w\times w$ minor indexed by $J$ has nonzero determinant with degree at most $w(k-1) \le d(k-1)<q$, so there is some $\beta\in\F$ at which its determinant is nonzero. This implies $\dim_\F W\ge w$ where
\[ W := \{P(\beta):P\in T\}\subset \F^d.\]
On the other hand, specialization at $X=\beta$ cannot increase rank generically, so $\dim_\F W\le w$, so $\dim_\F W=w$. 
Similarly, specialization to $\beta$ cannot increase the rank of $\F$-matrix $A_i$ on $T$, so 
\[ \on{codim}_W(W\cap V_i)=\dim_\F A_i(W) \le \dim_\K A_i(\W)=c_i(\W).\]
Therefore, using $k=Rn$, we obtain 
\[  \varphi_0(\W)
  \le kw-\sum_{i=1}^n\on{codim}_{W}(W\cap V_i)=n\phi_{\mbf V}\left(W,\frac{k}{n}\right)= n\phi_{\mbf V}(W,R).\qedhere\]
\end{proof}

Consequently, the hypothesis of \cref{thm:low-lms} gives, for any $\W\in\msc G(S_0)$, that
\begin{equation}
  \frac{\varphi_0(\W)}{\dim_{\K}\W}\le-\eta n.
  \label{eq:initial-potential-gap}
\end{equation}
Let $\mathscr{F}_t :=\sigma(\alpha_1, \dots, \alpha_t)$.
We call reveal $t$ a \emph{defect for $\W$} when, upon revealing $\alpha_t$, the number of constraints it imposes on candidates in $\W$ is less than the number $c_t(\W)$ we predicted, i.e. 
\begin{equation}
  \dim_{\F}\{A_t(P)(\alpha_t):P\in S_{t-1}\cap\W\}
  <c_t(\W).
  \label{eq:def-defect}
\end{equation}
We next show the deterministic part of \cite[Lemma~7.9]{lms25}, which says that the dynamic potential can only increase under a defect, and not by much.
\begin{lemma}
\label{lem:one-step}
For every $t$, we have that $\msc{G}(S_{t-1}')\subset\msc{G}(S_{t-1})$, and that
\begin{equation}
  \varphi_{t-1}'(\W)-\varphi_{t-1}(\W)
  \le\dim_{\K}\W
  \label{eq:one-step-dynamics}
\end{equation}for every
$\W\in\msc{G}(S_{t-1}')$.
The left hand side is positive if and only if reveal $t$ is defective for $\W$.
\end{lemma}
\begin{proof}
Fix any $\W\in \msc{G}(S_{t-1}' )$, so $T':=S_{t-1}' \cap \W$ generates $\W$. Then, $T:=S_{t-1}\cap \W$ also generates $\W$, giving $\W\in\msc{G}(S_{t-1})$. Now, by definition of $S_{t-1}' $ from $S_{t-1}$ and \cref{def:dynamic-potential}, we see that
\begin{equation}
\begin{aligned}
 c'&:=\dim_{\F}\{A_t(P)(\alpha_t):P\in T\} = \dim_\F T - \dim_\F T',
\\ \varphi_{t-1}'(\W)-\varphi_{t-1}(\W) &= \dim_\F T' - \dim_\F T+c_t(\W) = c_t(\W)-c'.
\end{aligned}
\end{equation}
Since $c'\ge 0$, we have that $\varphi_{t-1}'(\W)-\varphi_{t-1}(\W)
\le c_t(\W)
\le \dim_{\K}\W$,
which proves \cref{eq:one-step-dynamics}. Moreover, by
\cref{eq:def-defect}, reveal $t$ is defective for $\W$ if and
only if $c'<c_t(\W)$, which, by the preceding identity, holds if and only if $\varphi_{t-1}'(\W)-\varphi_{t-1}(\W)>0$.
\end{proof}

Note that, without branching, the process will choose $S_{t}=S_{t-1}'$, so this lemma shows that the normalized potential is usually unchanged; it can increase only at a defect and by at most one.   

\subsection{The Stochastic Process}
\label{subsec:gap-process}
In this section, we define the stochastic process formally to explain how parent state $S_{t-1}$ transitions to children state $S_t$ from intermediate state $S_{t-1}'$. One key idea that will facilitate the improvement on defect probability bound $p_d$ is the following notion of a presentation that we will construct.

\begin{definition}
For an $\F$-linear space $S$, we say that $(Z, E)$ \emph{presents} $S$ if $Z\subset E\subset \F[X]_{<k}^d$ and, after writing $\msc{Z}:= \on{span}_\K Z$ and $\msc{E}:=\on{span}_\K E$, we have that $E\cap\msc{Z}=Z$, as well as isomorphisms 
\begin{equation}
  E/Z\simeq S
  \qquad\text{and}\qquad 
  \msc E/\msc Z\simeq\U:=\on{span}_\K S.
  \label{eq:presentation}
\end{equation}
\end{definition}

A presentation is immediate for $S_0=\F[X]_{<k}^d$: we take
$E=\F[X]_{<k}^d$ and $Z=\{0\}$. We show that it is preserved whenever the candidate space is replaced by a subspace, e.g. from $S_{t-1}$ to $S_{t-1}'$.

\begin{lemma}
\label{lem:easy-presentation}
If $(Z, E)$ presents $S$ with $\F$-linear subspace
$S'\subset S$, then $(Z, E')$ presents $S'$ where $E'\subset E$ is the preimage under the quotient map of the set in $E/Z\simeq S$ identified with $S'\subset S$.
\end{lemma}

\begin{proof}
Let $\U=\on{span}_\K S$ and $\U'=\on{span}_\K S'$.
Note that $Z\subset E'\subset E$ and $E'/Z\simeq S'$. Observe that 
\[ Z\subset E'\cap \msc{Z}\subset E\cap\msc{Z}=Z,\]
so $E'\cap \msc{Z}=Z$. Finally, the restriction $\psi$ to $\msc{E'}/\msc{Z}$ of the isomorphism $\msc{E}/\msc{Z}\simeq \msc{U}$ is a $\K$-linear injection, so it commutes with $\on{span}_\K$. Thus, $\psi(\msc{E'}/\msc{Z}) = \on{span}_\K (\psi(E'/Z)) = \on{span}_\K S'=\U'$.
\end{proof}

Recall from \cref{eq:exposure} that $S_{t-1}'$ is defined to be polynomial tuples $P$ in $S_{t-1}$ that satisfy the constraint
$A_t(P)(\alpha_t)=0$.  Usually, we let $S_t=S_{t-1}'$. To define the stochastic process formally, we handle the only situation where we branch: we
identify a critical subspace $\W_t$ of $\U_{t-1}'$ called ``worrisome'' in
\cite{lms25}, and branch into a restriction and a quotient with
$\W_t$.  

The following proposition says that spans, dimensions, constraints,
potentials, and presentations are all well-behaved upon splitting at $\W_t$. This is in the proof of \cite[Lemma 7.10]{lms25}.

\begin{proposition}[Branching]
\label{prop:branching}
Let $S_{t-1}'\ne\{0\}$ and let
$\W_t\in\msc G(S_{t-1}')$ be a proper subspace of $\U_{t-1}'$
with $w_t:=\dim_\K \W_t$.  Let $ \pi_t:\U_{t-1}'\to \U_{t-1}'/\W_t$ be the quotient map, and define
\[
  S_t^{\on{in}}:=S_{t-1}'\cap\W_t,
  \qquad
  S_t^{\on{out}}:=\pi_t(S_{t-1}').
\]
Then, the following holds:
\begin{enumerate}
\item $\on{span}_{\K}(S_t^{\on{in}})=\W_t$ and has $\K$-dimension $w_t$;  
$\on{span}_{\K}(S_t^{\on{out}})
=\pi_t(\U_{t-1}')\simeq\U_{t-1}'/\W_t$ 
and has $\K$-dimension
$d_{t-1}'-w_t$. As $0\ne \W_t\subsetneq \U_{t-1}'$, both $\K$-spans are nonzero.
\item For each $i>t$, 
$S_t^{\on{in}}$ inherits the future constraint $A_i^{\on{in}}:=A_i\vert_{\W_t}$ and $S_t^{\on{out}}$ inherits
\[
  A_i^{\on{out}}:
  \U_{t-1}'/\W_t
  \longrightarrow
  A_i(\U_{t-1}')/A_i(\W_t)
\quad \text{where}\quad
  A_i^{\on{out}}(\pi_t(P))
  :=A_i(P)+A_i(\W_t).
\]
For any nonzero $P\in S_{t-1}'$ satisfying all
future parent constraints, the appropriate nonzero child $P$  or $\pi_t(P)$ satisfies the corresponding inherited
constraints.

\item
Let $\varphi_t^{\on{in}}$ and $\varphi_t^{\on{out}}$ be defined by
\cref{def:dynamic-potential} on $S_t^{\on{in}}$ and
$S_t^{\on{out}}$ with constraints in (2).  Then,
$\msc G(S_t^{\on{in}})\subset\msc G(S_{t-1}')$ and
$\varphi_{t-1}'$ agrees with $\varphi_t^{\on{in}}$ on the restriction.
Moreover, every
$\msc X\in\msc G(S_t^{\on{out}})$ has a generated preimage
$\msc Y:=\pi_t^{-1}(\msc X)\in\msc G(S_{t-1}')$ such that
$\W_t\subsetneq\msc Y\subset\U_{t-1}'$ and
\begin{equation}
\dim_{\K}\msc X=\dim_{\K}\msc Y-w_t,
\qquad
\varphi_t^{\on{out}}(\msc X)
=\varphi_{t-1}'(\msc Y)-\varphi_{t-1}'(\W_t).
\label{eq:quotient-identity}
\end{equation}

\item Suppose $(Z,E)$ presents $S_{t-1}'$, let
$H:=E\cap\msc H$ where $\msc H$ is the preimage under the quotient map
of the set in
$\msc E/\msc Z\simeq\U_{t-1}'$ identified with
$\W_t\subset\U_{t-1}'$.  Then,
$\on{span}_{\K}H=\msc H$, and $S_t^{\on{in}}$ is presented by
$(Z,H)$ while $S_t^{\on{out}}$ is presented by $(H,E)$.
\end{enumerate}
\end{proposition}

\begin{proof}
For (1), the definition of an $S_{t-1}'$-generated space
gives $\W_t\ne\{0\}$ and that
\[
  \on{span}_{\K}(S_t^{\on{in}})=\on{span}_{\K}(S_{t-1}'\cap\W_t)=\W_t,
\]
which has dimension $w_t$.  Moreover, since
$\on{span}_{\K}S_{t-1}'=\U_{t-1}'$ and $\pi_t$ is $\K$-linear,
\[
  \on{span}_{\K}(S_t^{\on{out}})
  =\on{span}_{\K}\pi_t(S_{t-1}')
  =\pi_t(\on{span}_{\K}S_{t-1}')
  =\pi_t(\U_{t-1}')
  =\U_{t-1}'/\W_t.
\]
By rank-nullity, this space has dimension
$d_{t-1}'-w_t$. This proves (1).

For (2), we first check that $A_i^{\on{out}}$ is well-defined: if $\pi_t(P)=\pi_t(P')$, then
$P-P'\in\W_t$, so
$A_i(P)-A_i(P')\in A_i(\W_t)$. Now, suppose nonzero $\ne P\in S_{t-1}'$ satisfies $A_i(P)=0$ for some $i>t$. 
If $P\in\W_t$, then $P\in S_t^{\on{in}}$ is nonzero and satisfies
$A_i^{\on{in}}$.  If $P\notin\W_t$, then
$\pi_t(P)\in S_t^{\on{out}}$ is nonzero, and 
\[  A_i^{\on{out}}(\pi_t(P))
  =A_i(P)+A_i(\W_t) =0,\]
so $\pi_t(P)$ satisfies $ A_i^{\on{out}}$. This proves (2).

For (3), if
$\msc X\in\msc G(S_t^{\on{in}})$, then $\msc X\subset\W_t$, so
 $S_t^{\on{in}}\cap\msc X
  =(S_{t-1}'\cap\W_t)\cap\msc X
  =S_{t-1}'\cap\msc X$. Therefore $\msc X\in\msc G(S_{t-1}')$. As $S_t^{\on{in}}$ and $A_i^{\on{in}}$ are simply restrictions to $\W_t$, we have $\varphi_t^{\on{in}}(\msc X)
  =\varphi_{t-1}'(\msc X)$.
Now let $\msc X\in\msc G(S_t^{\on{out}})$ and set
$\msc Y:=\pi_t^{-1}(\msc X)$.  Since $\msc X\ne\{0\}$, we have
$\W_t\subsetneq\msc Y\subset\U_{t-1}'$.  Now, observe that
\[
  \pi_t(S_{t-1}'\cap\msc Y)
  =\pi_t(S_{t-1}')\cap\msc X
  =S_t^{\on{out}}\cap\msc X.
\]
Thus $S_{t-1}'\cap\msc Y$ maps to a generating set of $\msc X$. It
contains $S_{t-1}'\cap\W_t$, which generates
$\ker(\pi_t|_{\msc Y})=\W_t$.  Generators of the kernel together with
lifts of generators of the quotient generate $\msc Y$.  Hence
$\msc Y\in\msc G(S_{t-1}')$.
Finally, rank-nullity for $\pi_t$ and its induced candidate and
constraint maps gives
\begin{align*}
  \dim_{\K}\msc X
  &=\dim_{\K}\msc Y-w_t,\\
  \dim_{\F}(S_t^{\on{out}}\cap\msc X)
  &=\dim_{\F}(S_{t-1}'\cap\msc Y)
    -\dim_{\F}(S_{t-1}'\cap\W_t),\\
  c_i^{\on{out}}(\msc X)
  &=c_i(\msc Y)-c_i(\W_t),
\end{align*}
for $i>t$. Substituting these
identities into the definition of $\varphi_{t-1}'$ proves (3).

For (4), the presentation identifies
$S_{t-1}'\simeq E/Z$ inside
$\U_{t-1}'\simeq\msc E/\msc Z$, which identifies the intermediate spaces $H$ and $\msc H$. Therefore, the isomorphisms required in the presentation $(Z,H)$ of $S_t^{\on{in}}$ and $(H,E)$ of $S_t^{\on{out}}$ follow from the correspondence and
isomorphism theorems.
It suffices to check that
$\on{span}_{\K}H=\msc H$.  Indeed, $\msc H/\msc Z\simeq \W_t$ and
$H/Z\simeq S_{t-1}'\cap\W_t$ by the isomorphism theorems. Now, quotient $H/Z$ spans $\msc{H}/\msc{Z}$ as $\W_t\in\msc G(S_{t-1}')$. Since $H$ also contains $Z$, it spans the entire $\msc H$.
\end{proof}

We now formally define the stochastic process and our monovariant. For every $t$ and state $S_t$, our monovariant is $d_t+m_t$ where
\begin{equation}
  m_t:=\max_{\W\in\msc G(S_t)}
       \frac{\varphi_t(\W)}{\dim_{\K}\W}\qquad \text{and}\quad d_t:=\on{dim}_\K \U_t,
  \label{eq:gap-function}
\end{equation}
where $\U_t = \on{span}_\K S_t$ as before.
Immediately before reveal
$t$, every active branch carries $\mathscr F_{t-1}$-measurable data
$S_{t-1}$ and its $\K$-span $\U_{t-1}$.
After revealing $\alpha_t$, we define $S_{t-1}'$ and $\U_{t-1}'$ with dimension $d_{t-1}'$ by \cref{eq:exposure}. We similarly define
\begin{equation}
  m_{t-1}':=\max_{\W\in\msc G(S_{t-1}')}
       \frac{\varphi_{t-1}'(\W)}{\dim_{\K}\W}.
\end{equation}

The maximum exists by finiteness, and we let $\W_t$ be the maximizer. We break ties by maximizing $w_t:=\dim_\K \W_t$ and with any deterministic rule thereafter, so $\W_t$ is $\mathscr F_t$-measurable. 
At initialization, \cref{lem:initial} gives $m_0+d_0\le d-\eta n$. For $t\ge 1$, 
if $S_t\ne\{0\}$, every space generated by $S_t$ was generated before the reveal, so
\cref{lem:one-step} gives
\begin{equation}
  m_{t-1}'\le m_{t-1}+1.
  \label{eq:one-step-gap}
\end{equation}
Since $S_{t-1}'\subset S_{t-1}$, one has
$d_{t-1}'\le d_{t-1}$. 
We apply exactly one of the following four
updates.
\begin{enumerate}
\item \emph{Terminate.}  If $S_{t-1}'=\{0\}$, stop and call the branch \emph{inactive} 
\item \emph{Continue.}  Otherwise, if
$d_{t-1}'<d_{t-1}$ or 
$d_{t-1}'+m_{t-1}'\le d_{t-1}+m_{t-1}$, set $S_t=S_{t-1}'$ and continue.
\item \emph{Mark.}  Otherwise,
$d_{t-1}'=d_{t-1}$ and
$d_{t-1}'+m_{t-1}' > d_{t-1}+m_{t-1}$. If $\W_t=\U_{t-1}'$, mark time $t$, set $S_t=S_{t-1}'$, and continue.
\item \emph{Branch.} Otherwise, $\W_t\subsetneq \U_{t-1}'$.
Apply \cref{prop:branching} and continue with both
$S_t^{\on{in}}$ and $S_t^{\on{out}}$
\end{enumerate}

Thus, all choices depend only on $\alpha_1, \dots, \alpha_t$, i.e. they are $\mathscr F_t$-measurable. 
The update is designed so that continuing a branch or replacing it by
either child never increases $d_t+m_t$. The only
nontrivial structural check is at \emph{branch}, where the two parts of
\cref{prop:branching} are applied separately. In all other cases, $d_t=d_{t-1}'$ and $m_t=m_{t-1}'$. The sole permitted
potential increase is at \emph{mark}. We record and check this monovariant below.
\begin{proposition}
\label{prop:monovariant}
Along every active branch $S_t$ indexed by time $t$, $S_t$ has a presentation and
\[
  d_t+m_t\le d_{t-1}+m_{t-1}+\mbf{1}\left\{\text{the branch is marked at time $t$}\right\}.
\]
Moreover, every mark is a defect for the $\mathscr F_{t-1}$-measurable space $\U_{t-1}$.
\end{proposition}

\begin{proof}
We induct along an active branch. For the base case, $S_0=\mb{F}[X]_{<k}^d$ is presented by $(\{0\}, \mb{F}[X]_{<k}^d)$ and $d_0+m_0\le d-\eta n$. We go through the four cases above.
\begin{enumerate}
\item \emph{Terminate.} Then there is nothing to prove.
\item \emph{Continue.} Here, $d_t=d_{t-1}'$ and $m_t=m_{t-1}'$; the existence of a presentation of $S_t=S_{t-1}'\subset S_{t-1}$ follows from \cref{lem:easy-presentation}. If
$d_t+m_t\le d_{t-1}+m_{t-1}$, the lemma is proved. Otherwise, if
$d_t<d_{t-1}$, then \cref{lem:one-step} gives
$m_t\le m_{t-1}+1$ so $(d_t+m_t)-(d_{t-1}+m_{t-1})
  \le 1-(d_{t-1}-d_t)\le 0$.
\item \emph{Mark.} Here, $d_t=d_{t-1}'$ and $m_t=m_{t-1}'$; the existence of a presentation of $S_t=S_{t-1}'\subset S_{t-1}$ follows from \cref{lem:easy-presentation}. Moreover, $\U_t=\U_{t-1}=\W_t$, so $m_t>m_{t-1}$ or else we are in \emph{continue}. 
As $\W_t$ attains $m_t$, 
\[
 \frac{\varphi_t(\U_{t-1})}{d_{t-1}}
 =m_t>m_{t-1}
 \ge\frac{\varphi_{t-1}(\U_{t-1})}{d_{t-1}}.
\]
Thus \cref{lem:one-step} shows that the reveal is defective. The lemma follows, since $m_t-m_{t-1}\le 1$ by \cref{lem:one-step} and $d_{t-1}=d_t$.
\item \emph{Branch.} 
The existence of a presentation of $S_{t-1}'$ follows from \cref{lem:easy-presentation}, and presentations of both children follow from \cref{prop:branching}(4). Here, $\W_t\subsetneq \U_{t-1}'$, so
$w_t:=\dim_{\K}\W_t<d_{t-1}'$.  By \cref{prop:branching}, the branches satisfy
\[
  d_t^{\on{in}}=w_t,\qquad m_t^{\on{in}}=m_{t-1}',
  \qquad
  d_t^{\on{out}}=d_{t-1}'-w_t,\qquad
  m_t^{\on{out}}\le m_{t-1}'.
\]
Indeed, $\W_t$ attains $m_{t-1}'$ in the restriction, while the quotient-potential identity bounds the normalized potential of every quotient-generated space by
$m_{t-1}'$. Since the branch case occurred instead of continue, $d_{t-1}'=d_{t-1}$, and $m_{t-1}'\le m_{t-1}+1$ by \cref{lem:one-step}. Since $1\le w_t\le d_{t-1}'-1$,
\[
\begin{aligned}
  d_t^{\on{in}}+m_t^{\on{in}}
  &=w_t+m_{t-1}'
   \le d_{t-1}+m_{t-1},\\
  d_t^{\on{out}}+m_t^{\on{out}}
  &\le d_{t-1}'-w_t+m_{t-1}'
   \le d_{t-1}+m_{t-1}.
\end{aligned}
\]
Thus a proper split cannot increase $d_t+m_t$ for either child.\qedhere
\end{enumerate}
\end{proof}

\begin{lemma}[Completion]
\label{lem:completion}
If $\mathsf{RS}(\boldsymbol\alpha,k)^{\equiv d}
  \cap\mc{M}_{\mbf{V}}^\star(\F^d)\ne \{0\}$, then some branch in the reveal process realized by $\boldsymbol{\alpha}$ runs actively to completion with $S_n\ne \{0\}$ and $d_n+m_n>0$.
\end{lemma}

\begin{proof}
A nonzero matrix in the intersection has a polynomial lift satisfying
every revealed constraint.  At each branch,
\cref{prop:branching}(2) forwards this witness to a nonzero child, so that branch runs to completion as $S_t\ne \{0\}$ for every $t$. For $S_n\ne\{0\}$, there are no remaining costs, so 
\[ m_n+d_n\ge m_n\ge \frac{\varphi_n(\U_n)}{d_n}=\frac{\dim_{\F}S_n}{d_n}>0.\qedhere\]
\end{proof}

\subsection{Defect Probability Bounds}
\label{subsec:low-rate-lms}

By \cref{lem:initial,lem:completion}, conditioned on the event whose
probability is bounded in \cref{thm:low-lms}, the quantity $d_t+m_t$
starts negative with magnitude $\Omega(n)$ and must finish positively.
By \cref{prop:monovariant}, this can happen only if there are
$L=\Omega(n)$ marked steps, which suggests that the event is unlikely.

Since every mark is a defect, it suffices to bound the probability of a defect, where we recall from \cref{eq:def-defect} that reveal $t$ is a \emph{defect for $\W$} if
\[
  \dim_{\F}\{A_t(P)(\alpha_t):P\in S_{t-1}\cap\W\}
  <c_t(\W).\]

Our defect probability bound below improves on \cite{lms25} by using the presentation $(Z, E)$ of $S_{t-1}$ that we construct inductively in \cref{prop:monovariant}. Since $Z$ and $E$ contain $d$-tuples of plain polynomials of degree less than $k$, the symbolically nonzero determinant $f(X)$ will have $\deg f\le d(k-1)$, so the defect $f(\alpha_t)=0$ occurs with small probability.

We emphasize that our proof mechanism is identical to \cite{lms25}, but the presentation of $S_{t-1}$ gives better degree bounds on the polynomial involved. The corresponding estimate in \cite{lms25} tracks candidate and map
degrees, which starts at $d$ but may increase by a factor of $d$ at every quotient step. They obtain a bound with
$p_d\le {2d^{2d}(k+1)}/{q}$. 

\begin{proposition}
\label{prop:uniform-defect}
Condition on any active branch immediately before revealing $\alpha_t$, its state and ambient are $\mathscr F_{t-1}$-measurable, and
\begin{equation}
  \mb{P}\left(
    \text{reveal $t$ is a defect for $\U_{t-1}$}
    \mathrel{\big|}\mathscr F_{t-1}
  \right)
  \le p_d,
  \qquad
  p_d:=\frac{d(k-1)}{q-n}.
  \label{eq:uniform-defect-p}
\end{equation}
\end{proposition}
\begin{proof}
By \cref{prop:monovariant}, a presentation $(Z, E)$ with $\K$-spans $(\msc{Z}, \msc{E})$ exists for $S_{t-1}$. By the compatibility of the inherited constraints in \cref{prop:branching}(2), the current generic constraint is the map induced by the original $A_t$ on $\msc E/\msc Z$. Hence
\[ c:=c_t(\U_{t-1})
=\dim_\K A_t(\msc{E})-\dim_\K A_t(\msc{Z}).\]
Let $z:=\dim_\K A_t(\msc{Z})$. As $\msc{Z}$ is the $\K$-span of $Z$, we can choose $\K$-basis of $A_t(\msc{Z})$ given by $\{A_t(Q_1), \dots, A_t(Q_z)\}$ where $Q_i\in Z\subset \F[X]_{<k}^d$. We can extend this by $\{A_t(P_1), \dots, A_t(P_c)\}$ for $P_i\in E\subset \F[X]_{<k}^d$ to form a $\K$-basis of $A_t(\msc{E})$.  By $\K$-independence, putting these columns together gives a matrix with some invertible $(z+c)\times (z+c)$ minor, i.e. its determinant $f(X)$ is nonzero. 

Now, as $P_i, Q_i\in \F[X]_{<k}^d$, we see that
\begin{equation}
\label{eq:deg-f-bd} 
z+c=\dim_\K A_t(\msc{E})\le \dim_\K \msc{E} \le d\implies \deg f\le (z+c)(k-1)\le d(k-1).
\end{equation}

If $f(\alpha_t)\ne 0$, then the $z+c$ columns remain independent upon specializing to $X=\alpha_t$. Since specialization cannot increase ranks, we see that $\{A_t(P)(\alpha_t):P\in E\}$ and $\{A_t(P)(\alpha_t):P\in Z\}$ have $\F$-dimensions $z+c$ and $z$, respectively. Then, via the presentation $S_{t-1}\simeq E/Z$, the inherited specialized constraint has rank
\[
(z+c)-z=c=c_t(\U_{t-1}).
\]
Therefore, reveal $t$ is a defect for $\U_{t-1}$ only if $f(\alpha_t)=0$, which happens with probability at most $p_d$ by \cref{eq:deg-f-bd} and that $\alpha_t$ is uniform over the $q-t+1\ge q-n$ remaining field elements.
\end{proof}

By \cref{prop:uniform-defect}, any specific branch has at least $L$ marked defects with small probability. 
The remaining question is how to bound, at initialization, the probability that such a branch exists. 

In principle, branching can occur at every time step, for a total of $2^n$ leaves. The key observation is that the sum of dimensions $d_t$ over all branches cannot increase. Indeed, this is true trivially at continuation, and at branching it holds by \cref{prop:branching}(1). Thus, the number of nonzero terminal branches of the tree is at most $d$.
To union bound over them, we specify the branches with $d$ labels and union bound over ``the branch containing label $j\in [d]$'', so that each labeled path is progressively measurable.

\begin{proposition}
\label{prop:predictable-marks}
Assume every $\mathscr{F}_{t-1}$-measurable active branch is
marked at time $t$ with conditional probability at most $p$. Then, for every
$L\in\N$,
\begin{equation}
  \mb{P}\left(
    \text{some branch contains at least $L$ marks}
  \right)
  \le d\binom{n}{L}p^L.
  \label{eq:many-marks}
\end{equation}
\end{proposition}
\begin{proof}
We formally explain the labeling process: we assign every state in the process a subset of labels in $[d]$ with a deterministic rule such that 
\begin{itemize}
    \item The root has labels $[d]$,
    \item Children of a single parent state $S_t$ contain disjoint subsets of parent labels,
    \item Any state $S_t$ with $\K$-span $\U_t$ has $\dim_\K \U_t$ many labels.
\end{itemize}
At a non-branch step, assign the child a deterministically chosen subset of the parent's labels of the required size. When we branch, \cref{prop:branching} implies the dimensions of the two
children sum to the dimension of the parent. Thus every nonzero state carries at least one label, while each label follows at most one state. Hence, every nonzero terminal branch is identified by one of its labels.

Fix a label $a\in[d]$. Immediately before reveal $t$, the branch carrying $a$ is determined by $\alpha_1, \dots, \alpha_{t-1}$ and is therefore
$\mathscr F_{t-1}$-measurable. Let $X_{a,t}$ be the indicator that this branch is active and marked at time $t$, setting $X_{a,t}=0$ if the label has already been discarded. By assumption and iterated
conditioning, we have that 
\begin{equation}
  \mb{P}\left(X_{a,t}=1\mid\mathscr F_{t-1}\right)\le p \implies
  \mb{P}\left(X_{a,t}=1\text{ for every }t\in I\right)\le p^L.
\end{equation}
for any fixed set of time indices $I=\{t_1, \dots, t_L\}$ with $1\le t_1<\dots< t_L\le n$. A union bound over $I\in\binom{[n]}{L}$ and $a\in [d]$ gives \cref{eq:many-marks}.
\end{proof}

We now combine the previous ingredients to prove \cref{thm:low-lms}.

\begin{proof}[Proof of \cref{thm:low-lms}]
Along any surviving branch, if $N_{\mathrm{mark}}$ denotes its
number of marks, then
\[
  0<d_n+m_n\le d_0+m_0+N_{\mathrm{mark}}
  \le d-\eta n+N_{\mathrm{mark}},
\]
by \cref{lem:initial}. Thus $N_{\mathrm{mark}}\ge L$, where \begin{equation}
  L:=\lfloor \eta n-d\rfloor +1\ge\frac{\eta n}{2}>0.
  \label{eq:mark-threshold}
\end{equation}
At every mark, \cref{prop:monovariant} identifies a defect for an
ambient fixed before $\alpha_t$, and \cref{prop:uniform-defect}
bounds its conditional probability by $p_d$. If
$2ep_d/\eta\ge1$, the first bound in \cref{eq:relative-bound} is
trivial; hence assume $2ep_d/\eta<1$. Hence
\cref{prop:predictable-marks} and \cref{eq:mark-threshold} give
\[
  \mb{P}\left(\text{some branch satisfies }S_n\ne \{0\}\right)
  \le d\binom{n}{L}p_d^{L} \le d\left(\frac{enp_d}{L}\right)^{\eta n-d},
\]
which is at most the bounds in \cref{eq:relative-bound} after
substituting the definition of $p_d$ and using $d\le r$ and $d\le \eta n/2$.
By \cref{lem:completion}, the event that
$\mathsf{RS}(\boldsymbol\alpha,k)^{\equiv r}
    \cap\mc{M}_{\mbf{V}}^\star(U)\ne\{0\}$ is contained in the event that some branch runs to completion. Undoing the coordinate identification $U\simeq \F^d$ proves
\cref{eq:relative-bound}.
\end{proof}

\section{Non-containment to Containment for Nested Codes}
\label{sec:black-box}

In the previous section, the key \cref{thm:low-lms} showed that a random Reed--Solomon code is unlikely to contain a local pair whenever the potential has a uniform negative gap on every nonzero subspace.  We now prove the reverse direction for some \emph{nested codes}. All dimensions and subspaces are over $\F$.

\begin{definition}
    A family of $\F$-additive codes $C_1,\dots,C_k\subset(\F^s)^n$ is a
    $k$-nested code if $\dim C_i=i$ and $C_i\subset C_{i+1}$.
\end{definition}

For any additive code $C$, a nesting can be defined by repeatedly dropping one basis element at a time. Observe that codes such as Reed--Solomon and folded Reed--Solomon codes are naturally nested, motivating us to treat $\F$-additive codes with explicit folding $s$. 
Recall that, for an $r$-local profile $\mbf V=(V_1,\dots,V_n)$ over $\F$ and subspace
$U\subset\F^r$, code $C\subset \F^n$ contains $r$-local pair $(\mbf{V}, U)$ if and only if $C^{\equiv r}\cap \mc{M}_\mbf{V} (U) \ne \emptyset$.
For $C\subset(\F^s)^n$, we unfold each symbol into $s$ scalar
coordinates and impose $V_i$ on each of the $s$ scalar rows in block
$i$ (see \cref{subsec:sdc}).

Our goal of this section is the containment direction: if a nested code $C_k$ has rate $R:=k/sn >R_\mbf{V}(U)$, then it contains a witness matrix satisfying $\mbf{V}$ with rowspan exactly $U$. The overall plan is similar to the RLC case in \cite{lms25} and \cite{ggsw}:
\begin{enumerate}
    \item We first lower bound $\dim (C_k^{\equiv r}\cap \mc{M}^\star_\mbf{V} (U))\ge sn\phi_{\mbf{V}}(U, R)$ by a dimension counting argument.
    \item We then upper bound it by the same potential plus some error terms, and apply this to every proper subspace $X$ of $U$ to control $\dim (C_k^{\equiv r}\cap \mc{M}^\star_\mbf{V} (X))$ by $\phi_{\mbf{V}}(X, R)$
    \item If no witness matrix has rowspan exactly $U$, then all rowspans lie in proper subspaces $X$. Hence, $q^{sn\phi_{\mbf{V}}(U, R)}$ is at most the sum of the same term over each $X\subsetneq U$, plus some error.
    \item However, the potential at $U$ is bigger than that at any subspace $X\subsetneq U$ by an additive slack, see e.g. \cref{eq:potential-gap}, enough to absorb the loss in the union bound and errors in Step 3, a contradiction. Therefore, $C^{\equiv r}\cap \mc{M}_\mbf{V} (U) \ne \emptyset$ giving containment.
\end{enumerate}

Observe that the only code-specific statement is Step 2, which is morally a non-containment type statement. Our goal in this section is to turn
such black-box non-containment statements into containment at high rate for a nested family of codes. We write an abstract non-containment statement below. The quotient maps $\pi$ we consider here are exactly those in \cref{lem:maximal-quotient}.

\begin{definition}[Goodness]
\label{def:goodness}
Let $(\mathbf V,U)$ be an $r$-local pair and let
$(C_1,\ldots,C_k)$ be a $k$-nested family in $(\F^s)^n$. Set $C_0=\{0\}$. The family $(C_1,\ldots,C_k)$ is
$(\mathbf V,\gamma)$-good for $U$ if the following holds.

For every $\{0\}\neq X\subset U$ and surjective linear map
$\pi:\F^r\to\F^{r'}$ whose kernel is an inclusion-maximal element of
$\argmax_{W\subsetneq X}R_{\mathbf V}(X,W)$, we have
\[
C_{j_\gamma(X)}^{\equiv r'}\cap \mc{M}_{\pi \mbf{V}}^\star(\pi X) = \{0\}\quad\text{where}\quad j_\gamma(X):=
\min\left\{
k,
\max\left\{
0,
\left\lceil
\left(R_{\mathbf V}(X)-\gamma\right)sn
\right\rceil-1
\right\}
\right\}.\]
\end{definition}

We parse this definition and level $j_{\gamma}(X)$. The code $C_j$ has rate
$j/(sn)$, so $j_\gamma(X)$ is the largest level $j$ such that the rate of $C_j$ is $j/sn < R_{\mathbf V}(X)-\gamma$.  If the right hand side is negative, then
$j_\gamma(X)=0$ and the required condition is trivial; the
upper clipping merely stops at the final available code $C_k$.  Thus,
goodness says that, once we descend to the available level $j_\gamma(X)$
corresponding to the threshold of $(\mbf{V}, X)$ minus $\gamma$, every critical
quotient in the sense of \cref{lem:maximal-quotient} is
absent.

\subsection{Containment for Good Nested Codes}
We now give the key theorem---the black-box containment statement for good nested codes.

\begin{theorem}\label{thm:key}
Let $(\mathbf V,U)$ be an $r$-local pair with $U\ne\{0\}$ and
$\mathbf V\in\mc L(\F^r)^n$, and let
$C_1\subset \dots\subset C_k\subset(\F^s)^n$ be a $(\mathbf V,\gamma)$-good $k$-nested family of codes.  Then, $C_k$ contains
$(\mathbf V,U)$ whenever its rate
\begin{equation}
\label{eq:key-level}
    R:=\frac{k}{sn}>R_{\mbf{V}}(U)+\gamma r+\frac{2r}{sn}.
\end{equation}
\end{theorem}

We will prove this theorem with the four-step outline above. 
At each row-span level $X$, goodness requires us to descend to $C_{j_\gamma(X)}$ of rate
at most $R_{\mathbf V}(X)-\gamma$ and
then climbs back to $C_k$, which incurs arbitrarily small additive error in rate. Recall from
\cref{def:potential-threshold} that
\begin{equation}
  R_{\mathbf V}(U)=\max_{W\subsetneq U}R_{\mathbf V}(U,W),
  \label{eq:local-threshold-max}
\end{equation}
and define $R_{\mathbf V}(\{0\})=0$ by convention. 

\begin{definition}[Local Chains]
\label{def:local-chains}
    Given an $r$-local pair $(\mathbf V, U)$, define its \emph{local chain} as $(W_1, \dots, W_m)$, where $W_m = \{0\}$, $W_1 = U$, and $W_{i+1}$ is a maximal element of 
    \[ \argmax_{W\subsetneq W_i} R_{\mathbf V} (W_i, W).\] Define the corresponding \emph{local threshold chain} $(R_1, \dots, R_{m-1})$ by $R_i = R_{\mathbf V} (W_i, W_{i+1})$.
\end{definition}

To understand the definition and our proof strategy, it is helpful to first consider one level.  Start $W_1=U$ and choose $W_2\subsetneq W_1$ to be an
inclusion-maximal maximizer in \cref{eq:local-threshold-max}, so
\begin{equation}
  R_1:=R_{\mathbf V}(W_1,W_2)=R_{\mathbf V}(W_1).
  \label{eq:first-local-threshold}
\end{equation}
Note that $R_1$ is the rate at which the potential of $W_1$ and $W_2$ agree, i.e. $  \phi_{\mathbf V}(W_1,R_1)
  =\phi_{\mathbf V}(W_2,R_1)$.
In fact, at the lower rate $R_1-\gamma$, if $\pi_1$ is a quotient by $W_2$, \cref{lem:maximal-quotient} implies that 
\begin{equation}
  \phi_{\pi_1\mathbf V}(Y,R_1-\gamma)\le -\gamma \dim Y
\end{equation}
for every nonzero $Y\subset\pi_1W_1$.
Therefore, any matrix with rowspan $X\subset W_1$ but $X\not\subset W_2$ has a nonzero image under $\pi_1$, so goodness (or the relevant non-containment input) rules it out at $C_{j_\gamma(W_1)}$.  When we add dimensions and climb back to $C_k$, the corresponding dimension cost is upper bounded by
$(\dim W_1-\dim W_2)(k-j_\gamma(W_1))$, as required in Step 2.

This quotient does not see matrices whose row span $X\subset W_2$ as they map to zero. Now, $W_2$ replaces $W_1$ as the remaining subspace to handle, so we choose a critical subspace $W_3\subsetneq W_2$, descend to its own threshold level, and repeat. This is the motivation for \cref{def:local-chains}. The next lemma says that the threshold rates decrease along the chain, so the initial condition $R>R_1$ propagates down as we repeat the argument above.

\begin{proposition}
\label{prop:decreasing-rate}
    For any $r$-local pair $(\mathbf V, U)$ and its local threshold chain $(R_1, \dots, R_{m-1})$, 
    \[ R_1\ge R_2\ge \dots \ge R_{m-1}.\]
\end{proposition}
\begin{proof}
Let $(W_1, \dots, W_m)$ be the corresponding local chain.
    Fix any $i\in [m-2]$. Note that
    \begin{align*}
        \phi_{\mathbf V}(W_{i+2}, R_i)& \leq 
        \phi_{\mathbf V}(W_i, R_{i})\\
        & = \phi_{\mathbf V}(W_{i+1}, R_i)\\
        & =(R_{i}-R_{i+1})(\dim W_{i+1})+\phi_{\mathbf V}(W_{i+1}, R_{i+1})
        \\
        & = (R_{i}-R_{i+1})(\dim W_{i+1})+\phi_{\mathbf V}(W_{i+2}, R_{i+1})\\
        & = (R_i-R_{i+1})(\dim W_{i+1}-\dim W_{i+2}) + \phi_{\mathbf V}(W_{i+2}, R_i)
    \end{align*}
    Thus, the result follows.
\end{proof}

Before we formalize the argument, we record the following lemma about quotients. This is analogous to \cref{lem:simple-quotient} but at the level of $\F$-linear spaces $C^{\equiv r}\cap \mc{M}_{\mbf V}^\star(U)$ instead of $\phi_\mbf{V}(U,  R)$.

\begin{lemma}\label{lem:profile-quotients}
    For any additive code $C$ and $r$-local pair $(\mathbf V,U)$, the
    following statements hold:
    \begin{enumerate}
        \item Let $\pi:\F^r\to\F^{r'}$ be a surjective linear map such
        that $W=\ker\pi\subsetneq U$.  Then
        \[
        \dim\left(C^{\equiv r}\cap\mc M_{\mathbf V}^{\star}(U)\right)
        \le
        \dim\left(C^{\equiv r'}\cap\mc M_{\pi\mathbf V}^{\star}(\pi U)\right)
        +\dim\left(C^{\equiv r}\cap\mc M_{\mathbf V}^{\star}(W)\right).
        \]
        \item  Let $C'\subset C$ be an additive subcode.  Then
        \[
        \dim\left(C^{\equiv r}\cap\mc M_{\mathbf V}^{\star}(U)\right)
        \le
        \dim\left((C')^{\equiv r}\cap\mc M_{\mathbf V}^{\star}(U)\right)
        +(\dim C-\dim C')\dim U.
        \]
    \end{enumerate}
\end{lemma}
\begin{proof}
Let $\Im$ denote the image of a linear map.
For (1), define the natural map
    \[
    \pi':C^{\equiv r}\cap\mc M_{\mathbf V}^{\star}(U)
    \longrightarrow
    C^{\equiv r'}\cap\mc M_{\pi\mathbf V}^{\star}(\pi U).
    \]
   by  applying $\pi$ to the
    rows of each matrix. By rank-nullity, 
    \[
    \dim\left(C^{\equiv r}\cap\mc M_{\mathbf V}^{\star}(U)\right)
    =\dim\ker(\pi')+\dim\Im(\pi').
    \]
    Moreover,
    $\dim\Im(\pi')\le
    \dim(C^{\equiv r'}\cap\mc M_{\pi\mathbf V}^{\star}(\pi U))$
    and
    $\ker(\pi')=C^{\equiv r}\cap\mc M_{\mathbf V}^{\star}(W)$.

For (2), let $\sigma:C\to C/C'$ be the quotient map and define
\begin{equation}
\sigma':
C^{\equiv r}\cap\mc M_{\mathbf V}^{\star}(U)
\longrightarrow (C/C')^r,
\qquad
\sigma'(c_1,\ldots,c_r)
:=
(\sigma(c_1),\ldots,\sigma(c_r)).
\end{equation}
Then, its kernel is $(C')^{\equiv r}\cap\mc M_{\mathbf V}^{\star}(U)$. Moreover, $\sigma'$ preserves every linear relation among the columns.
Since the rows of $(c_1,\ldots,c_r)$ lie in $U$, the rows of
$\sigma'(c_1,\ldots,c_r)$ also lie in $U$.  The space of
$r$-tuples over $C/C'$ whose rows lie in $U$ has dimension
$\dim(C/C')\dim U$, so
\[
  \dim\Im (\sigma')
  \le(\dim C-\dim C')\dim U.
\]
The result follows from rank-nullity.
\end{proof}

The next lemma records simultaneously the lower and upper bound on $\dim\left(C_k^{\equiv r}\cap\mc M_{\mathbf V}^{\star}(U)\right)$ alluded to in Steps 1 and 2 of the proof outline. As a reminder, we will apply the lower bound to $U$ and the upper bound to proper subspaces of $U$.

\begin{lemma}
\label{lem:key}
Let $(\mathbf V,U)$ be an $r$-local pair with local chain $(W_1, \dots, W_m)$ and let
$C_1, \dots, C_k\subset(\F^s)^n$ be a
$(\mathbf V,\gamma)$-good $k$-nested family for $\gamma\ge0$ and $R:=k/(sn)$. For $t:=\min\left\{i\in[m]:R_{\mathbf V}(W_i)\le R\right\}$,
\begin{equation}
  sn\phi_{\mathbf V}(U,R)
  \le
  \dim\left(C_k^{\equiv r}\cap
       \mc M_{\mathbf V}^{\star}(U)\right)
  \le
  r+\gamma rsn+sn\phi_{\mathbf V}(W_t,R).
  \label{eq:solution-dimension-bounds}
\end{equation}
\end{lemma}
\begin{proof}
We first prove the lower bound by dimension counting.  Inside $C_k^{\equiv r}$, the
condition that every row lie in $U$ leaves a space of dimension
$k\dim U$.  For each $i\in[n]$, imposing that the $s$ corresponding rows in block $i$ each lie in $V_i$ costs at most
$s\on{codim}_{U}(U\cap V_i)$ linear constraints.  Therefore,
\[
\dim\left(C_k^{\equiv r}\cap\mc M_{\mathbf V}^{\star}(U)\right)
\ge
k\dim U-s\sum_{i=1}^n\on{codim}_{U}(U\cap V_i) 
=sn\phi_{\mathbf V}(U,R).
\]

For the upper bound, take the corresponding local threshold chain $(R_1, \dots, R_{m-1})$, so that
\begin{equation}
\label{eq:Riji}
R_i:=R_{\mathbf V}(W_i,W_{i+1})
    =R_{\mathbf V}(W_i)\qquad 
j_i:=j_\gamma(W_i)=\min(k,\max(0,\lceil(R_{\mathbf V}(W_i)-\gamma)sn\rceil-1)).
\end{equation}
Choose surjective linear map $\pi_i:\F^r\to\F^{r'}$ with kernel $W_{i+1}$, so that $r'=r-\dim W_{i+1}$. Applying \cref{lem:profile-quotients} first to the profile quotient
and then to the inclusion $C_{j_i}\subset C_k$ gives
\begin{equation}
\label{eq:key-one-level}
\begin{aligned}
&\dim\left(C_k^{\equiv r}\cap
           \mc M_{\mathbf V}^{\star}(W_i)\right) \\
\le & 
\dim\left(C_k^{\equiv r'}\cap
          \mc M_{\pi_i\mathbf V}^{\star}(\pi_iW_i)\right)
+\dim\left(C_k^{\equiv r}\cap
          \mc M_{\mathbf V}^{\star}(W_{i+1})\right) \\
\le &
\left(\dim W_i-\dim W_{i+1}\right)(k-j_i)+\dim \left(C_{j_i}^{\equiv r'}\cap
          \mc M_{\pi_i\mathbf V}^{\star}(\pi_iW_i)\right)
+\dim\left(C_k^{\equiv r}\cap
          \mc M_{\mathbf V}^{\star}(W_{i+1})\right).
\end{aligned}
\end{equation}
The middle term is zero by goodness of the nested code. Iterating \cref{eq:key-one-level} gives
\begin{equation}
\begin{aligned}
\label{eq:key-chain-sum}
\dim\left(C_k^{\equiv r}\cap\mc M_{\mathbf V}^{\star}(U)\right)
 & \le
\sum_{i=1}^{m-1}
\left(\dim W_i-\dim W_{i+1}\right)(k-j_i)\\
& \le r+\gamma rsn +
sn\sum_{i=1}^{m-1}
\left(\dim W_i-\dim W_{i+1}\right)
\max\{0,R-R_i\},
\end{aligned}
\end{equation}
where we check that $k-j_i \le
1+\gamma sn+sn\max\{0,R-R_i\}$ for each $i$, and that $W_m=\{0\}$ so
 \[
\sum_{i=1}^{m-1}
\left(\dim W_i-\dim W_{i+1}\right)
 =\dim U\le r.
\]

The thresholds $R_i$ are nonincreasing by \cref{prop:decreasing-rate}.  By the definition of $t$, we have $R_i>R$ for $i<t$ and $R_i\le R$ for $i\ge t$.  Let $R_{t-1}':=R$ and $R_i':=R_i$ for $i\ge t$. The sum in \cref{eq:key-chain-sum} telescopes:
\begin{align*}
\sum_{i=1}^{m-1}
\left(\dim W_i-\dim W_{i+1}\right)
\max\{0,R-R_i\} & = \sum_{i=t}^{m-1}
\left(\dim W_i-\dim W_{i+1}\right)(R-R_i)
\\ & = \sum_{i=t}^{m-1}
\dim W_i\left(R_{i-1}'-R_i'\right)
\\ & =\sum_{i=t}^{m-1}
\left(
\phi_{\mathbf V}(W_i,R_{i-1}')
-\phi_{\mathbf V}(W_i,R_i')
\right)\\
&=\sum_{i=t}^{m-1}
\left(
\phi_{\mathbf V}(W_i,R_{i-1}')
-\phi_{\mathbf V}(W_{i+1},R_i')
\right)\\
&=\phi_{\mathbf V}(W_t,R).
\end{align*}
where crucially in the penultimate step we use the fact that the local chain definition satisfies
\[\phi_{\mathbf V}(W_i,R_i)
=
\phi_{\mathbf V}(W_{i+1},R_i),\]
similar to \cref{eq:Riji}. Combining this bound with \cref{eq:key-chain-sum} gives the upper bound.
\end{proof}
We are now ready to prove \cref{thm:key}, executing the final Step 4 of the proof outline. We recall from\cref{eq:potential-gap} that
\[ 
\phi_{\mathbf V}(U,R)-\phi_{\mathbf V}(W,R)
=
(\dim U-\dim W)\left(R-R_{\mathbf V}(U,W)\right),
\]
for every $W\subsetneq U$. This is \cref{eq:potential-gap} and can be seen directly from the definitions of $\phi$ and $R_{\mbf{V}}(U)$.
\begin{proof}[Proof of \cref{thm:key}]
Let $\eta:=R-R_{\mbf{V}}(U)$. By the display above and \cref{eq:key-level}, we see that
\[ 
\phi_{\mathbf V}(U,R)-\phi_{\mathbf V}(W,R)\ge \eta >\gamma r+\frac{2r}{sn},
\]
for every $W\subsetneq U$. By the lower bound in \cref{lem:key}, we see that
\[
\dim\left(C_k^{\equiv r}\cap\mc M_{\mathbf V}^{\star}(U)\right)
\ge sn\phi_{\mathbf V}(U,R).
\]
By the upper bound in
\cref{lem:key} applied to $W$, we see that
\[
\dim\left(C_k^{\equiv r}\cap\mc M_{\mathbf V}^{\star}(W)\right)
\le
r+\gamma rsn+sn\phi_{\mathbf V}(W_t,R)
\le r+\gamma rsn+sn\phi_{\mathbf V}(U,R)-sn\eta
<sn\phi_{\mathbf V}(U,R)-r,
\]
where $W_t\subset W\subsetneq U$ is the subspace selected by the
local chain for $W$, and the middle inequality uses the potential gap
above.
Now, suppose that $C_k$ does not contain $(\mbf{V}, U)$,  so every matrix in
$C_k^{\equiv r}\cap\mc M_{\mathbf V}^{\star}(U)$ has row span
contained in some codimension-one subspace of $U$. There are at most $q^r$ codimension-one subspaces, so combining the two displays above with the union bound, we have that
\[
    q^{sn\phi_{\mathbf V}(U, k/sn)}
    \leq \left|C_k^{\equiv r}\cap\mc M_{\mathbf V}^{\star}(U)\right|\le \sum_{\substack{W\subsetneq U\\ \dim U/W=1}}
      \left|C_k^{\equiv r}\cap\mc M_{\mathbf V}^{\star}(W)\right| < q^r\cdot q^{sn\phi_{\mathbf V}(U,R) -r}.
\]
This is a contradiction, so $C_k$ contains $(\mbf{V}, U)$, and we are done!
\end{proof}

\subsection{Local Equivalence for Random Reed-Solomon Codes}

We now use non-containment direction (\cref{thm:low-lms}) to show random Reed-Solomon codes are good with high probability, thereby concluding the containment direction, and thus the local equivalence theorem refinement of \cite[Theorem 3.1]{lms25}. This is stated informally in \cref{sec:intro} as \cref{cor:rrs-equiv}. 

We recall the canonical nesting of Reed-Solomon codes $\mathsf{RS}(\boldsymbol{\alpha}, j)$ with $j\in [k]$, where evaluation points $\boldsymbol{\alpha}:=(\alpha_1, \dots, \alpha_n)$ are sampled uniformly from the ordered $n$-tuples of distinct elements of $\F$ of size $q$. Since it is a $\F$-linear code, the folding parameter $s=1$ in this section.

\begin{theorem}[Goodness of RRS]\label{thm:RRS-is-good}
Let $(\mbf V,U)$ be an $r$-local pair with $U\ne\{0\}$ and let $k\in[n]$. Suppose that $q:=|\F|>\max\{n, r(k-1)\}$ and $\gamma \ge 2r/n$, then
\begin{equation}
\label{eq:goodness-rrs-1}
\mb{P}_{\boldsymbol{\alpha}}\left(
  (\mathsf{RS}(\boldsymbol\alpha,1),\dots,
   \mathsf{RS}(\boldsymbol\alpha,k))
  \text{ is }(\mathbf V,\gamma)\text{-good for }U
\right)
\ge 1-q^{2r^2}r\left(
    \frac{2erk}{\gamma(q-n)}
  \right)^{\gamma n/2}.
\end{equation}
Moreover, $\mathsf{RS}(\boldsymbol{\alpha}, k)$ contains $(\mbf{V}, U)$ with at least the same probability if rate 
\begin{equation}
\label{eq:goodness-rrs-2}
R:=\frac{k}{n}>R_{\mbf{V}}(U)+\gamma r+\frac{2r}{n}.
\end{equation}
\end{theorem}

\begin{proof}
For goodness, we unfold \cref{def:goodness}: Fix $\{0\}\ne X\subset U$, an inclusion-maximal
\[ W\in\argmax_{Y\subsetneq X}R_{\mbf V}(X,Y),\] and a quotient map
$\pi:\F^r\to\F^{r'}$ with kernel $W$.  Let
$j:=j_\gamma(X)$. It suffices to check that 
\begin{equation}
\label{eq:goodness-rrs-check}
C_j^{\equiv r'}\cap \mc{M}_{\pi\mbf{V}}^\star (\pi X) =\{0\}
\end{equation}
with high probability. If $j=0$, this is trivial. By \cref{lem:maximal-quotient} applied to $(\mbf{V}, X)$ and $\pi$, we see that
\[   \frac{j}{n}\le R_{\mbf V}(X)-\gamma
  =R_{\mbf V}(X,W)-\gamma\implies \phi_{\pi\mbf V}\left(Y,\frac{j}{n}\right)
  \le-\gamma\dim_{\F}Y,\]
for every nonzero $Y\subset \pi X$. This checks the condition in \cref{thm:low-lms}. Applying \cref{thm:low-lms}, \cref{eq:goodness-rrs-check} fails with sufficiently small probability that we obtain \cref{eq:goodness-rrs-1} upon union bounding over the $q^{2r^2}$ many subspaces $(X, W)$ of $\F^r$. On the event that the nested code is $(\mbf{V}, \gamma)$-good, \cref{eq:goodness-rrs-2} and \cref{thm:key} implies that the final code $\mathsf{RS}(\boldsymbol{\alpha}, k)$ contains $(\mbf{V}, U)$, as desired.
\end{proof}

We usually think of $r,\gamma >0$ as any constant and take $n$ (and therefore $k, q$) sufficiently large. This is the case in for example \cref{cor:rrs-equiv}. We spell out the exact parameter dependence below. 

\begin{theorem}[Fixed-Profile Equivalence for RRS]
\label{thm:rrs-equiv}
Let $(\mbf V,U)\in\mc{L}(\F^r)^{n}\times \mc{L}(\F^r)$ be an $r$-local profile, and let $C:=\mathsf{RS}(\boldsymbol{\alpha}, k)$ of rate $R=k/n$. For $\eta\in (0, 1]$ and $q:=|\F|>\max\{n, r(k-1)\}$,
\begin{enumerate}
\item If $R\le R_{\mbf V}(U)-\eta$ and $\eta n\ge2r$, then
\begin{equation}
  \mb{P}\left(C\text{ contains }(\mbf V,U)\right)
  \le
r
  \left(
    \frac{2e r k}{\eta(q-n)}
  \right)^{\eta n/2}.
  \label{eq:rrs-threshold-low}
\end{equation}

\item If $R>R_{\mbf V}(U)+\eta$ and $\eta n\ge2r(r+1)$, then
\begin{equation}
  \mb{P}\left(C\text{ contains }(\mbf V,U)\right)
  \ge 1-q^{2r^2}r
  \left(
    \frac{2e r k}
         {(\eta/r-2/n)(q-n)}
  \right)^{\frac{\eta n}{2r}-1}.
  \label{eq:rrs-threshold-high}
\end{equation}
\end{enumerate}
\end{theorem}

\begin{proof}
For (1), since
$R\le R_{\mbf V}(U)-\eta$,
{\cref{lem:maximal-quotient,lem:simple-quotient}} gives a quotient pair
$(\pi\mbf V,\pi U)$ that is contained in $C$ whenever $C$ contains $(\mbf V,U)$, and that for every nonzero $W\subset \pi U$,
\[
  \phi_{\pi\mbf V}(W,R)\le-\eta\dim_{\F}W.
\]
This checks the condition in \cref{thm:low-lms}, so \cref{eq:rrs-threshold-low} follows from $\dim \pi U \le r\le \eta n/2$. 
For (2), let
\[ 
\gamma:=\frac{\eta}{r}-\frac{2}{n}.
\]
Observe that $\eta n\ge2r(r+1)$ implies that
$0<\gamma\le1$ and $\gamma n\ge2r$, while
\[
  R_{\mbf V}(U)+\gamma r+\frac{2r}{n}=
    R_{\mbf V}(U)+\eta<R.
\]
Thus, \cref{thm:RRS-is-good} applies, and the desired bound follows upon substituting our choice of $\gamma$.
\end{proof}

Recall that a code $C$ contains a row-span constrained $r$-local family $\mc{F}$ if $C$ contains some $r$-local pair $(\mbf{V}, U)\in\mc{F}$, and threshold rate 
$R_\mc{F} = \min\{R_{\mbf{V}}(U): (\mbf{V}, U)\in\mc{F}\}$.
The following theorem is immediate from this definition upon taking a union bound over the pairs in $\mc F$
in \cref{eq:rrs-threshold-low}.

\rrsequiv*

\subsection{Local Equivalence for Subspace Design Codes}
Recall from \cref{def:sdc} subspace design codes, where the non-containment theorems (e.g. \cref{thm:sdc-threshold-pair}), and therefore goodness. We generalize to nested codes.

\begin{definition}[Optimal Nested Subspace Design Codes]
\label{def:slacked-sdc-nesting}
Let $\gamma\ge0$. A $k$-nested code
$(C_1,\ldots,C_k)$ with $C_i\subset(\F^s)^n$ is a
\emph{$(\gamma, r)$-optimal nested subspace design code}, if for every $j\in [k]$, that
\[
  \tau_j(r)\le\frac{j}{sn}+\gamma.
\]
\end{definition}

At a single level $j$, this means that $C_j$ of rate $j/sn$ is $\gamma$-slacked
$r$-subspace designable in the sense of \cite{bcdz25b}. We simply rephrase it in terms of \cref{def:sdc} and extend to nested codes. The next proposition says that such subspace design codes are good, with equal slack $\gamma$ in subspace design as goodness. This is the deterministic counterpart of \cref{thm:RRS-is-good}.

\begin{theorem}[Goodness of SDC]
\label{thm:sdc-good}
Let $(\mbf V,U)$ be an
$r$-local pair with $U\ne\{0\}$ and let $\gamma \ge 0$. Then, any $(\gamma, r)$-optimal nested subspace design code $(C_1,\ldots,C_k)$ is
$(\mbf V,\gamma)$-good for $U$.
\end{theorem}
\begin{proof}
We unfold \cref{def:goodness}: Fix $\{0\}\ne X\subset U$, an inclusion-maximal
\[ W\in\argmax_{Y\subsetneq X}R_{\mbf V}(X,Y),\] and a quotient map
$\pi:\F^r\to\F^{r'}$ with kernel $W$.  Let
$j:=j_\gamma(X)$. It suffices to check that 
\begin{equation}
\label{eq:goodness-sdc-check}
C_j^{\equiv r'}\cap \mc{M}_{\pi\mbf{V}}^\star (\pi X) =\{0\}.
\end{equation}
If $j=0$, this is trivial. Applying definition of $j=j_\gamma(X)$ and \cref{lem:maximal-quotient} with $(\mbf{V}, X)$ and $\pi$ gives
\[   \frac{j}{sn}< R_{\mbf V}(X)-\gamma
  \le R_{\pi\mbf{V}}(Y)-\gamma\]
for every nonzero $Y\subset \pi X$. Since $\dim Y\le r$ and $\tau_j$ is nondecreasing,
\[
{\tau_j(r')\le\tau_j(r)}
 \le \frac{j}{sn}+\gamma<R_{\pi\mbf V}(Y),
\]
by the $\gamma$-slacked nested subspace design property.
Thus \cref{thm:sdc-threshold-pair}(1), applied to
$(\pi\mbf V,Y)$, rules out every possible nonzero row span $Y$
inside $\pi X$, which implies \cref{eq:goodness-sdc-check} and therefore goodness.
\end{proof}

We immediately obtain the following equivalence theorems for subspace design codes, first for fixed profile pairs, then for row-span constrained LCL properties.

\begin{theorem}[Fixed-Pair Equivalence for SDC]
\label{thm:sdc-equiv}
Let $(\mbf{V}, U)$ be an $r$-local pair and let $\gamma \ge 0$.
For a $(\gamma, r)$-optimal nested subspace design code $(C_1,\ldots,C_k)$, define rate and slack to be
\[
  R:=\frac{k}{sn}\qquad\text{and}\quad \eta := r\gamma+\frac{2r}{sn}.
\]
Then the following hold.
\begin{enumerate}
\item If $R< R_{\mbf{V}}(U)-\gamma$, then $C_k$ avoids $(\mbf{V}, U)$.
\item If $R> R_{\mbf{V}}(U)+\eta$, then $C_k$ contains $(\mbf{V}, U)$.
\end{enumerate}
\end{theorem}

We note that the non-containment is exactly \cref{thm:sdc-threshold-pair}(1); the containment direction follows from \cref{thm:key} where goodness is checked in \cref{thm:sdc-good}. We have now proven \cref{cor:sdc-equiv}.

\sdcequiv*

The containment direction of \cref{thm:sdc-threshold-pair} proven in \cite{ggsw} constructs
some subspace-designable code (via RLC's) containing the prescribed local property. In contrast, this theorem applies to every suitable nesting. We next apply
it to folded Reed--Solomon codes (see \cref{def: FRS Codes}) with the natural nesting. The containment direction of this theorem is therefore novel. 
By \cref{prop:unfolding}, the threshold itself is unchanged by
folding, so the following statement compares the folded construction
to the same pair threshold.

\begin{theorem}[Fixed-Pair Equivalence for FRS]
\label{thm:frs-equiv}
For any $\eta\in (0, 1]$, let $(\mbf{V}, U)$ be an $r$-local pair, and let
$C$ be an $s$-folded Reed-Solomon code of dimension $k$ with
rate $R:=k/(sn)$ and folding parameter $s\ge 4r^2/\eta$. Then
\begin{enumerate}
\item If $R\le R_{\mbf V}(U)-\eta/r$, then $C$ avoids $(\mbf{V}, U)$.
\item If $R\ge R_{\mbf V}(U)+\eta$, then $C$ contains $(\mbf{V}, U)$.
\end{enumerate}
\end{theorem}
\begin{proof}
Let $\{0\}=:C_0\subset C_1\subset \dots \subset C_k=C$ be the natural nesting of folded Reed-Solomon codes where $C_j$ contains evaluations of polynomials of degree less than $j$ (see \cref{def: FRS Codes}). Let
\[ \gamma:=\frac{r-1}{s-r+1}.\]

For every $j\in [k]$, $C_j$ has rate $R_j=j/sn$.
As $s\ge 4r^2/\eta\ge r$, \cref{thm: GKS16 subspace design} implies that
\[ 
\tau_j(r) \le \frac{sR_j}{s-r+1} = R_j+\frac{(r-1)R_j}{s-r+1}\le R_j+\gamma.
\]
Thus $(C_1,\ldots,C_k)$ is a $\gamma$-slacked
$r$-subspace-designable nesting.
Moreover, as $s\ge 4r^2/\eta$, we get
\[
  r\gamma+\frac{2r}{sn}
  =
  \frac{r(r-1)}{s-r+1}+\frac{2r}{sn} <
  \frac{2r^2}{s}+\frac{2r}{s}
  \le \eta.
\]
The two conclusions now follow immediately from
\cref{thm:sdc-equiv}.

\end{proof}

\begin{corollary}[Local Equivalence for FRS]
\label{cor:frs-equiv}
For any $\eta\in (0, 1]$, let $\mc{F}$ be row-span constrained $r$-local family, and let
$C$ be an $s$-folded Reed-Solomon code of dimension $k$ with
rate $R:=k/(sn)$ and folding parameter $s\ge 4r^2/\eta$. Then
\begin{enumerate}
\item If $R\le R_{\mc F}-{\eta/r}$, then $C$ avoids $\mc F$.
\item If $R\ge R_{\mc F}+\eta$, then $C$ contains $\mc F$.
\end{enumerate}
\end{corollary}
This follows from \cref{thm:frs-equiv} by applying the fixed-pair
statement to each $(\mbf V,U)\in\mc F$ and using the definition of
$R_{\mc F}$.

\begin{remark}
    We note that the only two other constructions of explicit subspace design codes are given by Univariate Multiplicity codes \cite{GK16}, and the AEL based codes constructed by \cite{ggh26}. Univariate multiplicity codes have a nesting identical to the one for folded Reed-Solomon codes and thus the same local equivalence applies to them as well. 

    For \cite{ggh26}, a natural nesting directly works as well, since their construction is based on the AEL method where they show that a rate $1-\varepsilon$ code with $\mathsf{poly}(\varepsilon)$ distance, an inner subspace design code of rate $R$, $\tau(r)\leq R+\varepsilon$, and dimension $d\ge O_{q, r, 1/\varepsilon}(1)$, and a Ramanujan degree-$d$ spectral expander for the AEL method leads to a deterministic proof of the resulting code being a near optimal subspace design. 

    They then brute force over a choice of the inner subspace design code to find it, which exists since a random additive inner code suffices. The same observation occurs here that an inner random nested code suffices to get the nested $(\gamma, r)$-optimal nested subspace design code. These would then give a nested family of codes with gaps in dimension which can be filled in by arbitrary dimension additions. We do not present all details here as we believe it is out of scope.
\end{remark}

\section{Compressing the Curve Decoding Profile Family}
\label{sec:profile-counting}
By \cref{thm:ggsw-family}, a code $C$ fails $(\ell,\delta,a,b)$-curve decodability
precisely when it contains a profile in
$\mc{F}_{\ell,\delta,a,b}$. In this section, we treat $(\ell,\delta,a,b)$ as fixed, with choices of parameters in \cref{eq:ggsw-param}, and abbreviate the family as $\mc{F}$. We partition the family based on $(A, U)$.

\begin{definition}[Curve profile family]
\label{def:curve-profile-family}
Fix $A\in\binom{\F}{a}$ and an $(\ell,A,b)$-curve-free subspace
$U\subset\F^A$.  Let $\mc F_{A,U}$ consist of all pairs
$(\mbf V,U)$ for which there exist sets
$A_1,\ldots,A_n\subset A$ such that
\begin{itemize}
\item $V_i=U\cap K_A(A_i)$ for all $i\in[n]$;
\item $|\{i\in[n]:\alpha\in A_i\}|\ge(1-\delta)n$ for all
$\alpha\in A$.
\end{itemize}
Then
\begin{equation}
  \mc{F}
  =
  \bigcup_{ A\in\binom{\F}{a}}\bigcup_{U\subset\F^A\text{ is }(\ell,A,b)\text{-curve-free}}
  \mc F_{A,U}.
  \label{eq:curve-profile-decomposition}
\end{equation}
\end{definition}

There are at most $q^{a+a^2}$ choices for $(A,U)$, and \cref{thm:ggsw-family} bounds each $|\mc{F}_{A, U}|\le 2^{an}$
by counting the raw agreement sets $(A_1,\ldots,A_n)$. The latter is the dominant contribution to 
$|\mc{F}|$ and the bottleneck for the alphabet-size bounds. 
The main observation in this section is that one should instead count
the distinct subspaces $V_i=U\cap K_A(A_i)$
that arise from these sets, since many different choices of $A_i$ induce
the same subspace $V_i$. This observation was first made for list decoding in \cite{gg26}; we refine their counting technique slightly further for curve decoding.

We construct two certificate families,
$\mc F^{\on{high}}$ and $\mc F^{\on{low}}$.  If $C$ contains a
profile $(\mbf V,U)\in\mc{F}$ witnessed by a matrix
$M$, then it contains a profile in $\mc F^{\on{high}}$ when
$d:=\dim_{\F}U=\on{rank}M\ge D$, and it contains a profile in
$\mc F^{\on{low}}$ when $d\le D$. We set the dimension cutoff to be
\begin{equation}
  D:=\ceil{\frac{\ell+1}{\eta}}.
  \label{eq:def-D}
\end{equation}

In \cref{subsec:high-rank}, we construct
$\mc F^{\on{high}}$ and show that all its profiles have negative potential.
In \cref{subsec:low-rank}, we construct
$\mc F^{\on{low}}$, where the family is much smaller.
In \cref{subsec:rlc-count}, we give the codimension-sensitive count on $\mc F^{\on{low}}$ used for RLC and RLDPC. In \cref{sec:alphabet}, we combine the same count with the gap-sensitive estimate for RRS.

\subsection{High-Rank Family}
\label{subsec:high-rank}

For a high-rank witness matrix $M$, we retain $D$ independent columns.  The row
span of the resulting matrix is the full space on those $D$ coordinates. The choice of $D$ in \cref{eq:def-D} will be clear in
\cref{lem:high-rank-family-bound}: we want the potential gap $\eta D$ to dominate the loss of $\ell+1$. 

\begin{definition}[High-rank family]
\label{def:high-rank-family}
Fix $S\in\binom{\F}{D}$ and let
$\mc F^{\on{high}}_{S}$ consist of all pairs $(\mbf W,\F^S)$ for which
there exist sets $S_1,\ldots,S_n\subset S$ such that
\begin{enumerate}
\item $W_i=K_S(S_i)$ for all $i\in[n]$;
\item $|\{i\in[n]:\alpha\in S_i\}|\ge(1-\delta)n$ for all
$\alpha\in S$.
\end{enumerate}
Then, the \emph{high rank family} is defined as
\begin{equation}
  \mc F^{\on{high}}:=
  \bigcup_{S\in\binom{\F}{D}}
  \mc F^{\on{high}}_{S}.
  \label{eq:high-family-decomposition}
\end{equation}
\end{definition}

The family is small, and every one of its profiles has
negative potential.

\begin{lemma}
\label{lem:high-rank-family-bound}
If $R\le1-\delta-2\eta$, then $\phi_{\mbf W}(\F^S,R)\le-\eta D$ for every 
$(\mbf W,\F^S)\in\mc F^{\on{high}}$.
Moreover,
\begin{equation}
  |\mc F^{\on{high}}|
  \le\binom{q}{D}2^{Dn}.
  \label{eq:high-rank-family-size}
\end{equation}
\end{lemma}

\begin{proof}
By the coverage condition, $\sum_{i=1}^n|S_i|\ge(1-\delta)Dn$. Moreover, for every $i\in[n]$,
\[
  \on{codim}_{\F^S}W_i = \on{codim}_{\F^S}K_S(S_i)
  =|S_i|-\dim_\F K_{S_i}(S_i) \ge |S_i|-\ell-1.
\]
As $\ell+1\le\eta D$ by our choice of $D$ and $R\le 1-\delta-2\eta$, we see that 
\[
  \sum_{i=1}^n\on{codim}_{\F^S}W_i
  \ge \sum_{i=1}^n|S_i|-(\ell+1)n
  \ge (1-\delta)Dn-(\ell+1)n\ge (1-\delta-\eta)Dn\ge (R+\eta)Dn. 
\]
Therefore, we have that
\[
  \phi_{\mbf W}(\F^S,R)
  =RD-\mb{E}_{i\in [n]}\on{codim}_{\F^S}W_i
  \le -\eta D.
\]
For \cref{eq:high-rank-family-size}, there are $\binom qD$ choices for $S$ and at most $2^D$
choices for each $S_i$.
\end{proof}

We next connect this certificate family to the original curve-profile
family.

\begin{lemma}
\label{lem:high-rank}
Suppose a code $C$ contains a profile
$(\mbf V,U)\in\mc F_{A,U}$ witnessed by a matrix $M$, and suppose
$\dim_{\F}U=\on{rank}M\ge D$.  Then there is a set $S\subset A$
indexing $D$ independent columns of $M$ such that $M|_S$ witnesses
containment of a profile in
$\mc F^{\on{high}}_{S}$.
\end{lemma}
\begin{proof}
Choose $D$ linearly independent columns of $M$ indexed by $S\subset A$. Let $M\vert_S$ be the submatrix, so $\on{rowspan}_\F(M\vert_S) = \F^S$. Since $M$ witnesses containment of $\mc{F}_{A, U}$, there exist $A_1, \dots, A_n\subset A$ such that for every $\alpha\in A$, $|\{i\in [n]:\alpha\in A_i\}|\ge (1-\delta)n$.
Moreover, as $S_i=S\cap A_i$,
\[M_{i\star}\in U\cap K_A(A_i)\implies (M\vert_S)_{i\star}\in K_S(S_i)=:W_i. \]
Now, for every $\alpha\in S$
\[|\{i\in [n]:\alpha\in S_i\}|=|\{i\in [n]:\alpha\in A_i\}|\ge (1-\delta)n.\]
Therefore, $M\vert_S$ witnesses the containment of $(\mbf{W}, \F^S)\in \mc{F}_S^{\on{high}}$.
\end{proof}
\subsection{Low-Rank Family}
\label{subsec:low-rank}
We now consider profiles whose row space has dimension
$d:=\dim_{\F}U\le D$. We are considering $(\ell, \delta, a, b)$-curve decodability. As in \cref{def:curve-profile-family,eq:curve-profile-decomposition}, we have \[\mathcal F = \bigcup_{ A\in\binom{\F}{a}}\bigcup_{U\subset\F^A\text{ is }(\ell,A,b)\text{-curve-free}}
  \mc F_{A,U}.\]
Fixing such a profile, we have $A \in \binom{\F}{a}$, an $(\ell, A, b)$-curve-free subspace $U\subset \F^A$ with $\dim U \le D$, sets $A_1, \dots, A_n \subset A$, $V_i = U\cap K_{A}(A_i)$, and $|\{i\in [n]\mid \alpha \in A_i\}|\geq (1-\delta)n$ for all $\alpha \in A$. Our goal is to produce a profile whose potential is negative. Our strategy is analogous to \cite{Sri25,AHS26, gg25a}, but at the level of local profiles themselves: we prune down to a subset $A'$ where each $A_i'$ imposes many independent polynomial constraints imposed on $U$, formally defined below.

\begin{definition}[Constraint codimension]
\label{def:constraint-codimension}
Fix a pair $(A,U)$ with $U\subset\F^A$. For $S\subset A$, define
\begin{equation}
c_{A, U}(S):=\on{codim}_{U}\left(U\cap K_A(S)\right).
  \label{eq:def-constraint-codimension}
\end{equation}
\end{definition}

We begin with some observations that $c(S)$ is preserved under two transformations of the ambient space: first, it is preserved when we prune the ambient coordinates; second, it is preserved under a \emph{curve translation} by a polynomial interpolating a fixed set of coordinates.
\begin{lemma}
\label{lem:c-restriction}
Fix a pair $(A,U)$ with $U\subset\F^A$.  If $S\subset A'\subset A$, then $c_{A, U}(S)=c_{A', U|_{A'}}(S)$.
\end{lemma}
\begin{proof}
Consider the linear maps $\psi:U\to\F^S/K_S(S)$ and $\psi':U\vert_{A'}\to\F^S/K_S(S)$ both given by $u\mapsto u\vert_S+K_S(S)$. As $S\subset A'$, they have the same image. By rank-nullity, its dimension is $c_{A, U}(S)$ and also $c_{A', U\vert_{A'}}(S)$, so they must be equal.
\end{proof}

\begin{lemma}[Curve Translation]
\label{lem:c-translation}
For any $A\in\binom{\F}{\ge\ell+1}$ and $ U\subset\F^{A}$, fix $J\in\binom{A}{\ell+1}$.  For each $u\in U$, let $p_u\in\F[X]_{\le\ell}$ be the unique polynomial agreeing with $u$ on $J$. Let $\psi(u):=u-p_u|_{A}$ and let $U':=\psi(U)$.  Then, for every $S\subset A$,
\begin{equation}
  c_{A,U}(S)=c_{A, U'}(S).
\end{equation}
Moreover, $U'\cap K_{A}(A)=\{0\}$, so $\dim_{\F}U'=c_{A,U'}(A)=c_{A,U}(A)$.  If $U$ is $(\ell,A,b)$-curve-free, then so is $U'$.
\end{lemma}
\begin{proof}
Interpolation on $J$ is linear, so $\psi$ is linear. Put $P:=U\cap K_A(A)$. Then $\ker\psi=P$: indeed, $\psi(u)=0$ exactly when $u=p_u|_A$ is the evaluation of a polynomial of degree at most $\ell$. Fix $S\subset A$. Since $p_u|_A\in K_A(S)$ and $K_A(S)$ is linear,
\begin{equation}
u\in K_A(S)\quad\Longleftrightarrow\quad \psi(u)\in K_A(S).
\label{eq:translation-membership}
\end{equation}
Consequently, $\psi(U\cap K_A(S))=U'\cap K_A(S)$. Moreover, $P\subset K_A(S)$, so the restriction of $\psi$ to $U\cap K_A(S)$ also has kernel $P$. Rank--nullity gives
\begin{equation}
\dim U'=\dim U-\dim P,
\qquad
\dim(U'\cap K_A(S))=\dim(U\cap K_A(S))-\dim P.
\end{equation}
Subtracting proves $c_{A,U'}(S)=c_{A,U}(S)$. Every vector in $U'$ vanishes on $J$. If it also lies in $K_A(A)$, it is the evaluation of a degree-at-most-$\ell$ polynomial with the $\ell+1$ elements of $J$ as roots, and hence is zero. Thus $U'\cap K_A(A)=\{0\}$, and the asserted dimension identity follows by taking $S=A$. Finally, if $U'$ were not $(\ell,A,b)$-curve-free, then $c_{A,U'}(B)=0$ for some $B\in\binom Ab$; the equality just proved would give $c_{A,U}(B)=0$, contradicting the curve-freeness of $U$.
\end{proof}

Fix a pair $(A, U)$ from \cref{def:curve-profile-family} and let $c=c_{A, U}$. We hope to prune down to a subset $A'\subset A$ of coordinates
such that $|A'|=b+hc(A')$ and, for every $i\in[n]$,
\begin{equation}
\label{eq:def-h}
 |A_i\cap A'|\le b+h c(A_i\cap A')\quad\text{where}\quad h:=\left\lceil\frac{2b}{\eta}\right\rceil.
\end{equation}
We formally define this low-rank family below, before motivating it and the choice of $h$ above by seeing how such profiles have negative potential.

\begin{definition}[Low-rank family]
\label{def:low-rank-family}
Fix $A\subset \F$ and an $(\ell,A,b)$-curve-free subspace $U\subset\F^{A}$ with
\begin{equation}
\label{eq:def-low-cA}
c(A)=\dim_{\F}U\le D \quad\text{and}\quad |A|=b+h\cdot c(A),
\end{equation}
where $c=c_{A,U}$. 
Let $\mc F^{\on{low}}_{A,U}$ consist of all pairs $(\mbf V,U)$ for
which there exist $A_1,\ldots,A_n\subset A$ with
\begin{enumerate}
\item $V_i=U\cap K_{A}(A_i)$ for all $i\in[n]$;
\item $|A_i|\le b+hc(A_i)$ for all $i\in[n]$;
\item $|\{i\in[n]:\alpha\in A_i\}|\ge(1-\delta)n$ for all
$\alpha\in A$.
\end{enumerate}
Define the \emph{low-rank family} $\mc F^{\on{low}}$ to be the union of $\mc F^{\on{low}}_{A,U}$ over such pairs $(A, U)$.
\end{definition}

Note that $c(A)=\dim_\F U$ we can assume for free by \cref{lem:c-translation}. Aiming to prove potential
\begin{equation}
\label{eq:dream}
\phi_{\mathbf V}(U, R) = R \dim U- \mb{E}_{i\in [n]}\on{codim}_{U}\left(U\cap K_A(A_i)\right) = Rc(A) - \mb{E} c(A_i)
\end{equation}
is negative for $R\le 1-\delta-2\eta$, 
we see that condition \cref{eq:def-low-cA} on $A$ and the analogous condition (2) for $A_i$ allows us to upper bound the potential by relating it to sizes $|A|$ and $|A_i|$, at which point (3) allows us to conclude the following proposition. This motivates the family $\mc{F}^{\on{low}}$. 

\begin{proposition}
\label{prop:low-rank-gap}
Let $(\mbf V,U)\in\mc F^{\on{low}}_{A,U}$ have agreement sets
$A_1,\ldots,A_n$, let $c:=c_{A,U}$, and let
$d:=\dim_{\F}U=c(A)$. Then
\begin{equation}
  \sum_{i=1}^n c(A_i)
  \ge
    \left(1-\delta-\frac{\eta}{2}\right)dn.
  \label{eq:terminal-margin}
\end{equation}
Moreover, $\phi_{\mbf V}(U,R)\le-3\eta d/2$ whenever $R\le 1-\delta-2\eta$.
\end{proposition}

\begin{proof}
We double count the number of $(i,\alpha)\in[n]\times A$ such that $\alpha\in A_i$ to obtain
\[ 
(1-\delta)|A|n\le \sum_{i=1}^n |A_i| \le bn+h\sum_{i=1}^n c(A_i).
\]
Since $|A|=b+hc(A)=b+hd$, $h\ge 2b/\eta$, and $\delta\le1\le c(A)=d$, we obtain that
\begin{equation}
\label{eq:why-h}
\sum_{i=1}^n c(A_i)\ge
  \left((1-\delta)c(A)-\frac{\delta b}{h}\right)n\ge
  \left(1-\delta-\frac{\eta}{2}\right)dn.
\end{equation}
Now, the last statement follows from \cref{eq:dream} and that
$R\le1-\delta-2\eta$.
\end{proof}

We remark that \cref{eq:why-h} explains the choice of $h$ in \cref{eq:def-h}: we prefer $h$ as small as possible, but in order to absorb the loss $\delta b/h$ by the constant additive rate slack $\eta$, we require $h=\Omega(b/\eta)$.

Now, we briefly explain the pruning process to produce a profile in \cref{def:low-rank-family} from a profile $(\mbf{V}, U)\in\mc{F}_{\ell, \delta, a, b}$ with rank at most $D$. Let $A$ be the coordinate set. We prune it iteratively from $A$ to $A'$: whenever some $A_i'=A_i\cap A'$ satisfies $|A_i'| > b+h c(A_i')$, we simply replace $A'$ by $A_i'$ and continue. Throughout the process we maintain $|A'|\ge b+hc(A')$, and termination implies condition (2) from \cref{def:low-rank-family}.
The baseline $b$ ensures that we do not prune to a trivial profile where $c(A')=0$. Indeed, since $|A'|\ge b$, then
$U\vert_{A'}\subset K_{A'}(A')$ whenever $c(A')=0$, contradicting curve-freeness. After we produce $A'$, we apply the curve translation in \cref{lem:c-translation} to replace $U\vert_{A'}$ by $U'$ to ensure $c(A')=\dim_\F U'$. 

Note that \cref{lem:c-restriction} implies that, in the process of pruning $A$ to $A'$ and $A_i$ to $A_i'=A_i\cap A'$, we think of $c$ as static $c_{A, U}$. \cref{lem:c-translation} implies that $c$ is also preserved when we apply the curve translation to assume $c(A')=\dim U'$.
We prove this formally below.
\begin{proposition}[Pruning]
\label{prop:pruning}
Suppose that an $\F$-additive code $C$ contains a profile
$(\mbf V,U)\in\mc F_{A,U}$ with $\dim_{\F}U\le D$, then
there exist $A'\subset A$ and $U'\subset\F^{A'}$ such that
$C$ contains a profile in
$\mc F^{\on{low}}_{A',U'}$.
\end{proposition}
\begin{proof}
Fix a witness matrix $M$ and write $c:=c_{A,U}$.
By our choices of parameters in \cref{eq:ggsw-param,eq:def-D,eq:def-h}, we observe that 
\[ |A|=a\ge b+hD \ge b+h\dim_\F U \ge b+hc(A).\]
Then, the following collection of subsets of $A$ is non-empty:
\[ \mc{S}:=\{S\subset A:|S|\ge b+hc(S)\}.\]
Fix $A'$ to be any inclusion-minimal element of $\mc{S}$. 
We claim that $|A'|=b+hc(A')$. If not, then $|A'|>b+hc(A')$. As $c(S)$ is non-decreasing in $S$ with respect to inclusion,
\[ |A'\setminus\{\alpha\}|\ge b+hc(A')\ge b+hc(A'\setminus\{\alpha\}),\]
which means $A'\setminus\{\alpha\}\in \mc{S}$, contradicting minimality of $A'$. This proves the claim. 

Now, set $U':=U\vert_{A'}$ and $A_i':=A'\cap A_i$. Then, $U'$ is $(\ell, A',b)$-curve-free as $U'\vert_B = U\vert_B\not\subset K_B(B)$ by curve-freeness of $U$. Then, $|A'|\ge b$ implies $c(A')>0$ and 
\begin{equation}
\label{eq:prune-ineq}
    c(A')\le \dim_\F U'\le\dim_\F U\le D.
\end{equation}

It suffices to check the bullet point conditions in \cref{def:low-rank-family}. If $|A_i'|>b+hc(A_i')$ for some $i\in [n]$, then $A_i'$ violates the minimality of $A'$. Moreover, 
\[M_{i\star}\in U\cap K_A(A_i)\implies (M\vert_{A'})_{i\star}\in K_{A'}(A'\cap A_i),\]
and, for every $\alpha\in A'$,
\[|\{i\in [n]:\alpha\in A_i'\}|=|\{i\in [n]:\alpha\in A_i\}|\ge (1-\delta)n.\]
Therefore, the matrix $M\vert_{A'}$ witnesses
containment of profile, which is almost in $\mc{F}^{\on{low}}_{A', U'}$ except \cref{eq:prune-ineq} has $c(A')\le \dim U'$ while \cref{eq:def-low-cA} demands equality. 
To do this, replace $U'$ by its image upon applying the linear map $\psi$ in
\cref{lem:c-translation}. We also apply $\psi$ rowwise to
$M\vert_{A'}$, obtaining a matrix $M'$.  By
\cref{lem:c-translation}, every value $c(A_i')$ and
curve-freeness is preserved, and
\[
  U'\cap K_{A'}(A')=\{0\}\implies 
  \dim_{\F}U'=c(A').
\]
Thus $|A'|=b+h\dim_{\F}U'$, and all the previously verified conditions required by $\mc{F}^{\on{low}}_{A', U'}$ remain true. By linearity of interpolation, each column of $M'$ is an $\F$-linear combination of columns of $M\vert_{A'}$, and hence remains in $C$. Therefore, $M'$ witnesses containment of a profile in
$\mc F^{\on{low}}_{A',U'}$.
\end{proof}

\subsection{Codimension-Sensitive Counting}
\label{subsec:rlc-count}
From this point onward, we drop the primes and denote a pruned,
translated profile simply by $(A,U,(A_i)_{i\in[n]})$.
Fix a pair $(A,U)$ in \cref{def:low-rank-family} and let $c=c_{A,U}$. Put $d:=\dim U$ and $a_d:=|A|=b+hd$. Let $Z\subset A$ be such that $\pi_{Z}:U\to \F^{Z}$ is bijective. Let $Z' = A\setminus Z$. Let $z=|Z|=\dim U$, and $z'=|Z'|=a_d-z$.

We wish to bound $|\mc F^{\on{low}}_{A,U}|$ by counting the number of spaces $V_i=U\cap K_A(A_i)$ directly. Formally, we define for integers $s_1, s_2,t \ge 0$
\begin{equation}
  \label{eq:def-Nst}
  N_{s_1, s_2, t}(A,U)
  :=
  \left|\left\{U\cap K_A(S_1 \cup S_2):
      S_1\in\binom{Z}{s_1}, S_2\in\binom{Z'}{s_2},c(S_1\cup S_2)=t\right\}\right|.
\end{equation}

To bound $N_{s_1,s_2,t}(A,U)$, the first idea is that each $V_i$ has a small defining point set we can use instead of $A_i$.  This is the curve decoding analogue of
the canonical constraint-counting principle in \cite{gg26}.

\begin{proposition}
\label{prop:canonical-config}
For any $U\subset\F^A$, $S_1\subset Z$, and $S_2\subset Z'$, let $c(S_1\cup S_2)=c_{A,U}(S_1\cup S_2)$. There exists $T\subset S_2$ such that
\begin{equation}
  U\cap K_A(S_1\cup T)=U\cap K_A(S_1\cup S_2)\quad\text{and}\quad |T|\le \ell+1+c(S_1\cup S_2)-|S_1|.
\end{equation}
\end{proposition}
\begin{proof}
If $|S_2|\le \ell+1-|S_1|$, then taking $T=S_2$ suffices since $c(S_1 \cup S_2)=0$ in this case.

Now, if $|S_1|\geq \ell+1$, then set $T_0:=\emptyset$ and fix $S_0$ to be the first $\ell+1$ elements of $S_1$ in some canonical ordering of $\F$. Else, fix $S_0 = S_1 \cup T_0$ where $T_0 \in \binom{S_2}{\ell+1-|S_1|}$. 
For any $u\in U$, there exists a unique $p_u \in \F[X]_{\le\ell}$ that agrees with $u$ on $S_0$. Moreover, $u\mapsto p_u$ is linear. 
Consider the collection of linear functionals on $U$ given by
\[ 
\{u\mapsto u(\beta)-p_u(\beta): \beta\in S_1\cup S_2\setminus S_0\}
\]
They have a common kernel $U\cap K_A(S_1\cup S_2)$ and so their span has $\F$-dimension exactly $c(S_1\cup S_2)$.

Choose any basis of such functionals indexed by $T_1\cup (S_1\setminus S_0)$, where $T_1\in \binom{S_2\setminus S_0}{c(S_1\cup S_2)-|S_1\setminus S_0|}$, so these $c(S_1\cup S_2)$ functionals have common kernel $U\cap K_A(S_1\cup S_2)$. Now, let $T=T_0\cup T_1$, so $U\cap K_A(S_1\cup S_2)=U\cap K_A(S_1\cup T)$ and $|T|\le \ell+1+c(S_1\cup S_2)-|S_1|$.
Here the functionals indexed by $S_1\setminus S_0$ are
independent because $S_1\subset Z$ and $U\vert_Z=\F^Z$, so they may
indeed be included in the chosen basis.
\end{proof}

\begin{remark}\label{rem:nice-stop}
This proposition is the diea of canonical profile counting and already gives useful alphabet size bounds. To count the number of profiles in $\mc{F}^{\on{low}}_{A, U}$ with $|A|=a$ and $c_i = c_{A, U}(A_i)$, we note that the space $U\cap K_A(A_i)$ is determined by $\ell+1+c_i$ points of $A$. Hence, for a fixed $(A, U)$ and $c_1, \dots, c_n$, there are at most $a^{(\ell+1)n+\sum_i c_i}$ such profiles. Combining this with \cref{lem:rlc-fixed-pair}, one can show that RLC has $(\ell, \delta, a, b)$ curve decoding at rate $R\le 1-\delta-2\eta$, provided alphabet size 
\[
q\ge
\left(\frac{\ell}{\eta}\right)^{O(\ell/\eta)}.
\]
For RRS, we note that since $c_i\le d\le D=O(\ell/\eta)$, one can show using \cref{thm:rrs-equiv} that RRS has $(\ell, \delta, a, b)$ curve decoding at rate $R\le 1-\delta-2\eta$, provided alphabet size 
\[
q\ge
\left(\frac{\ell}{\eta}\right)^{O(\ell/\eta^2)}.
\]

\end{remark}
To improve beyond the bounds above by log-factors in the exponent, we do further refinements. 
We first give our bound on $N_{s_1,s_2,t}(A,U)$ that is sensitive to codimension $t$ via double counting.

\begin{theorem}
\label{thm:codim-N-count}
Fix any $(A,U)$ from \cref{def:low-rank-family}. For any $s_1, s_2, t\in\N$ with $a_d:=|A|\ge s_1+s_2\ge\ell+2$, writing $z:=\dim U$ and $z':=a_d-z$, one has
\begin{equation}
\label{eq:Nst-bound}
  N_{s_1, s_2,t}(A,U)
  \le
  \binom{z}{s_1}\binom{z'}{t-\max(s_1-\ell-1, 0)}\frac{\binom{z'}{\max{(\ell+1-s_1, 0)}}}{\binom{s_2}{\max{(\ell+1-s_1, 0)}}}
\end{equation}
\end{theorem}
\begin{proof}
Continuing from \cref{prop:canonical-config}, observe that $N_{s_1, s_2, t}$ can be counted using the choices made for $S_0, S_1, T_0, T_1$. Observe that if $s_1 \ge \ell+1$, then there are $\binom{z}{s_1}$ choices for $S_1$, each of which naturally fixes $S_0$. $T_0 = \emptyset$, and now, there are at most $\binom{z'}{t+\ell+1-s_1}$ choices of $T_1$.

In the case $\ell+1>s_1$, we first choose $S_1$, which has $\binom{z}{s_1}$ choices. Now, observe that there are $\binom{z'}{\ell+1-s_1}$ choices for $T_0$ and $\binom{z'}{t}$ choices for $T_1$, but the same set of constraints is produced by at least $\binom{s_2}{\ell+1-s_1}$ choices for the set $T_0$. Thus, the number of possibilities is at most $\binom{z}{s_1}\binom{z'}t\frac{\binom{z'}{\ell+1-s_1}}{\binom{s_2}{\ell+1-s_1}}$.

This concludes the proof.

\end{proof}

\begin{lemma}
    For sequences $\mbf{s}=((s_1^{(1)}, s_2^{(1)}),\ldots,(s_1^{(n)}, s_2^{(n)}))$ and
$\mbf{t}=(t^{(1)},\ldots,t^{(n)})$, let
$\mc F^{\on{low}}_{A,U,\mbf s,\mbf t}$ denote the profiles in
$\mc F^{\on{low}}_{A,U}$ admitting agreement sets with
$|A_i\cap Z|=s_1^{(i)}$, $|A_i\cap Z'|=s_2^{(i)}$, and $c_{A,U}(A_i)=t^{(i)}$ for all $i\in[n]$. Then,
\begin{itemize}
    \item If $\dim U \geq \ell+1$: \[\left|
    \mc F^{\on{low}}_{A,U,\mbf s,\mbf t}
  \right|\leq 2^{\dim U\cdot n}a_d^{2(\ell+1)n + \sum_{i=1}^n \left(t^{(i)}-s_1^{(i)}\right)} \leq 2^{\dim U\cdot n}a_d^{2(\ell+1)n}a_d^{-(1-\delta)\dim U\cdot n+\sum_{i=1}^n t^{(i)}}\]
    \item If $(\ell+1)/2\leq \dim U < (\ell+1)$: \[\left|
    \mc F^{\on{low}}_{A,U,\mbf s,\mbf t}
  \right|\leq 2^{\dim U\cdot n}a_d^{(\ell+1)n + \sum_{i=1}^n t^{(i)}}\leq 2^{\dim U\cdot n}a_d^{3\dim U\cdot n}\]
  \item If $\dim U < (\ell+1)/2$: 
  \[\left|
    \mc F^{\on{low}}_{A,U,\mbf s,\mbf t}
  \right|\leq 2^{\dim U\cdot n}a_d^{\sum_{i=1}^n t^{(i)}}\min\left(6^{2dn} (100d/\eta^2)^{w}, (100d/\eta^2)^{n}\right)^{(\ell+1)}\]
  where $w:=\sum_{i=1}^n t^{(i)}-(1-\delta-\eta/2)nd$.
\end{itemize}
\end{lemma}
\begin{proof}
    
Note that $z' \leq a_d$.

For sequences $\mbf{s}=((s_1^{(1)}, s_2^{(1)}),\ldots,(s_1^{(n)}, s_2^{(n)}))$ and
$\mbf{t}=(t^{(1)},\ldots,t^{(n)})$, let
$\mc F^{\on{low}}_{A,U,\mbf s,\mbf t}$ denote the profiles in
$\mc F^{\on{low}}_{A,U}$ admitting agreement sets with
$|A_i\cap Z|=s_1^{(i)}$, $|A_i\cap Z'|=s_2^{(i)}$, and $c_{A,U}(A_i)=t^{(i)}$ for all $i\in[n]$. Then,
\begin{equation}
  \left|
    \mc F^{\on{low}}_{A,U,\mbf s,\mbf t}
  \right|
  \le
  \prod_{i=1}^n N_{s^{(i)}_1,s^{(i)}_2,t^{(i)}}(A,U).
\end{equation}

When $t^{(i)}=0$, the corresponding space is simply $U$ and
contributes one possibility. When $t^{(i)}>0$, necessarily
$s_1^{(i)}+s_2^{(i)}\ge\ell+2$, so \cref{thm:codim-N-count} applies.

\paragraph{Case 1: $z\geq \ell+1$:} Observe that in this case, we have,

\[N_{s_1, s_2,t}(A,U)
  \le\binom{z}{s_1}\binom{z'}{t-\max(s_1-\ell-1, 0)}\frac{\binom{z'}{\max{(\ell+1-s_1, 0)}}}{\binom{s_2}{\max{(\ell+1-s_1, 0)}}} \leq 2^z \cdot a_d^{2(\ell+1)+t-s_1}\]

\paragraph{Case 2: $\ell+1\ge z\ge (\ell+1)/2$}

\begin{align*}
    N_{s_1, s_2,t}(A,U)
  &\le \binom{z}{s_1}\binom{z'}{t-\max(s_1-\ell-1, 0)}\frac{\binom{z'}{\max{(\ell+1-s_1, 0)}}}{\binom{s_2}{\max{(\ell+1-s_1, 0)}}}\\
  &\le\binom{z}{s_1}\binom{z'}{t}\binom{z'}{\ell+1-s_1}/\binom{s_2}{\ell+1-s_1}\\
  &\leq 2^zz'^t (z')^{\ell+1}\\
  & \leq 2^{z}a_d^{t+\ell+1}
\end{align*}

\paragraph{Case 3: $z\leq (\ell+1)/2.$}

Thus, $s_2\ge s_1$ whenever $t>0$.

\begin{align*}
    N_{s_1, s_2,t}(A,U)
  &\le \binom{z}{s_1}\binom{z'}{t-\max(s_1-\ell-1, 0)}\frac{\binom{z'}{\max{(\ell+1-s_1, 0)}}}{\binom{s_2}{\max{(\ell+1-s_1, 0)}}}\\
  &\le\binom{z}{s_1}\binom{z'}{t}\binom{z'}{\ell+1-s_1}/\binom{s_2}{\ell+1-s_1}\\
  &\leq 2^zz'^t \binom{z'}{\ell+1-s_1}/\binom{s_2}{\ell+1-s_1}\\
  & \leq 2^{z}z'^{t}\left(\frac{ea_d}{s_2}\right)^{\ell+1}\\
  & \leq 2^z z'^t \left(\frac{6a_d}{s_1+s_2}\right)^{\ell+1}
\end{align*}

Let $s^{(i)}=s_1^{(i)}+s_2^{(i)}$. Now, observe that $\sum_{i=1}^n s^{(i)}\geq (1-\delta)n|A|$, and $s^{(i)}=|A_i|\leq b+ht^{(i)}$. Thus, let $r^{(i)} = \max\left(\frac{s^{(i)}-b}{h}, 0\right)\leq t^{(i)}$.

Observe that \[\sum_{i=1}^n r^{(i)}\geq \sum_{i=1}^n \frac{s^{(i)}}{h} - \frac{bn}{h}\geq  \frac{(1-\delta)n(b+dh) -bn}{h} = (1-\delta)nd - \frac{\delta bn}{h} \geq (1-\delta -\eta/2)dn\]

Thus, let $w:=\sum_{i=1}^n t^{(i)}-(1-\delta-\eta/2)dn$. Observe that $w\ge \sum_i(t^{(i)}-r^{(i)})$.

For $t^{(i)}>0$, the factor $6a_d/s^{(i)}$ is at most
$6a_d/(\ell+1)$, and
$1\le r^{(i)}+(t^{(i)}-r^{(i)})$. Thus the excess-factor argument
gives the bound below with exponent $w$. 
Thus, $z<(\ell+1)/2$, 
\begin{align*}
    \left|
    \mc F^{\on{low}}_{A,U,\mbf s,\mbf t}
  \right|& \leq 2^{\dim U\cdot n}a_d^{\sum_{i=1}^n t^{(i)}}\prod_{i: t^{(i)} \ne 0}(6a_d/s^{(i)})^{\ell+1}\\
  & \leq 2^{\dim U\cdot n}a_d^{\sum_i t^{(i)}}
       \prod_i(6a_d/s^{(i)})^{r^{(i)}(\ell+1)}
       (6a_d/(\ell+1))^{w(\ell+1)}\\
  & \leq 2^{\dim U\cdot n}a_d^{\sum_{i=1}^n t^{(i)}}\left[\prod_{i}(6a_d/s^{(i)})^{s^{(i)}(\ell+1)/h}\right] (6a_d/(\ell+1))^{w(\ell+1)}\\
  & \leq 2^{\dim U\cdot n}a_d^{\sum_{i=1}^n t^{(i)}}6^{2dn(\ell+1)} (100d/\eta^2)^{w(\ell+1)}
\end{align*}

But additionally,

\begin{align*}
    \left|
    \mc F^{\on{low}}_{A,U,\mbf s,\mbf t}
  \right|& 
  \leq 2^{\dim U\cdot n}a_d^{\sum_{i=1}^n t^{(i)}}\prod_{i: t^{(i)} \ne 0}(6a_d/s^{(i)})^{\ell+1}\\
  & \leq 2^{\dim U\cdot n}\cdot a_d^{\sum_{i=1}^n t^{(i)}}\cdot (100d/\eta^2)^{n(\ell+1)}
\end{align*}

These inequalities give the requisite intermediate relations directly. In the first, we additionally use that $\sum_{i=1}^n s_1^{(i)}\geq (\dim U)(1-\delta)n$. The second and third simply use that $t^{(i)}\leq \dim U$.

\end{proof}

\begin{definition}
For $d\in[D]$, set
\begin{equation}
  a_d:=b+hd,
  \qquad
  T_d:=\left\lceil\left(1-\delta-\frac{\eta}{2}\right)dn\right\rceil.
\label{eq:def-low-strata}
\end{equation}
For $A\in\binom{\F}{a_d}$ and $w\in\Z_{\ge0}$, let
$\mc F_{A,d,w}^{\on{low}}$ be the union, over all $U$ satisfying
\cref{def:low-rank-family} with $\dim U=d$, of the profiles
$(\mbf V,U)\in\mc F_{A,U}^{\on{low}}$ admitting agreement sets
$A_1,\ldots,A_n$ such that
\begin{equation}
  \sum_{i=1}^n c_{A,U}(A_i)=T_d+w.
\end{equation}
Finally, let
\begin{equation}
  \mc F_{d,w}^{\on{low}}
  :=\bigcup_{A\in\binom{\F}{a_d}}\mc F_{A,d,w}^{\on{low}},
  \qquad
  \mc F_d^{\on{low}}
  :=\bigcup_{w\ge0}\mc F_{d,w}^{\on{low}}.
\end{equation}
\end{definition}

\begin{lemma}
\label{lem:compressed-profile-cover}
If $C$ is not $(\ell,\delta,a,b)$-curve-decodable, then $C$
contains either a profile in $\mc F_S^{\on{high}}$ for some
$S\in\binom{\F}{D}$ or a profile in $\mc F_{d,w}^{\on{low}}$ for
some $d<D$ and $w\in\Z_{\ge0}$.
\end{lemma}
\begin{proof}
By \cref{thm:ggsw-family}, $C$ contains a profile in $\mc F$.
If its row-space dimension is at least $D$, apply
\cref{lem:high-rank}. Otherwise, apply \cref{prop:pruning} to obtain
a contained profile $(\mbf V,U)\in\mc F_{A,U}^{\on{low}}$. Put
$d:=\dim U$ and
$w:=\sum_i c_{A,U}(A_i)-T_d$. By
\cref{prop:low-rank-gap}, $w\ge0$, so this profile belongs to
$\mc F_{d,w}^{\on{low}}$.
\end{proof}

\begin{restatable}[Profile counts for curve-decoding]{lemma}{curvedeccount}\label{lem:low-family-count}
    Let $d \in \mathbb N$:
    \begin{itemize}
        \item if $d\ge \ell+1$: \[|\mathcal F_{d, w}^{\on{low}}|\leq q^{da_d+a_d}\cdot 2^{dn}\cdot \left((a_d)^3\cdot (a_d)^{2(\ell+2)}\right)^n(a_d)^w\]
        \item if $(\ell+1)/2\le d < (\ell+1):$ \[|\mathcal F_{d}^{\on{low}}|\leq q^{da_d+a_d}\cdot \left((a_d)^{3}\cdot (2a_d)^{3d}\right)^n\]
        \item if $d<(\ell+1)/2$:
        \[|\mathcal F_{d}^{\on{low}}|\leq q^{da_d+a_d}\cdot (2a_d)^{4dn}\cdot (100d/\eta^2)^{(\ell+1)n}\]
        \item Additionally, if $d\leq\min\{\log(1/\eta),(\ell+1)/2\}$,
        \[|\mathcal F_{d, w}^{\on{low}}|\leq q^{da_d+a_d}\cdot \left((a_d)^{3}\cdot (2a_d)^{d}\cdot (6)^{2d(\ell+1)}\right)^n(100\log(1/\eta)/\eta^2)^{(w+1)(\ell+1)}\]
    \end{itemize}
\end{restatable}
\begin{proof}
    The $q^{d a_d}$ lets us quantify over the choices of the matrix $U$ in each of the cases. Additionally, $(a_d)^{3n}$ lets us quantify over the sequences $\mbf s_1, \mbf s_2, \mbf t$. 

    Finally, in the first case we know that $\sum t^{(i)}\leq (1-\delta)dn+1+w$. In the last case, the excess in the preceding lemma is
    $T_d+w-(1-\delta-\eta/2)dn<w+1$ by the definition of $T_d$.
    Multiplying these choices by the corresponding three bounds in the preceding lemma, and using $\sum_i t^{(i)}\le dn$, gives the four displayed estimates.
\end{proof}

\section{Alphabet-Size Bounds and Applications}
\label{sec:alphabet}

We now combine the compressed curve-profile families from
\cref{sec:profile-counting} with the three fixed-pair estimates. For
RLC and RLDPC, codimension-sensitive counting is paired with a
codimension-sensitive containment probability. For RRS, the
probability is instead sensitive to the negative potential gap, and
the information-set count in \cref{thm:codim-N-count} supplies the
corresponding saving. No new profile count is needed in this section.

Throughout, let $a,b,D$, and $h$ be as in
\cref{eq:ggsw-param,eq:def-D,eq:def-h}, and suppose
$R\le1-\delta-2\eta$. This means the code has rate at least $\eta$ below the
threshold of the curve decoding profile family. For a sufficiently large absolute constant
$C_0$, the two alphabet bounds used below are
\begin{equation}
  q\ge
  \exp\left(
    C_0\left(\log(1/\eta)+\ell\right)/\eta
  \right)
  \label{eq:q-linear-final}
\end{equation}
for RLC and RLDPC, and
\begin{equation}
  q\ge
  n\exp\left(
    C_0\left(
      \ell/\eta^2
      +\ell\log(\ell+1)/\eta
    \right)
  \right)
  \label{eq:q-rrs-final}
\end{equation}
for RRS. The goal of this section is
to make the resulting union bounds quantitative at the two alphabet
sizes above.
\subsection{Random Linear Codes}

\begin{restatable}[Curve Decoding for Random Linear Codes]{theorem}{rlccurvedec}
\label{thm:rlc-curve-final}
Let $\ell\geq1$, $0<\eta\leq1/4$, $\delta\in(0,1)$, and
$R\leq1-\delta-2\eta$.  Set
\[
 a:=\ceil{\frac{100\ell^2}{\eta^3}},
 \qquad
 b:=\ceil{\frac{4\ell}{\eta}+\ell}.
\]
Then, there exists an absolute constant $c_0$ such that  if
\[
 q\geq\exp\!\left(
   c_0\!\left(
     \frac{\log (1/\eta)}{\eta}
     +\frac{\ell}{\eta}
   \right)
 \right),
\]
then, for all sufficiently large $n$, with probability at least $1- q^{-\eta n/8}$, a random linear code $C\subset\F^n$ of rate $R$ is
$(\ell,\delta,a,b)$-curve-decodable.
\end{restatable}

\begin{proof}
Let $\mc F:=\mc F_{\ell,\delta,a,b}$. By
\cref{lem:high-rank,lem:compressed-profile-cover}, failure of curve
decodability produces a contained profile in either
$\mc F^{\on{high}}$ or one of the low-rank families
$\mc F_{d,w}^{\on{low}}$.

For the high-rank family,
\cref{lem:high-rank-family-bound,lem:rlc-fixed-pair} give
\[
  \mb{P}\left(C\text{ contains }\mc F^{\on{high}}\right)
  \le
  \binom{q}{D}2^{Dn}q^{-\eta Dn}
  \leq \left(q^{D(1-\eta n/2)}\cdot 2^{Dn(1-\eta\log q/2)}\right) \leq q^{-\eta n /4}.
\]
The final inequality follows as long as $n\ge 2/\eta$ and $\log q \ge 4/\eta$.

For a low-rank profile, let $d:=\dim_{\F}U$ and
$T:=\sum_{i=1}^n c(A_i)$. By the definition of $w$,
\begin{equation}
  G:=T-Rdn
  \ge\frac{3\eta}{2}dn+w.
  \label{eq:low-profile-gap}
\end{equation}

\begin{remark}\label{rmk:RLDPC-stuff}
    Observe that this gives that for any profile, probability of containment is at most $q^{-3\eta dn/2-w}$. Throughout, we will only use $q^{-\eta dn-w}$, and thus observe that the same argument directly applies to RLDPCs in the next subsection.
\end{remark}

We now recount the profile counts for all $d, w$:

\curvedeccount*

Additionally, $n\phi_{\mbf V}(U,R)=-G$, the fixed-pair estimate in
\cref{lem:rlc-fixed-pair} is $q^{-G}$. Multiplying this estimate by
the four corresponding bounds in \cref{lem:low-family-count}, we get:

\paragraph{For $d \ge \ell+1$:}

\begin{align*}
    \sum_{d\ge (\ell+1)}\sum_{w\ge 0}\mb{P}\left(C \text{ contains }\mathcal F_{d, w}^{\on{low}}\right) &\leq q^{2da_d-\eta n/4}\cdot 2^{dn(1-\eta \log q/4)}\cdot \left(a_d^{7\ell+7}q^{-\eta d/4}\right)^n\cdot q^{-\eta dn/4}\sum_{w\ge 0} (a_d/q)^w\\
    & \leq \sum_{d\ge (\ell+1)} 1 \cdot 1 \cdot (a_d^7/q^{\eta/4})^{dn} \cdot q^{-\eta dn/4} \cdot \sum_{w\ge 0}2^{-w}\\
    & \leq \sum_{d\ge (\ell+1)} 2q^{-\eta n d/4}\\
    & \leq 4q^{-\eta n/4}
\end{align*}

Each of the first terms is at most $1$ since $n\ge 8da_d/\eta$, $\log q \ge 4/\eta$, $q\ge a_d^{28/\eta}$ and $q\ge 2a_d$. The final geometric sum over $d$ also follows since $n\ge 4/\eta$.

\paragraph{For $(\ell+1)/2\le d<\ell+1$:}
\begin{align*}
    \sum_{(\ell+1)/2\le d<\ell+1} \mb{P}\left(C \text{ contains }\mathcal F_{d}^{\on{low}}\right) & \leq \sum_{(\ell+1)/2\le d<\ell+1} q^{2da_d - \eta n/4}\cdot \left((2a_d)^{6}q^{-\eta/2}\right)^{dn}\cdot q^{-\eta n/4}\\
    & \leq \sum_{d\ge (\ell+1)/2} 1\cdot (1/2)^{dn} \cdot q^{-\eta n/4}\\
    & \leq 2q^{-\eta n/4}
\end{align*}
The first relation follows since $n\ge 8da/\eta$, then, $q\ge (4a)^{12/\eta}$, so the second geometric sum is at most $2$.

\paragraph{For $\log(1/\eta)\le d<(\ell+1)/2$:}
\begin{align*}
    \mb{P}\left(C \text{ contains }\mathcal F_{d}^{\on{low}}\right) & \leq q^{2da_d - \eta n/4}\cdot \left((2a_d)^{4}q^{-\eta/4}\right)^{dn}\cdot \left((100d/\eta^2)^{(\ell+1)/d}q^{-\eta/4}\right)^{dn}\cdot q^{-\eta dn/4}\\
    & \leq 1\cdot 1\cdot \left((40e^2)^{(\ell+1)}q^{-\eta/4}\right)^{dn}\cdot q^{-\eta dn/4}\\
    & \leq  1\cdot 1 \cdot 1\cdot q^{-\eta nd/4}
\end{align*}
The first relation follows since $n\ge 8da/\eta$, then $q\ge 2\cdot (2a)^{16/\eta}$, and then $q\ge (40e^2)^{4(\ell+1)/\eta}$.

Now, summing over $d$ we get the bound of $2q^{-\eta n/4}$.

\paragraph{For $d \leq \min(\log(1/\eta), (\ell+1)/2)$:}

\begin{align*}
    \sum_{w\ge 0}\mb{P}\left(C \text{ contains }\mathcal F_{d, w}^{\on{low}}\right) &\leq q^{2da_d-\eta n/4}\cdot \left((2a_d)^{4}q^{-\eta/4}\right)^{dn}\cdot \left(36^{\ell+1}q^{-\eta/4}\right)^{dn}\notag\\
    &\qquad\cdot q^{-\eta dn/4}\left(\frac{100\log(1/\eta)}{\eta^2}\right)^{\ell+1}\sum_{w\ge 0} \left(\frac{\left(100\log(1/\eta)/\eta^2\right)^{\ell+1}}q\right)^w\\
    & \leq 1 \cdot 1 \cdot 1 \cdot q^{-\eta dn/4} \cdot \sum_{w\ge 0}2^{-w}\\
    & \leq 2q^{-\eta n d/4}
\end{align*}

The first relation follows since $q^{2da_d}\left(100\log(1/\eta)/\eta^2\right)^{\ell+1}\le q^{\eta n/4}$ for sufficiently large $n$, then $q\ge 2\cdot (2a)^{16/\eta}$, then $q\ge 10^{20(\ell+1)/\eta}$, and finally the geometric term follows since $q\ge 2\cdot \left(100\log (1/\eta)/\eta^2\right)^{\ell+1}$.

Now, summing over $d$ we get the bound of $4q^{-\eta n/4}$.

Thus, \[\mb{P}\left(C\text{ contains } \mathcal F^{\on{low}}\right)\leq 4q^{-\eta n/4}+2q^{-\eta n/4}+2q^{-\eta n/4}+4q^{-\eta n/4}\leq 12q^{-\eta n/4}.\]

Thus, $\mb{P}\left(C \text{ is } (\ell, \delta, a, b)\text{-curve decodable}\right)\geq 1- 13q^{-\eta n/4}\geq 1-q^{-\eta n/8}$ whenever $n\ge 32/\eta$ but this is $1-o(1)$ as desired.

\end{proof}

To obtain the large-output consequences without adding an
$\exp(O(1/\eta^2))$ requirement to the RLC alphabet, we use the
sharpened random-linear list-decoding theorem in
\cref{thm:rlc-ld}, quoted from \cite[Corollary~3.11]{gg26}.

\begin{corollary}[PG, CA, and MCA for random linear codes]
\label{cor:rlc-proximity-package}
Under the hypotheses of \cref{thm:rlc-curve-final}, suppose that the
absolute constant in \cref{eq:q-linear-final} is sufficiently
large.
Then, with high probability, the following hold.
\begin{enumerate}
\item For every integer $t\ge1$, $C$ is curve-decodable with parameters
\[
  \left(
    \ell,\delta,
    (t-1)\ceil{(1-R)/\eta}+a,
    t
  \right).
\]
\item For every $m>1$, $C$ has correlated agreement and proximity gap
with parameters
\[
  \left(
    \ell,\delta,
    \frac{\ell m\ceil{(1-R)/\eta}+a}{q},
    \frac{1}{m}
  \right).
\]
\item The code $C$ has mutual correlated agreement with parameters
\[
  \left(
    \ell,\delta,
    \frac{\ell n\ceil{(1-R)/\eta}+a}{q}
  \right).
\]
\end{enumerate}
\end{corollary}

\begin{proof}
By \cref{thm:rlc-ld}, whose alphabet requirement is absorbed by
\cref{eq:q-linear-final}, with high probability $C$ is
$\left(1-R-3\eta/2,\ceil{(1-R)/\eta}\right)$ list-decodable. On the other hand,
the definition of $b$ gives
\[
  \delta\left(1+\frac{\ell}{b-\ell}\right)
  \le\delta+\frac{\eta}{4}
  \le1-R-\frac{7\eta}{4}
  <1-R-\frac{3\eta}{2}.
\]
Thus \cref{thm:rlc-curve-final,thm: main-ca-mca-gg} gives all three
claims.
\end{proof}

\subsection{Random LDPC Codes}
The random model used is Gallager's ensemble of RLDPCs. See \cref{def:rldpc} for more details.
\begin{restatable}[Curve Decoding for Random LDPC Codes]{theorem}{ldpccurvedec}
\label{thm:rldpc-curve-final}
Let $C$ be sampled from the Gallager ensemble of \cref{def:rldpc}.
Let $\ell\geq1$, $0<\eta\leq1/4$, $\delta\in(0,1)$, and
$R\leq1-\delta-2\eta$.  Set
\[
 a:=\ceil{\frac{100\ell^2}{\eta^3}},
 \qquad
 b:=\ceil{\frac{4\ell}{\eta}+\ell}.
\]
Then, there exists an absolute constant $c_0$ such that if
\begin{equation}
 q\geq\exp\!\left(
   c_0\!\left(
     \frac{\log (1/\eta)}{\eta}
     +\frac{\ell}{\eta}
   \right)
 \right),
 \label{eq:rldpc-rate-independent}
\end{equation}
then, for all odd $s\ge s_0(\ell,\eta,\delta,R,q)$ with
$(1-R)s\in\N$ and all sufficiently large multiples $n$ of $s$, with high probability, a random $s$-sparse LDPC code $C\subset\F^n$ of rate $R$ is $(\ell,\delta,a,b)$-curve-decodable.
\end{restatable}

\begin{proof}
Condition on the event $\mathsf{Good}$ from
\cref{prop:rldpc-fixed-pair}, and apply that proposition with
auxiliary error $\zeta:=\eta/2$. For a low-rank profile,
$n\phi_{\mbf V}(U,R)=-G\le-3\eta dn/2-w$, so its conditional
containment probability is at most $q^{-\eta dn-w}$. For a high-rank
profile, the same estimate and \cref{lem:high-rank-family-bound} give
conditional probability at most
$q^{-\eta Dn+\eta n/2}\le q^{-\eta Dn/2}$. Thus the four low-rank
union bounds above apply verbatim, and the high-rank union bound applies
with $\eta$ replaced by $\eta/2$. Finally,
$\mb{P}\left(\mathsf{Good}^{\mathsf c}\right)=o(1)$, which proves the result.
\end{proof}

Thus the RLDPC curve-decoding result is formally the RLC
argument with two existing inputs substituted: the RLDPC fixed-pair
estimate from \cref{prop:rldpc-fixed-pair} and the same compressed
profile count from \cref{lem:low-family-count}.  Together with the
refined list-decoding count recorded in \cref{thm:rldpc-ld}, this uses
the ideas of \cite{gg26,ggsw} but is entirely formal.

\begin{corollary}[PG, CA, and MCA for random LDPC codes]
\label{cor:rldpc-proximity-package}
Under the hypotheses of \cref{thm:rldpc-curve-final}, with high
probability the following hold.
\begin{enumerate}
\item For every integer $t\ge1$, $C$ is curve-decodable with parameters
\[
  \left(
    \ell,\delta,
    (t-1)\ceil{2(1-R)/\eta}+a,
    t
  \right).
\]
\item For every $m>1$, $C$ has correlated agreement and proximity gap
with parameters
\[
  \left(
    \ell,\delta,
    \frac{\ell m\ceil{2(1-R)/\eta}+a}{q},
    \frac{1}{m}
  \right).
\]
\item The code $C$ has mutual correlated agreement with parameters
\[
  \left(
    \ell,\delta,
    \frac{\ell n\ceil{2(1-R)/\eta}+a}{q}
  \right).
\]
\end{enumerate}
\end{corollary}

\begin{proof}
By \cref{thm:rldpc-ld}, whose alphabet requirement is absorbed by
\cref{eq:rldpc-rate-independent}, with high probability $C$ is
$\left(1-R-3\eta/2,\ceil{2(1-R)/\eta}\right)$ list-decodable. The estimate
\[
  \delta\left(1+\frac{\ell}{b-\ell}\right)
  \le1-R-\frac{7\eta}{4}
  <1-R-\frac{3\eta}{2}
\]
checks the list decoding hypothesis of
\cref{thm: main-ca-mca-gg}. Combine that theorem with
\cref{thm:rldpc-curve-final}.
\end{proof}

\subsection{Random Reed--Solomon Codes}

\begin{restatable}[Curve Decoding for Random Reed--Solomon Codes]{theorem}{rrscurvedec}
\label{thm:rrs-curve-dec}
Let $\ell\geq1$, $0<\eta\leq1/4$, $\delta\in(0,1)$, and
$R\leq1-\delta-2\eta$.  Set
\[
 a:=\ceil{\frac{100\ell^2}{\eta^3}},
 \qquad
 b:=\ceil{\frac{4\ell}{\eta}+\ell}.
\]
Let $C\subset\F^n$ be a random Reed--Solomon code of rate $R$,
with uniformly random distinct evaluation points. Then, there exists an absolute constant $c_0$ such that  if
\begin{equation}
 q\geq n\exp\left(
 c_0\left(\ell/\eta^2+\ell\log(\ell+1)/\eta\right)
 \right),
 \label{eq:rrs-rate-independent}
\end{equation}
then, for all sufficiently large $n$, with high probability $C$ is
$(\ell,\delta,a,b)$-curve-decodable.
\end{restatable}

\begin{proof}

\newcommand{\qrs}{q_{\mathsf{RS}}}

We let \begin{equation*}
  q_{\mathsf{RS}}
  :=
  n+\frac{2eDn}{\eta}\exp\!\left(
    c_1\left(
      \frac{\ell+1}{\eta^2}
      +\frac{(\ell+1)\log(\ell+1)}{\eta}
    \right)
  \right).
\end{equation*} The field-size hypothesis in \cref{eq:q-rrs-final} dominates this auxiliary threshold after increasing its absolute constant.
We now let $q\ge \qrs$, $\kappa_0 = c_1\left(
      \frac{\ell+1}{\eta^2}
      +\frac{(\ell+1)\log(\ell+1)}{\eta}
    \right)$ and let $\alpha = q/\qrs$. For a low-rank profile
$(\mbf V,U)\in\mc F_{d,w}^{\on{low}}$, maximal quotienting followed
by \cref{thm:low-lms} gives the bound below; for a high-rank profile,
the same bound holds with $w=0$.  Indeed, in the low-rank case the
normalized root gap is at least
$3\eta/2+w/(dn)$ by \cref{eq:low-profile-gap}, while in the high-rank
case it is at least $\eta$. Since $d\le D\le\eta n/2$ for sufficiently
large $n$, the first bound in \cref{eq:relative-bound} may be weakened
uniformly as follows:
\begin{align*}
    \mb{P}\left(C \text{ contains }(\mbf V, U)\right)& \leq D
  \left(
    \frac{2\mathrm eD(k-1)}{\eta(q-n)}
  \right)^{\eta n/2+w/(2D)}\\
  & \leq \left(
    \frac{2\mathrm eD(k-1)}{\eta(\qrs-n)}
  \right)^{w/(2D)}\cdot D\left(
    \frac{2\mathrm eD(k-1)}{\eta(\qrs-n)}
  \right)^{\eta n/2}\cdot \alpha^{-\eta n/2}\\
  & \leq 2^{-\kappa_0w/(2D)}\cdot D\cdot 2^{-\kappa_0(\eta n/2)}\cdot \alpha^{-\eta n/2}\\
  & \leq 2^{-\kappa_0w/(2D)}\cdot 2^{-\kappa_0 \eta n/4}\cdot
  q^{-\kappa_0\eta n/(8\log \qrs)}.
\end{align*}

Here we used $k-1\le n$ and the definition of $\qrs$ to bound
$2\mathrm eD(k-1)/(\eta(\qrs-n))\le e^{-\kappa_0}\le2^{-\kappa_0}$;
the last line also uses $\alpha=q/\qrs$, and holds for all sufficiently
large $n$.

Let $\kappa_0 \eta /4 = c_2\left(
      \frac{\ell+1}{\eta}
      +(\ell+1)\log(\ell+1)\right) = \kappa_1$ and $m = \kappa_0 \eta n/(8\log \qrs)$. Thus, 

\[\mb{P}\left(C \text{ contains }(\mbf V, U)\right) \leq 2^{-\kappa_1 w/(\eta D)}\cdot 2^{-\kappa_1 n}\cdot q^{-m}.\]

As before, a curve-decoding failure produces either a high-rank or a
low-rank contained profile. 

We bound the high-rank family separately and split the low-rank
family according to $d$ and $w$. Now, 
\begin{align*}
    |\mathcal F^{\on{high}}|\cdot 2^{-\kappa_1 n}\cdot q^{-m} \leq q^{D-m}\cdot 2^{(D-\kappa_1)n}\leq q^{-m/2}
\end{align*}
as long as $n/\log n >\mathsf{poly}(\ell, 1/\eta)$.

It remains to now bound 
\begin{align*}
    \mb{P}\left(C \text{ contains }\mathcal F^{\on{low}}\right) & \leq \sum_{\ell+1\le d<D}\sum_{w\ge 0} |\mathcal F_{d,w}^{\on{low}}|\cdot 2^{-\kappa_1 w/(\eta D)}\cdot 2^{-\kappa_1 n}\cdot q^{-m} + \sum_{d=1}^{\ell} \mb{P}\left(C \text{ contains }\mathcal F_d^{\on{low}}\right)\\
    & \leq \sum_{D>d\ge (\ell+1)} q^{2da_d-m}\cdot \left(2^{d+7(\ell+1)\log a_d - \kappa_1}\right)^n\cdot \sum_{w\ge 0}(a_d\cdot 2^{-\kappa_1/(\eta D)})^w \\&~~~ + \sum_{d=1}^{\ell} \mb{P}\left(C \text{ contains }\mathcal F_d^{\on{low}}\right)\\
    & \leq (D-\ell)\cdot q^{-m/2}\cdot 1\cdot 1\cdot \sum_{w\ge 0}2^{-w} + \sum_{d=1}^{\ell} \mb{P}\left(C \text{ contains }\mathcal F_d^{\on{low}}\right)\\
    & \leq 2(D-\ell)q^{-m/2} + \sum_{d=1}^{\ell} \mb{P}\left(C \text{ contains }\mathcal F_d^{\on{low}}\right)
\end{align*}

Now, 
\begin{align*}
    \sum_{d=1}^{\ell} \mb{P}\left(C \text{ contains }\mathcal F_d^{\on{low}}\right) &\leq \left(\ell \max_{d\le \ell} |\mathcal F_d^{\on{low}}|\right)\cdot 2^{-\kappa_1 n}\cdot q^{-m}\\
    & \leq \ell\cdot q^{2Da-m}\cdot \left((2a)^{6(\ell+1)} (100\ell/\eta^2)^{(\ell+1)}2^{-\kappa_1}\right)^n\\
    & \leq \ell\cdot q^{-m/2}\cdot \left((2a)^{6(\ell+1)} (a)^{(\ell+1)}2^{-\kappa_1}\right)^n\\
    & \leq \ell \cdot q^{-m/2}\cdot \left(2^{7(\ell+1)\log 2a -\kappa_1}\right)^n\\
    & \leq \ell q^{-m/2}
\end{align*}

The final conclusion follows in both cases as long as $m\ge 4Da = \mathsf{poly}(\ell, 1/\eta)$, and $\kappa_1 \ge D+7(\ell+1)\log (a)$ and
$\kappa_1/(\eta D)\ge\log(2a)$, which is true for a sufficiently large constant choice of $c_2$. As a reminder, $\kappa_1 = c_2\left(
      \frac{\ell+1}{\eta}
      +(\ell+1)\log(\ell+1)\right)$.

Thus, \[\mb{P}\left(C \text{ is not }(\ell, \delta, a, b)\text{-curve decodable}\right)\leq (2D-\ell+1)q^{-m/2}\leq q^{-m/4}\]

when $n\ge \mathsf{poly}(\ell, 1/\eta)$ for a sufficiently large polynomial.
\end{proof}

\begin{corollary}[PG, CA, and MCA for random Reed--Solomon codes]
\label{cor:rrs-proximity-package}
Under the hypotheses of \cref{thm:rrs-curve-dec}, with high probability
the following hold.
\begin{enumerate}
\item For every integer $t\ge1$, $C$ is curve-decodable with parameters
\[
  \left(
    \ell,\delta,
    (t-1)\ceil{(1-R)/\eta}+a,
    t
  \right).
\]
\item For every $m>1$, $C$ has correlated agreement and proximity gap
with parameters
\[
  \left(
    \ell,\delta,
    \frac{\ell m\ceil{(1-R)/\eta}+a}{q},
    \frac{1}{m}
  \right).
\]
\item The code $C$ has mutual correlated agreement with parameters
\[
  \left(
    \ell,\delta,
    \frac{\ell n\ceil{(1-R)/\eta}+a}{q}
  \right).
\]
\end{enumerate}
\end{corollary}

\begin{proof}
The field-size hypothesis in \cref{eq:rrs-rate-independent} absorbs
the hypothesis of \cref{thm:rrs-ld}. Thus, with high probability,
$C$ is $\left(1-R-3\eta/2,\ceil{(1-R)/\eta}\right)$ list-decodable. Since
\[
  \delta\left(1+\frac{\ell}{b-\ell}\right)
  \le1-R-\frac{7\eta}{4}
  <1-R-\frac{3\eta}{2},
\]
the claims follow from
\cref{thm:rrs-curve-dec,thm: main-ca-mca-gg}.
\end{proof}

\begin{remark}
The formal statements use $R\le1-\delta-2\eta$. The headline
parameterization $\delta=1-R-\eta$ follows by replacing $\eta$ in the
formal statements by $\eta/2$, which changes only the absolute
constants in the alphabet bounds.
\end{remark}

\addcontentsline{toc}{section}{References}
\bibliographystyle{alpha}
\bibliography{ref, coding-theory}
\end{document}